\documentclass[11pt,reqno]{amsart}
\usepackage[top=2.54cm, bottom=2.54cm, left=2.5cm, right=2.5cm]{geometry}
\usepackage{dsfont, amssymb,amsmath,amscd,latexsym, amsthm, amsxtra,amsfonts}
\usepackage[all]{xy}
\usepackage[active]{srcltx}
\usepackage{chngpage}
\usepackage{array}
\usepackage{tabularx}
\usepackage{datetime}
\usepackage{hyperref}
\usepackage{url}
\usepackage[round]{natbib}
\usepackage{bbm}
\usepackage{enumerate}
\usepackage{mathrsfs}
\usepackage{graphicx}
\usepackage{comment}
\usepackage{mathtools}
\usepackage[hang,flushmargin]{footmisc}

\usepackage{tikz}
\usetikzlibrary{calc,arrows}
\usepackage{verbatim}
\usepackage{graphicx}
\usepackage{subfigure}

\newtheorem{theorem}{Theorem}[section]

\newtheorem{proposition}[theorem]{Proposition}
\newtheorem{lemma}[theorem]{Lemma}

\theoremstyle{definition}
\newtheorem{definition}[theorem]{Definition}

\renewcommand{\theequation}{\arabic{section}.\arabic{equation}}

\theoremstyle{definition}

\theoremstyle{definition}
\newtheorem{remark}{Remark}
\theoremstyle{definition}

\DeclareMathOperator*{\esssup}{ess\,sup}

\newcommand{\rd}{\mathrm{d}}

\renewcommand{\epsilon}{\varepsilon}

\renewcommand{\cite}{\citet*}

\usepackage{graphics}
\usepackage[round]{natbib}
\setcitestyle{authoryear,round}

\begin{document}
	\makeatletter
	\def\@setauthors{%
		\begingroup
		\def\thanks{\protect\thanks@warning}%
		\trivlist \centering\footnotesize \@topsep30\p@\relax
		\advance\@topsep by -\baselineskip
		\item\relax
		\author@andify\authors
		\def\\{\protect\linebreak}%
		{\authors}%
		\ifx\@empty\contribs \else ,\penalty-3 \space \@setcontribs
		\@closetoccontribs \fi
		\endtrivlist
		\endgroup } \makeatother
	\baselineskip 18pt
	\title[{{\tiny Many-insurer robust games}}]
	{{\tiny  Many-insurer robust games of reinsurance and investment under model uncertainty in incomplete markets}} \vskip 10pt\noindent
		\author[{\tiny  Guohui Guan, Zongxia Liang, Yi Xia}]
		{\tiny {\tiny  Guohui Guan$^{a,b,\dag}$, Zongxia Liang$^{c,\ddag}$, Yi Xia$^{c,*}$}
				\vskip 10pt\noindent
				{\tiny ${}^a$Center for Applied Statistics, Renmin University of China, Beijing, 100872, China
						\vskip 10pt\noindent\tiny ${}^b$School of Statistics, Renmin University of China, Beijing 100872, China
						\vskip 10pt\noindent\tiny ${}^c$Department of Mathematical Sciences, Tsinghua
						University, Beijing 100084, China}
				\footnote{ 
						$^{\dag}$  e-mail: guangh@ruc.edu.cn\\
						$^{\ddag}$  e-mail:  liangzongxia@mail.tsinghua.edu.cn\\
						$^*$  Corresponding author, e-mail:  xia-y20@mails.tsinghua.edu.cn}}
	\numberwithin{equation}{section}
	\maketitle
		\begin{abstract}
			This paper studies the robust reinsurance and investment games for competitive insurers. Model uncertainty is characterized by a class of equivalent probability measures. Each insurer is concerned with relative performance under the worst-case scenario. Insurers' surplus processes are approximated by drifted Brownian motion with common and idiosyncratic insurance risks. The insurers can purchase proportional reinsurance to divide the insurance risk with the reinsurance premium calculated by the variance principle. We consider an incomplete market driven by the 4/2 stochastic volatility mode. This paper formulates the robust mean-field game  for a non-linear system originating from the variance principle and the 4/2 model.  For the case of an exponential utility function, we derive closed-form solutions for the $n$-insurer game and the corresponding mean-field game. We show that relative concerns lead to new hedging terms in the investment and reinsurance strategies. Model uncertainty can significantly change the insurers' hedging demands. The hedging demands in the investment-reinsurance strategies exhibit highly non-linear dependence with the insurers' competitive coefficients, risk aversion and  ambiguity aversion coefficients. Finally, numerical results demonstrate the herd effect of competition.
			\end{abstract}
		\vskip15pt
		\setcounter{equation}{0}
	\section{\bf Introduction}
Stochastic differential games have garnered considerable attention in recent financial mathematics literature. In the realm of two-person zero-sum differential games, \cite{elliott1976existence} establishes that under the Isaacs condition, the upper and lower values of the game are equal, and a saddle point exists in feedback strategies.  \cite{bensoussan2000stochastic} extend the discussion to stochastic differential games involving $N$ players by systems of nonlinear partial differential equations. While explicit solutions for mean-field games are generally rare outside linear-quadratic examples (as noted in \cite{yong2013linear} and \cite{bensoussan2016linear}), \cite{lacker2019mean} and \cite{MFE2020} contribute by constructing explicit constant equilibrium strategies for both finite population games and their corresponding mean-field games. In contrast to the non-competitive  case studied in \cite{merton1969lifetime} and \cite{merton1975optimum}, \cite{lacker2019mean} reveal that the unique constant Nash equilibrium comprises two components: the traditional Merton portfolio and the portfolio influenced by competition. { 
Since \cite{lacker2019mean}, portfolio games under various settings have been extensively studied. Some studies extend the results by exploring different optimization rules. In portfolio games, the optimization rule for a representative player is shaped by both the player's preferences and their reactions to other players. \cite{lacker2019mean} examine the game under CRRA and CARA utility functions. However, in practice, players may have diverse preferences, such as the forward relative performance criteria considered in \cite{anthropelos2022competition}, \cite{dos2022forward}, \cite{chong2024optimal}, and the mean-variance criteria discussed in \cite{guan2022time}, \cite{deng2024peer}, among others. The investment-consumption games studied in \cite{MFE2020} are further extended to settings involving non-exponential discounting in \cite{liang2023time}, habit formation in \cite{liang2024mean}, and in \cite{bo2024mean}. In \cite{lacker2019mean}, the reaction in portfolio games is homogeneous, where agents are symmetric and identical. However, in reality, players often aim to outperform a specific group of competitors rather than an entire population. To address this, \cite{caines2021graphon}, \cite{xu2024linear}, and \cite{tangpi2024optimal} extend portfolio games to heterogeneous settings on random graphs, where interactions among managers are heterogeneous. Additionally, \cite{hu2022n} study portfolio games where players interact through competition or homophily. The study also explores incomplete markets and random risk tolerance coefficients. Other contributions in the literature consider various market settings. For instance, \cite{kraft2020dynamic} model a two-agent framework under Heston's stochastic volatility model, while \cite{fu2020mean} investigate the case of two correlated stocks in a non-Markovian setting. The unique equilibrium in \cite{fu2020mean} is characterized by forward-backward stochastic differential equations. Portfolio games under random parameters are explored in \cite{park2022optimal}, \cite{fu2023mean}, and \cite{tangpi2024optimal}, among others. Optimal relative investment with jump risk is studied in \cite{guan2022time}, \cite{bo2024mean}, and \cite{aydougan2024optimal}, etc. Most of the studies mentioned above focus on portfolio games with regular controls. However, some recent works have also examined portfolio games with singular controls (see \cite{cao2022approximation}, \cite{campi2022mean}, \cite{denkert2024extended}) and impulse controls (see \cite{sadana2021nash}, \cite{sadana2022feedback}, \cite{basei2022nonzero}).
}

Abundant research has acknowledged the influence of competition on equilibrium strategies in stochastic differential games. However, numerous studies mentioned earlier tend to neglect model uncertainty and assume financial markets driven by diffusion processes, focusing primarily on mean-field games within linear systems. This paper addresses this research gap by exploring robust stochastic differential reinsurance and investment games under volatility risk. We contribute to the field of stochastic differential games by formulating the robust mean-field game in a non-linear system under model uncertainty.  Experimental results on the Ellsberg paradox often reveal behavior interpreted as ambiguity aversion (\cite{ellsberg1961risk}). Due to a lack of full confidence in the reference financial model, ambiguity attitude becomes crucial in decision-making. To account for model uncertainty, we employ the max-min expected utility  model (\cite{gilboa2004maxmin}), which is widely applied in mathematical finance and actuarial sciences. Previous research (\cite{maenhout2004robust}, \cite{zeng2016robust}, \cite{guan2019robust}, \cite{li2022robust}) highlights that robustness significantly reduces the demand for equities and decreases exposure to insurance and financial risks under decision-making ambiguity. While single-agent problems demonstrate the impact of model uncertainty on traditional portfolios, the effect on portfolios influenced by competition has not been thoroughly explored. This paper contributes to the literature by examining the impacts of model uncertainty and competition on insurers' reinsurance and investment strategies. Reports (\cite{biggar1998competition}, \cite{murat2002competition}) indicate that competition among insurers has long played a complex role in the insurance industry, with empirical evidence emphasizing its significance in decision-making. \cite{alhassan2018competition} examine the effect of competition on risk-taking behavior in an emerging insurance market. Additionally, \cite{chang2022does} investigates the relationship between insurers' competitiveness and risk-taking decisions, suggesting that insurers with low (high) competitiveness tend to take more (less) risks. In this paper, insurers care about relative wealth concerns and aim to outperform their peers.  

{ Moreover, in the works of \cite{lacker2019mean}, \cite{MFE2020}, and \cite{fu2023mean}, decision-makers only consider equity risk. However, the Black-Scholes model utilized in \cite{lacker2019mean} and \cite{MFE2020} fails to capture empirical phenomena like volatility smile and clustering. Financial markets, especially during periods of distress such as the 2008 financial crisis and the COVID-19 pandemic (\cite{kumar2021response}, \cite{enow2023investigating}), often experience heightened volatility. Recognizing the significance of return volatility in practical financial management decisions, this paper considers an incomplete market under volatility risk. \cite{liu2007portfolio} addresses stock portfolio selection under Heston's stochastic volatility model, revealing an optimal portfolio composed of myopic and hedging components.} \cite{kraft2020dynamic} investigates dynamic asset allocation with relative wealth concerns in a two-agent framework under Heston's stochastic volatility model. In contrast to optimal portfolios in \cite{lacker2019mean}, \cite{kraft2020dynamic} shows that relative wealth concerns introduce new hedge terms, and both heterogeneity in risk aversion and relative wealth concerns have similar effects on portfolio decision heterogeneity.
 Our work introduces model uncertainty into the stochastic differential game under incomplete markets, aiming to illustrate the impact of ambiguity on hedging demands in reinsurance and investment strategies. { While Heston's stochastic volatility model is widely used in risk management (\cite{kraft2005optimal}, \cite{vcerny2008mean}, \cite{huang2017robust}, etc.), it often violates the Feller condition when calibrated to real market data, highlighting limitations in its practical application (\cite{da2011riding}). To capture the steep skews observed in equity markets, the model typically requires a high volatility-of-volatility parameter. However, this results in volatility paths that spend significantly more time near zero than what is observed in the empirical distribution of volatility. Moreover, the Heston model predicts that, during periods of market stress, when instantaneous volatility increases, the skew will flatten. Consequently, the model places excessive weight on low or vanishing volatility scenarios and struggles to generate extreme paths with high volatility-of-volatility, which limits its ability to capture significant market dislocations. To address these limitations, the 3/2 model proposed by \cite{heston1997simple} is often considered, as it allows for extreme paths featuring spikes in instantaneous volatility. The 3/2 model has been found to be better supported by empirical studies compared to the Heston model. \cite{drimus2012options} demonstrates that, unlike the Heston model, the 3/2 model generates upward-sloping implied volatility-of-variance smiles, which align with the way variance options are traded in practice. Furthermore, the 3/2 model’s predicted behavior appears to be more accurate, as shown by \cite{itkin2013new}. Additional studies supporting the 3/2 model include \cite{carr2007new}, \cite{goard2013stochastic}, and \cite{yuen2015pricing}. The 4/2 stochastic volatility model, which generalizes both the Heston model in \cite{heston1993closed} and the 3/2 model in \cite{heston1997simple}, is examined for its ability to capture the dynamics of the implied volatility surface (\cite{grasselli20174}). In the Heston model, the short-term skew flattens as instantaneous variance increases, while in the 3/2 model, the short-term skew steepens under the same conditions. As a result, the 4/2 model offers greater flexibility in characterizing different dynamics in the evolution of the implied volatility surface. Optimal investment under the 4/2 stochastic volatility model for a single agent has been studied in \cite{cheng2021optimal}, \cite{wang2022optimal}, and others.} Given that the financial market is incomplete under the 4/2 model, we aim to explore the hedging demands, taking into account model uncertainty and competition.

Our research contributes to the literature on optimal reinsurance and investment, as well as the application of mean-field games, by incorporating model uncertainty and incomplete markets. We formulate robust $n$-insurer and mean-field games under model uncertainty within incomplete markets, where insurers face both common and idiosyncratic insurance risks. The option to purchase proportional reinsurance is available, with the reinsurance premium determined by the variance principle. Insurers, concerned with relative performance, seek robust equilibrium strategies under worst-case scenarios. Our primary emphasis is on determining robust $n$-insurer and mean-field equilibria within a competitive framework, particularly addressing insurers' strategies under volatility risk and model uncertainty.
 Our work extends the scope of \cite{lacker2019mean}, \cite{kraft2020dynamic}, and \cite{guan2022time} by incorporating model uncertainty and stochastic volatility. We introduce clear definitions for robust $n$-insurer and mean-field equilibrium, expanding upon the non-ambiguity case presented in \cite{lacker2019mean}. Using stochastic dynamic programming methods, we derive closed-form solutions by $n$-dimensional HJBI equations for the robust $n$-insurer game under the CARA utility function. We provide suitable conditions ensuring to verify the verification theorem. We also demonstrate the  consistency between the results of the $n$-insurer game and the mean-field game.

Our study introduces a solvable robust mean-field game within an incomplete market setting. Notably, mean-field games are seldom straightforward to solve, and our model adds complexity by incorporating common and idiosyncratic insurance risks, stochastic volatility, and model uncertainty. While research on model uncertainty and mean-field games has gained popularity, their combined exploration remains limited. Besides, the solvable mean-field games are mainly introduced for linear system, see \cite{carmona2013control},  \cite{bensoussan2016linear},  \cite{lacker2019mean}, \cite{MFE2020}, \cite{liang2022robust},  etc. Especially, in the incomplete markets presented in our work, the system is highly non-linear caused by the 4/2 model. Moreover, the variance principle of the reinsurance premium also leads to non-linearity in the system. The financial system in our work contains quadratic and square root terms.  We formulate the robust problem as a max-min optimization problem, deriving a system of $n$-dimensional HJBI equations. The robust problem necessitates additional integrable conditions, and our work presents some compatible conditions in Theorem \ref{Verification}.  Unlike previous studies (\cite{moon2016robust}, \cite{moon2016linear}) that achieve approximate optimality in the sense of $\epsilon$-Nash equilibrium in robust mean-field games, our work focuses on deriving Nash equilibrium strategies in the robust mean-field game. Besides, the optimization rules in \cite{moon2016robust}, \cite{moon2016linear}  are quadratic, this paper presents a solvable mean-field game with a non-quadratic goal. Solutions in close-form are derived for the robust mean-field game in a non-linear system.

The robust equilibrium strategies presented in Theorem \ref{thm:result} yield several noteworthy findings. The robust equilibrium investment strategy outlined in \eqref{dn-opt} is comprised of four distinct components: a myopic component, a hedging component, a myopic component influenced by competition,  and a hedging component influenced by competition. The first component optimally addresses Merton's problem under model uncertainty and constant volatility. The addition of the second component is designed to hedge against volatility risk. All four components depend on the risk aversion and ambiguity aversion coefficients. The influence of competition, risk aversion, and ambiguity aversion coefficients on the third  and fourth components results in highly non-linear dependence. Notably, the dependence on ambiguity aversion is more non-linear than that on risk aversion and competition coefficients. Moreover, the robust equilibrium reinsurance proportions are determined by two components: the first optimized for a single insurer and the second influenced by competition. Both components display non-linear dependencies on risk aversion and ambiguity aversion coefficients. The numerical results shed light on the impact of competition, risk aversion, and ambiguity aversion coefficients on the robust equilibrium strategies. A pronounced herd effect of competition is observed, wherein insurers tend to mimic their peers' behaviors, leading to convergence in strategies. In a competitive environment, insurers exhibit similar approaches, and when some prioritize relative performance or display higher levels of risk or ambiguity aversion, all insurers tend to increase proportional reinsurance and cash allocation in response.

The remainder of this paper is organized as follows. Section 2 shows the financial model. Section 3 presents the robust $n$-insurer game. The solutions of the robust $n$-insurer game are obtained in Section 4. Section 5 presents and  verify the solutions of the robust mean-field game. Section 6 shows the economic behaviors of the insurers and Section 7 is a conclusion. All the proofs are in Appendices.
\section{\bf Financial model}
Let $\left(\Omega, \mathcal{F}, \mathbb{P}\right)$ be a  complete probability space, $T>0$ be a fixed time horizon and  $\mathbb{F}= \left\{\mathcal{F}_{t},{t \in[0, T]}\right\} $ is the usual augmentation of the natural	filtration generated by the involved standard Brownian motions below supported in $\left(\Omega, \mathcal{F}, \mathbb{P}\right)$. $\mathbb{P}$ represents the reference probability measure in the market. In this paper, we consider $n$  insurers who compete with each other. The insurers have own preferences and invest in the financial market consisting with cash and one stock. The insurers can adjust the strategies within time horizon $[0,T]$. Suppose that all the processes introduced below are well defined and adapted to the filtration  $\mathbb{F}$.  


To characterize the insurers' wealth processes, we  consider an approximate model to the Cramér-Lundberg insurance model. In the Cramér-Lundberg insurance model, the insurer $i$'s, $i \in\{1,2, \ldots, n\}$, dynamic surplus process $\tilde{X}_i=\left\{\tilde{X}_{i}(t), t\in[0,T]  \right\}$ involves as follows
\begin{equation}\label{equ:cl}
\tilde{X}_{i}(t)=\tilde{x}^{0}_{i}+p_{i} t-\sum_{j=1}^{K_{i}(t)} Y_{j}^{i},
\end{equation}
where $\tilde{x}^{0}_{i} \geqslant0 $ is the initial surplus; the constant $p_{i}>0$ is the premium rate; $Y_{j}^{i}$ is the size of the {$j$-th} claim, and $\{Y_{j}^{i}, j=1,2, \ldots\}$ is a sequence of independent identically distributed (i.i.d.) random variables  with  common distribution function $F_{Y^{i}}\left(\cdot\right)$. $N_{i}=\left\{N_{i}(t), t\in[0,T]\right\}$ and $\hat{N}=\{\hat{N}(t), t\in[0,T]\}$ are mutually independent Poisson processes, with intensities being $\lambda_{i}>0$ and $\hat{\lambda}>0$, respectively. $K_{i}(t)=N_{i}(t)+\hat{N}(t)$ represents the number of {claims} up to time $t$ for the insurer $i$.  $N_i$ and $\hat{N}$ represent the common and idiosyncratic shocks of the insurance businesses, respectively. We assume that $Y^{i}$ has finite mean $\mu_{i 1}=\mathbb{E}\left(Y^{i}\right)$ and second-order moment $\mu_{i 2}=\mathbb{E}\left[\left(Y^{i}\right)^{2}\right]$, and the premium rate $ p_i $ is determined via the expected value principle, that is,$$
p_{i}=\frac{1}{t}\left(1+\eta_{i}\right) \mathbb{E}\left[\sum_{j=1}^{K_{i}(t)} Y_{j}^{i}\right]=\left(1+\eta_{i}\right)\left(\lambda_{i}+\hat{\lambda}\right) \mu_{i 1},
$$where  $ \eta_{i}>0 $ is the safety loading coefficient of the insurer $ i $.

 As the diffusion approximation model works well for the insurance portfolio, whose wealth is large enough such that each individual claim is relatively small compared with the total reserve. Therefore, we use a drifted Brownian motion to approximate \eqref{equ:cl}, that is,
\begin{equation}\label{equ:appr}
\sum_{j=1}^{K_{i}(t)} Y_{j}^{i} \approx\left(\hat{\lambda}+\lambda_{i}\right) \mu_{i 1} t-\sqrt{\left(\hat{\lambda}+\lambda_{i}\right) \mu_{i 2}} W_{i}(t).
\end{equation}
{
 
The diffusion approximation of the collective risk model was first introduced by \cite{iglehart1969diffusion} and \cite{grandell1977class}. They demonstrated that as \( t \to +\infty \), the Cramér-Lundberg insurance model converges to the diffusion model in distribution. Since the diffusion model is more analytically tractable than the compound Poisson process, it has been widely adopted in insurance research (see, for example, \cite{taksar2003optimal}, \cite{yi2015robust}, \cite{guan2023dynamic}). The accuracy of this approximation\footnote{
Let \( F_t(x) \) denote the distribution function of the standardized Poisson random sum, given by
$
\frac{\sum_{j=1}^{K_i(t)} Y_j^i - (\hat{\lambda} + \lambda_i) \mu_{i1} t}{\sqrt{(\hat{\lambda} + \lambda_i) \mu_{i2} t}}
$. Define the uniform distance between the distribution function of the standardized Poisson random sum and the standard normal distribution function \( \Phi(x) \) as
$
\Delta_t = \sup\limits_x | F_t(x) - \Phi(x) |
$. 
\cite{shevtsova2014accuracy} proves that the uniform distance satisfies the bound
$
\Delta_t \leq 0.266 l_t + 0.5 l_t^2,
$ 
where \( l_t = \frac{\mathbb{E}[|Y_j^i|^3]}{\sqrt{(\hat{\lambda} + \lambda_i) t} (\mathbb{E}[|Y_j^i|^2])^{3/2}} \). The uniform distance between the diffusion model and the collective risk model can also be derived.
} is discussed in \cite{shevtsova2014accuracy}.
}

Here $W_{i}=\{W_{i}(t), t\in[0,T]\}$ is a standard Brownian motion, and that any two Brownian motions $W_i$ and $W_{k}, i \neq k \in\{1,2, \ldots, n\}$, are correlated with the correlation coefficient being
$$
\rho_{i k}:=\frac{\hat{\lambda} \mu_{i 1} \mu_{k 1}}{\sqrt{\left(\hat{\lambda}+\lambda_{i}\right)\left(\hat{\lambda}+\lambda_{k}\right) \mu_{i 2} \mu_{k 2}}}.
$$ 
In other words, we assume that $$ W_i(t)=\sqrt{\frac{\hat{\lambda}\mu_{i 1}^2}{(\hat{\lambda}+\lambda_i)\mu_{i 2}}} \tilde{W}(t)+\sqrt{1-\frac{\hat{\lambda}\mu_{i 1}^2}{(\hat{\lambda}+\lambda_i)\mu_{i 2}}}\hat{W}_i(t),$$ where $\tilde{W}=\{\tilde{W}(t), t\in[0,T]\}$ and $\hat{W}_i=\{\hat{W}_i(t), t\in[0,T]\}$, $i\leqslant n $ are standard Brownian motions and are independent with each other. 
Denote $\hat{\textbf{W}}(t)= (\hat{W}_1(t),\hat{W}_2(t),\cdots,\hat{W}_n(t))^T $. $ \tilde{W}$ represents the common insurance risk and $\hat{W}_i$ represents the idiosyncratic insurance of insurer $i$.


The insurers are faced with insurance risk and financial risk. To manage these two risks, the insurers can purchase reinsurance products and invest in the financial market. We assume that the risk-free rate in the market is $ r $, i.e., the price of the money account in the financial market involves as
{$$
\frac{\rd S_0(t)}{S_0(t)}=r\rd t, \quad S_0(0)=1.
$$}
We consider the $ 4/2 $ stochastic-volatility model as in \cite{grasselli20174} for the stock, that is, the price of  the stock involves as follows
\begin{equation}\label{equ:st}
\frac{\rd S(t)}{S(t)}=\left(r +m({a}{Z(t)}+{{b}})\right)\rd t+{\Sigma}(t)\rd {W}(t), \quad S(0)=s_0,
\end{equation}
where\begin{equation*}
	{\Sigma}(t)={a}\sqrt{Z(t)}+\frac{{b}}{\sqrt{Z(t)}},\quad
	\rd Z(t)={\kappa}(\bar{Z}-Z(t))\rd t+{\nu}\sqrt{Z(t)}\left[{\rho}\rd {W}(t)+\sqrt{1-{\rho}^2}\rd{B}(t)\right],\quad Z(0)=z^0,
\end{equation*}
where $a,b\geqslant0$ and $a^2+b^2\neq 0$. \eqref{equ:st} includes the Heston model in \cite{heston1993closed} ($a>0, b=0$) and the 3/2 model in \cite{heston1997simple} ($a=0,b>0$) as special cases. $m$ is a constant and ${m}\sqrt{Z(t)}$ represents the stock's market price of risk. \eqref{equ:st} can describe the time-varying volatility and  accurately capture the evolution of the implied volatility surface. ${W}(t), {B}(t)$ are mutually independent standard Brownian motions, which represent the equity risk and volatility risk, respectively. $\kappa$, $\bar{Z}$ and $\nu$ are positive constants. Besides, we impose the Feller condition $ 2{\kappa}\bar{Z}>{\nu}^2$ to ensure that the process $ Z(t) $ is strictly positive. 

The insurers can manage their insurance risk through purchasing proportional reinsurance from the reinsurer and invest in the financial market.  While in the reinsurance model,  {let  $a_{i}(t)\in$ $[0,1] $ be the retention level  of the insurer $i$, then $1-a_{i}(t) $ is the proportional reinsurance level of the insurer $i$}. Denote  the reinsurance premium rate paid by insurer $ i $   at time $ t $ by $\hat{p}_i(t)$. We suppose that  $\hat{p}_i(t)$ is calculated by variance principle, that is, $$ \hat{p}_i(t) =  (1-a_i(t))\left(\lambda_{i}+\hat{\lambda}\right)\mu_{i 1}+\hat{\eta} (1-a_i(t))^2 \left(\lambda_{i}+\hat{\lambda}\right) \mu_{i 2},
$$ where  $ \hat{\eta}>0 $ is the safety loading coefficient of the reinsurer.
Moreover, the insurer $i$ can invest his surplus in a risk-free money account with a constant interest rate $r>0$ and stock. Denote $\pi_{i}(t)$  as the amount invested in the stock  for insurer $i$  at time $t$ and the rest part is invested in cash. Then with the reinsurance and money account being incorporated, the surplus process of insurer $i$ involves as
{	$$
		\mathrm{d} \hat{X}_{i}(t)= (\hat{X}_{i}(t)-\pi_i(t))\frac{\rd S_0(t)}{S_0(t)}+\pi_{i}(t)\frac{\rd S(t)}{ S(t)}+p_i\rd t-\hat{p}_i(t)\rd t- a_i(t)\mathrm{d}\sum_{j=1}^{K_i(t)} Y_j^i .
	$$Based on \eqref{equ:appr}, in the diffusion approximation model, the surplus process of insurer $i$ $X_i=\left\{X_{i}(t), t\in[0,T]  \right\}$ involves as$$
	\begin{aligned}
		\mathrm{d} X_{i}(t)=& (X_{i}(t)-\pi_i(t))\frac{\rd S_0(t)}{S_0(t)}+\pi_{i}(t)\frac{\rd S(t)}{ S(t)}+p_i\rd t-\hat{p}_i(t)\rd t\\&- a_i(t)\left[\left(\lambda_{i}+\hat{\lambda}\right) \mu_{i 1}\rd t -\left(\sqrt{\hat{\lambda}}\mu_{i 1} \mathrm{d}\tilde{W}(t)+\sqrt{(\hat{\lambda}+\lambda_{i})\mu_{i 2}-\hat{\lambda}\mu_{i 1}^2}\mathrm{d}\hat{W}_i(t)\right)\right], 
		\end{aligned}
	$$which indicates that}
\begin{equation}\label{equ:xi}
\begin{split}
	\mathrm{d} X_{i}(t)
	=&\left[r X_{i}(t)+\eta_{i}\left(\lambda_{i}+\hat{\lambda}\right) \mu_{i 1}-\hat{\eta}(1-a_i(t))^2\left(\lambda_{i}+\hat{\lambda}\right) \mu_{i 2}+\pi_{i}(t)m\sqrt{Z(t)}\Sigma(t)\right] \mathrm{d} t\\&+a_{i}(t) \left(\sqrt{\hat{\lambda}}\mu_{i 1} \mathrm{d}\tilde{W}(t)+\sqrt{(\hat{\lambda}+\lambda_{i})\mu_{i 2}-\hat{\lambda}\mu_{i 1}^2}\mathrm{d}\hat{W}_i(t)\right)+\pi_{i}(t){\Sigma}(t)\rd {W}(t).
\end{split}
\end{equation}
The reinsurance and investment strategy at time $t$ of insurer $i$ is $(a_i(t),\pi_i(t))$. \eqref{equ:xi} contains quadratic and square root terms and is a non-linear system arising from the variance principle and 4/2 model.

\section{\bf Robust $n$-insurer game}
In the following section, we present the robust equilibrium game for insurers, which involves two key aspects. Firstly, the insurers engage in a competitive environment and prioritize relative performance, leading to the formulation of the $n$-insurer game. Secondly, due to the inability to estimate parameters accurately, it is crucial to consider model uncertainty. Compared with previous work, such as \cite{carmona2013control},  \cite{bensoussan2016linear},  \cite{lacker2019mean}, \cite{MFE2020}, \cite{liang2022robust},  etc, we formulate the robust $n$-insurer and mean-field games for a non-linear system.

The insurers compete with each other and are concerned with the relative performance over the average wealth $$\bar{X}(T)=\frac{1}{n}\sum_{k=1}^{n} X_k(T).$$ Insurer $i$ is not concerned with the absolute wealth $X_{i}(t)$ while with the relative wealth over the others defined as follows
$$ Y_i(t):=X_{i}(t) {-\theta_{i}}\bar{X}(t),$$ 
where $\theta_{i} \in[0,1]$ describes the insurer $i$'s risk preference for her/his own wealth versus relative wealth. The insurer is more cared about the relative wealth for a larger $\theta_i$. The insurers compete with each other and we are interested in the Nash equilibrium strategies of them.

Furthermore, insurer $i$ is not fully {confident} with the reference model and is ambiguity averse. To characterize insurer $i$'s belief over the reference model, we define a class of  probability measures that are equivalent to $\mathbb{P}$ by
$$
\mathcal{Q}^{i}:=\left\{\mathbb{Q}^{\varphi_i,\chi_i,\phi_i,\vartheta_{i}} \mid \mathbb{Q}^{\varphi_i,\chi_i,\phi_i,\vartheta_{i}} \sim \mathbb{P}\right\}.
$$Functions  $\varphi_{i}(t)$, $ \chi_{i}(t),\phi_i(t),\vartheta_{i}(t)=(\vartheta_{i,1}(t),\vartheta_{i,2}(t),\cdots,\vartheta_{i,n}(t))^T$ satisfy the following condition.
\begin{enumerate}
	\item $\varphi_{i}(t), \chi_{i}(t),\phi_i(t),\vartheta_{i}(t)$ are progressively measurable w.r.t.{(with respect to)} the filtration $\mathbb{F}$;
	\item The Novikov's condition is satisfied:
		$$
	\mathbb{E}\left[\exp \left(\frac{1}{2} \int_{0}^{T}\left[{\varphi_{i}^2(t)}+\chi_{i}^2(t)+\phi_i^2(t)+\vartheta_{i}(t)^T\vartheta_{i}(t)\right] \mathrm{d} t\right)\right]<+\infty.
	$$
\end{enumerate} 

For each $\varphi_{i}(t), \chi_{i}(t),\phi_i(t),\vartheta_{i}(t)$, the real valued process $\left\{\Theta^{\varphi_i,\chi_i,\phi_i,\vartheta_{i}}(t)\right\}_{t \in[0, T]}$ is defined by
$$
\begin{aligned}
		\Theta^{\varphi_i,\chi_i,\phi_i,\vartheta_{i}}(t)=& \exp \left(\int_{0}^{t} \varphi_{i}(s) \mathrm{d} W(s)+\int_{0}^{t} \chi_{i}(s) \mathrm{d}B(s)-\frac{1}{2} \int_{0}^{t} \varphi_{i}^2(s) +\chi_{i}^2 (s)\mathrm{d} s\right)\\
		&\times \exp \left(\int_{0}^{t} \phi_{i}(s) \mathrm{d} \tilde{W}(s)+\int_{0}^{t} \vartheta_{i}(s)^T \mathrm{d}\hat{\textbf{W}}(s)-\frac{1}{2} \int_{0}^{t} \phi_{i}^2(s) +\vartheta_{i}(s)^T\vartheta_{i}(s)\mathrm{d} s\right)
\end{aligned}
$$
Then $\left\{\Theta^{\varphi_i,\chi_i,\phi_i,\vartheta_{i}}(t)\right\}_{t \in[0, T]}$ is a $\mathbb{P}$-martingale and there is an equivalent probability measure $ \mathbb{Q}^{\varphi_i,\chi_i,\phi_i,\vartheta_{i}}$ defined by
$$
\left.\frac{\mathrm{d} \mathbb{Q}^{\varphi_i,\chi_i,\phi_i,\vartheta_{i}}}{\mathrm{~d} \mathbb{P}}\right|_{\mathcal{F}_{T}}=\Theta^{\varphi_i,\chi_i,\phi_i,\vartheta_{i}}(T) .
$$
It follows from the Girsanov's Theorem that the  processes $  {{W}}^{ \mathbb{Q}^{\varphi_i,\chi_i,\phi_i,\vartheta_{i}}}=\left\{{{W}}^{ \mathbb{Q}^{\varphi_i,\chi_i,\phi_i,\vartheta_{i}}}(t),0\leqslant t\leqslant T\right\}$, $ B^{ \mathbb{Q}^{\varphi_i,\chi_i,\phi_i,\vartheta_{i}}}=\left\{{{B}}^{ \mathbb{Q}^{\varphi_i,\chi_i,\phi_i,\vartheta_{i}}}(t),0\leqslant t\leqslant T\right\} $, $ \tilde{W}^{ \mathbb{Q}^{\varphi_i,\chi_i,\phi_i,\vartheta_{i}}}=\left\{\tilde{W}^{ \mathbb{Q}^{\varphi_i,\chi_i,\phi_i,\vartheta_{i}}}(t),0\leqslant t\leqslant T\right\}$ and $  \hat{\textbf{W}}^{ \mathbb{Q}^{\varphi_i,\chi_i,\phi_i,\vartheta_{i}}}=\left\{\hat{\textbf{W}}^{ \mathbb{Q}^{\varphi_i,\chi_i,\phi_i,\vartheta_{i}}}(t),0\leqslant t\leqslant T\right\}$ following  \eqref{newBrown} are standard Brownian motions under $ \mathbb{Q}^{\varphi_i,\chi_i,\phi_i,\vartheta_{i}}$, where
\begin{equation}\label{newBrown}
	\begin{aligned}
	&\mathrm{d} {{W}}^{ \mathbb{Q}^{\varphi_i,\chi_i,\phi_i,\vartheta_{i}}}(t)=\mathrm{d}W(t)-\varphi_i(t) \mathrm{d} t, \\
	&\mathrm{d} B^{ \mathbb{Q}^{\varphi_i,\chi_i,\phi_i,\vartheta_{i}}}(t)=\mathrm{d} B(t)-\chi_{i}(t) \mathrm{d} t,\\
	&\mathrm{d} {\tilde{W}}^{ \mathbb{Q}^{\varphi_i,\chi_i,\phi_i,\vartheta_{i}}}(t)=\mathrm{d}\tilde{W}(t)-\phi_i(t) \mathrm{d} t,\\
	&\mathrm{d} {\hat{\textbf{W}}}^{  \mathbb{Q}^{\varphi_i,\chi_i,\phi_i,\vartheta_{i}}}(t)=\mathrm{d}\hat{\textbf{W}}(t)-\vartheta_{i}(t) \mathrm{d} t.
\end{aligned}
\end{equation}
As such, under the equivalent probability measure $ \mathbb{Q}^{\varphi_i,\chi_i,\phi_i,\vartheta_{i}}$ and the strategies $ \left\{(\pi_{k},a_k)_{k=1}^{n}\right\} $,  the wealth process of insurer $k, k=1,2, \cdots, n$ involves as
\begin{equation}\label{SDE-X}
	\begin{aligned}
		\mathrm{d} X_{k}(t)\!
		=\!&\left[r X_{k}(t)+\eta_{k}\left(\lambda_{k}\!+\!\hat{\lambda}\right) \!\mu_{k 1}\!-\!\hat{\eta}(1\!-\!a_k(t))^2\left(\lambda_{k}\!+\!\hat{\lambda}\right)\! \mu_{k 2}\!+\!\pi_{k}(t)m\sqrt{Z(t)}\Sigma(t)\right]\!\mathrm{d} t,\\
		&+a_{k}(t)\left(\sqrt{\hat{\lambda}}\mu_{k1}(\mathrm{d} {\tilde{W}}^{ \mathbb{Q}^{\varphi_i,\chi_i,\phi_i,\vartheta_{i}}}(t)+\phi_{i}(t)\rd t)+\sqrt{(\hat{\lambda}+\lambda_{k})\mu_{k 2}-\hat{\lambda}\mu_{k 1}^2}(\mathrm{d} {\hat{\textbf{W}}}_k^{ \mathbb{Q}^{\varphi_i,\chi_i,\phi_i,\vartheta_{i}}}(t)+\vartheta_{i,k}(t) \mathrm{d} t )\right)\\
		&+\pi_{k}(t){\Sigma}(t)(\rd W^{ \mathbb{Q}^{\varphi_i,\chi_i,\phi_i,\vartheta_{i}}}(t)+\varphi_i(t)\rd t),
	\end{aligned}
\end{equation}where
\begin{equation*}
	\rd Z(t)={\kappa}(\bar{Z}-Z(t))\rd t+{\nu}\sqrt{Z(t)}\left[{\rho}(\rd W^{ \mathbb{Q}^{\varphi_i,\chi_i,\phi_i,\vartheta_{i}}}(t)\!+\!\varphi_i(t)\rd t)\!+\!\sqrt{1\!-\!{\rho}^2}(\rd B^{ \mathbb{Q}^{\varphi_i,\chi_i,\phi_i,\vartheta_{i}}}(t)\!+\!\chi_i(t)\rd t)\right]\!.
\end{equation*}

 Let $\gamma\left( \mathbb{Q}^{\varphi_i,\chi_i,\phi_i,\vartheta_{i}}\right)$ be the penalty term of $ \mathbb{Q}^{\varphi_i,\chi_i,\phi_i,\vartheta_{i}}$ to describe the distance with the reference probability measure.  We use the following form of the penalty to measure the ambiguity averse attitude
$$
\gamma\left( \mathbb{Q}^{\varphi_i,\chi_i,\phi_i,\vartheta_{i}}\right)=\int_{0}^{T}  \left[\frac{\varphi_{i}^{2}(t)}{2 \Psi_1^{i}}+\frac{\chi_{i}^2(t)}{2 \Psi_2^{i}}+\frac{\phi_{i}^{2}(t)}{2 \Psi_3^{i}}+\frac{\vartheta_{i}(t)^T\vartheta_{i}(t)}{2 \Psi_4^{i}} \right]\mathrm{d} t,
$$
where functions ${\varphi_{i}^{2} \over 2}$, ${\chi_{i}^2(t)\over 2}$, ${\phi_{i}^{2}(t) \over 2}$ and ${\vartheta_{i}(t)^T\vartheta_{i}(t) \over 2}$ are  the increasing rates of the relative entropy, which has been show in \cite{guan2022robust}. $ \Psi_1^{i} $, $ \Psi_2^{i} $, $ \Psi_3^{i} $ and $ \Psi_4^{i} $ represent the degrees of ambiguity aversion of insurer $i$   with respect to the equity risk, volatility risk, common insurance risk and idiosyncratic insurance risk, respectively.  $ \Psi_1^{i} $, $ \Psi_2^{i} $, $ \Psi_3^{i} $ and $ \Psi_4^{i} $ will be determined later to get closed-form solutions. 

 The objective function of insurer $i$ with competition and ambiguity aversion is as follows:
  \begin{equation*}
	\begin{aligned}
&J_{i}\left(\left\{(\pi_{k},a_k)_{k=1}^{n}\right\},\left(\varphi_i,\chi_i,\phi_i,\vartheta_{i}\right),t,y,z\right)\\=&\mathbb{E}_{t,y,z}^{ \mathbb{Q}^{\varphi_i,\chi_i,\phi_i,\vartheta_{i}}}\left[ U_f\left(Y_{i}(T); \delta_{i}\right)+\int_{t}^{T}  \frac{\varphi_{i}^{2}(s)}{2 \Psi_1^{i}}+\frac{\chi_{i}^2(s
	)}{2 \Psi_2^{i}} +\frac{\phi_{i}^{2}(s)}{2 \Psi_3^{i}}+\frac{\vartheta_{i}(s)^T\vartheta_{i}(s)}{2 \Psi_4^{i}} \mathrm{d} s \right],
	\end{aligned}
\end{equation*}where $ \mathbb{E}_{t,y,z}^{ \mathbb{Q}^{\varphi_i,\chi_i,\phi_i,\vartheta_{i}}}[\cdot] $ means the conditional expectation under probability measure $ \mathbb{Q}^{\varphi_i,\chi_i,\phi_i,\vartheta_{i}}$ given states $ Y_i(t)=y$ and $Z(t)=z $. The risk aversion attitude of insurer $i$ is represented by function $U_f\left(Y_{i}(T); \delta_{i}\right)$ and $\delta_{i}$ describes insurer $i$'s degree of risk aversion. We suppose that the insurers have different degrees of risk aversion.

 Then  the {preference} of insurer $i$ under probability measure $ \mathbb{Q}^{\varphi_i,\chi_i,\phi_i,\vartheta_{i}}$ is formulated as
$$
J_{i}\left(\left\{(\pi_{k},a_k)_{k=1}^{n}\right\},\left(\varphi_i,\chi_i,\phi_i,\vartheta_{i}\right),0,y_i^0,z^0\right)
$$

{
 In this paper, we consider noncooperative games as in \cite{lacker2019mean}, where players aim to maximize their individual preferences in response to the aggregation effects of other players. The players search for a Nash equilibrium in the game. In contrast to noncooperative games, the social control game involves players cooperating to optimize the common social cost, defined as the sum of individual costs. \cite{huang2012social} study the social optimal control of mean field LQG (linear-quadratic-Gaussian) models, while \cite{wang2020social} examine social optimal control of mean field LQG models under drift uncertainty. \cite{huang2021social} extend this by studying mean field LQG social optimum control with volatility-uncertain common noise. Our work differs from these studies in several ways. First, while the financial system we consider is nonlinear with a non-quadratic optimization rule, \cite{wang2020social} and \cite{huang2021social} focus on linear stochastic systems with quadratic cost functions. Second, our study centers on noncooperative games, whereas the social control games in \cite{wang2020social} and \cite{huang2021social} are cooperative. In cooperative games, players seek a particular Pareto optimum.

}

Next, we  present the definition of a robust Nash equilibrium game for the $n$-insurer game under insurance and volatility risks. First, to ensure the robust control problem {has} a unique solution and  the verification theorem holds, we define the admissible set of strategies  and probability transformation functions for each insurer:
\begin{equation*}
	\begin{aligned}
		\mathscr{U}=&\Big\{\{(\pi_k,a_k)_{k=1}^n\}:(\pi_k,a_k) \text{ is progressively measurable w.r.t. }\mathbb{F}\text{ and }0\leqslant a_k(t)\leqslant1, \forall 1\leqslant k\leqslant n\\
		&\exists \{\mathcal{C}_k\}_{k=1}^n\subseteq\mathbb{R}_+, \sup_{1\leqslant k\leqslant n}\mathcal{C}_k<\infty, \text{ such that } \pi_k(t)=\ell_k(t)\frac{Z(t)}{aZ(t)+b},|\ell_k(t)|\leqslant \mathcal{C}_k, \forall t\in[0,T]\Big\},\\
		\mathscr{U}_i=&\Big\{(\pi_i,a_i):(\pi_i,a_i) \text{ is progressively measurable w.r.t. }\mathbb{F}\text{ and }0\leqslant a_i(t)\leqslant1,\\
		&\exists \mathcal{C}_i\in\mathbb{R}_+, \text{ such that } \pi_i(t)=\ell_i(t)\frac{Z(t)}{aZ(t)+b},|\ell_i(t)|\leqslant \mathcal{C}_i, \forall t\in[0,T]\Big\},\\
		\mathscr{U}_{-i}=&\Big\{\{(\pi_k,a_k)_{k\neq i}\}:(\pi_k,a_k) \text{ is progressively measurable w.r.t. }\mathbb{F}\text{ and }0\leqslant a_k(t)\leqslant1, \forall k\neq i\\
		&\exists \{\mathcal{C}_k\}_{k\neq i}\subseteq\mathbb{R}_+, \sup_{k\neq i}\mathcal{C}_k<\infty, \text{ such that } \pi_k(t)=\ell_k(t)\frac{Z(t)}{aZ(t)+b},|\ell_k(t)|\leqslant \mathcal{C}_k, \forall t\in[0,T]\Big\},\\
		\mathscr{A}=&\Big\{\left(\varphi, \chi,\phi,\vartheta\right):\left(\varphi, \chi,\phi,\vartheta\right) \text{ is progressively measurable w.r.t. }\mathbb{F}\text{ and }\left(\frac{\varphi(t)}{\sqrt{Z(t)}}, \frac{\chi(t)}{\sqrt{Z(t)}}\right)=\left(\hslash(t),\hbar(t) \right),\\&\quad
	|\hslash(t)|, |\hbar(t)|\leqslant \mathcal{C},\sup|\phi|<\infty,\sup\|\vartheta\|<\infty,~ \forall t\in[0,T]\Big\},\\
	\end{aligned}
\end{equation*} 
where $ \mathcal{C} $ is a constant and $ \mathcal{C}^2<\frac{\kappa^2}{2\nu^2} $.
\begin{remark}
For $ \left(\varphi, \chi,\phi,\vartheta\right)\in\mathscr{A} $, we see that \begin{equation*}
	\begin{aligned}
&\mathbb{E}\left[\exp \left(\frac{1}{2} \int_{0}^{T}\left[{\varphi^2(t)}+\chi^2(t)+\phi^2(t)+\vartheta(t)^T\vartheta(t)\right] \mathrm{d} t\right)\right]\\
=&\mathbb{E}\left[\exp \left(\frac{1}{2} \int_{0}^{T}\left[{\hslash^2(t)}+\hbar^2(t)\right]Z(t) \mathrm{d} t+\frac{1}{2} \int_{0}^{T}\phi^2(t)+\vartheta(t)^T\vartheta(t) \mathrm{d} t\right)\right]\\
\leqslant&\mathbb{E}\left[\exp\left(\frac{T}{2}\sup|\phi|^2+\frac{T}{2}\sup\|\vartheta\|^2\right)\times\exp \left( \mathcal{C}^2\int_{0}^{T}Z(t) \mathrm{d} t\right)\right]
	\end{aligned}
\end{equation*}
by Proposition 5.1 in \cite{kraft2005optimal}, if condition\begin{equation*}
	\mathcal{C}^2<\frac{\kappa^2}{2\nu^2}
\end{equation*}is satisfied, then \begin{equation*}
\mathbb{E}\left[\exp \left( \mathcal{C}^2\int_{0}^{T}Z(t) \mathrm{d} t\right)\right]<\infty.
\end{equation*}
Thus \begin{equation*}
	\mathbb{E}\left[\exp \left(\frac{1}{2} \int_{0}^{T}\left[{\varphi^2(t)}+\chi^2(t)+\phi^2(t)+\vartheta(t)^T\vartheta(t)\right] \mathrm{d} t\right)\right]<\infty.
\end{equation*}
The Novikov's condition holds and $\mathbb{Q}^{\varphi,\chi,\phi,\vartheta}$ is a well-defined probability  measure which is equivalent to $ \mathbb{P} $.
\end{remark}

\begin{definition}\label{def1}
	We say that   $\left\{(\pi_{k}^*,a_k^*)_{k=1}^{n}\right\}\in\mathscr{U}$ is a \textit{\textbf{robust equilibrium}} if $\forall 1\leqslant i\leqslant n$, $\{(\pi_k^*,a_k^*)_{k\neq i}\}\in\mathscr{U}_{-i}$,  and {for all} admissible strategies  $(\pi_i,a_i)\in\mathscr{U}_i$, we have
	\begin{equation*}
		\begin{aligned}
			&\inf _{\left(\varphi_i,\chi_i,\phi_i,\vartheta_{i}\right)\in\mathscr{A}} J_{i}\left(\left\{(\pi_{k}^*,a_k^*)_{k=1}^{n}\right\},\left(\varphi_i,\chi_i,\phi_i,\vartheta_{i}\right),0,y_i^0,z^0\right) \\\geqslant& \inf _{\left(\varphi_i,\chi_i,\phi_i,\vartheta_{i}\right)\in\mathscr{A}} J_{i}\left(\left\{(\pi_{k}^*,a_k^*)_{k\neq i},(\pi_{i},a_i)\right\},\left(\varphi_i,\chi_i,\phi_i,\vartheta_{i}\right),0,y_i^0,z^0\right).
		\end{aligned}
	\end{equation*}
\end{definition}
{ The robust equilibrium strategies here mean that no one can improve her/his performance in the worst-case scenario by unilaterally changing her/his strategies.}
In the market, the insurers are concerned with relative performances and are faced with insurance and financial risks, as well as model uncertainty. Then, we formulate the robust equilibrium game among the insurers. The goals of the insurers are to {search for} the robust Nash equilibrium strategies satisfying Definition \ref{def1} within the admissible set.

\section{\bf Robust $n$-insurer equilibrium}

We suppose that the utility functions be modeled as in the CARA case, that is 
$$
U_f(x ; \delta)=-\frac{1}{\delta} e^{-\delta x},
$$
where $\delta>0$ represent the risk aversion parameter. The insurer is more risk aversion with a larger $\delta$. In the robust equilibrium, each insurer searches for the optimal strategies under the worst-case scenario in response to the strategies choices of other insurers.  To obtain robust equilibrium strategies, we need to solve  Problem \eqref{oi}  for every insurer simultaneously.
\begin{equation}\label{oi}
(\pi^*_{i},a^*_i)=\arg	\sup _{(\pi_{i},a_i)\in\mathscr{U}_i} \inf _{\left(\varphi_i,\chi_i,\phi_i,\vartheta_{i}\right)\in\mathscr{A}} J_{i}\left(\left\{(\pi_{k}^*,a_k^*)_{k\neq i},(\pi_{i},a_i)\right\},\left(\varphi_i,\chi_i,\phi_i,\vartheta_{i}\right),0,y_i^0,z^0\right),
\end{equation}
where
$$
\left\{(\pi_{k}^*,a_k^*)_{k\neq i}\right\}\in\mathscr{U}_{-i}.
$$

First, we derive the dynamics of $Y_i$ under  under the equivalent probability measure $ \mathbb{Q}^{\varphi_i,\chi_i,\phi_i,\vartheta_{i}}$. For $ \left(\varphi_i,\chi_i,\phi_i,\vartheta_{i}\right)\in\mathscr{A} $, under the equivalent probability measure $ \mathbb{Q}^{\varphi_i,\chi_i,\phi_i,\vartheta_{i}}$ and the strategy $ \left\{(\pi_{k}^*,a_k^*)_{k\neq i}\right\}\in\mathscr{U}_{-i}$, $(\pi_{i},a_i)\in\mathscr{U}_{i} $, based on \eqref{SDE-X},
 we see 

\begin{equation*}
	\begin{aligned}
		\mathrm{d}Y_i(t)\!
		=&r Y_i(t)\rd t\!+\!(1\!-\!\frac{\theta_i}{n})\!\left[\!\eta_{i}\!\left(\lambda_{i}\!+\!\hat{\lambda}\right) \!\mu_{i 1}\!-\!\hat{\eta}(1\!-\!a_i(t))^2\!\left(\lambda_{i}\!+\!\hat{\lambda}\right) \mu_{i 2}\!+\!\pi_{i}(t)\!\left(\!m\sqrt{Z(t)}\!+\!\varphi_i(t)\!\right)\!\Sigma(t)\right]\!\mathrm{d} t\\
		&-\frac{\theta_{i}}{n}\sum_{k\neq i}\left[\eta_{k}\!\left(\!\lambda_{k}\!+\!\hat{\lambda}\right)\! \mu_{k 1}\!-\!\hat{\eta}(1\!-\!a_k^*(t))^2\!\left(\lambda_{k}\!+\!\hat{\lambda}\right) \mu_{k 2}+\pi_{k}^*(t)\!\left(\!m\sqrt{Z(t)}\!+\!\varphi_i(t)\!\right)\!\Sigma(t)\right]\!\mathrm{d} t\\
		&+(1-\frac{\theta_i}{n})a_{i}(t)\left(\sqrt{\hat{\lambda}}\mu_{i1}\mathrm{d} {\tilde{W}}^{ \mathbb{Q}^{\varphi_i,\chi_i,\phi_i,\vartheta_{i}}}(t)+\sqrt{(\hat{\lambda}+\lambda_{i})\mu_{i 2}-\hat{\lambda}\mu_{i 1}^2}\mathrm{d} {\hat{\textbf{W}}}_i^{ \mathbb{Q}^{\varphi_i,\chi_i,\phi_i,\vartheta_{i}}}(t)\right)\\
		&+(1-\frac{\theta_i}{n})a_{i}(t)\left(\sqrt{\hat{\lambda}}\mu_{i1}\phi_{i}(t)\rd t+\sqrt{(\hat{\lambda}+\lambda_{i})\mu_{i 2}-\hat{\lambda}\mu_{i 1}^2}\vartheta_{i,i}(t) \mathrm{d} t \right)\\
		&-\frac{\theta_{i}}{n}\sum_{k\neq i}a_{k}^*(t) \left(\sqrt{\hat{\lambda}}\mu_{k1}\mathrm{d} {\tilde{W}}^{ \mathbb{Q}^{\varphi_i,\chi_i,\phi_i,\vartheta_{i}}}(t)+\sqrt{(\hat{\lambda}+\lambda_{k})\mu_{k 2}-\hat{\lambda}\mu_{k 1}^2}\mathrm{d} {\hat{\textbf{W}}}_k^{ \mathbb{Q}^{\varphi_i,\chi_i,\phi_i,\vartheta_{i}}}(t)\right)\\
		&-\frac{\theta_{i}}{n}\sum_{k\neq i}a_{k}^*(t) \left(\sqrt{\hat{\lambda}}\mu_{k1}\phi_{i}(t)\rd t+\sqrt{(\hat{\lambda}+\lambda_{k})\mu_{k 2}-\hat{\lambda}\mu_{k 1}^2}\vartheta_{i,k}(t) \mathrm{d} t \right)\\
		&+\left((1-\frac{\theta_i}{n})\pi_{i}(t)-\frac{\theta_{i}}{n}\sum_{k\neq i}\pi_{k}^*(t)\right){\Sigma}(t)\rd {W}^{ \mathbb{Q}^{\varphi_i,\chi_i,\phi_i,\vartheta_{i}}}(t).
	\end{aligned}
\end{equation*}
For $ \{(\pi_k,a_k)_{k=1}^n\}\in\mathscr{U} $ and $ \left(\varphi_i,\chi_i,\phi_i,\vartheta_{i}\right)\in\mathscr{A} $, define the infinitesimal generator of the financial system with processes $Y_i$ and $Z$ as follows:
\begin{equation*}
	\begin{aligned}
&\mathcal{A}^{\{(\pi_k,a_k)\}_{k=1}^n,(\varphi_i,\chi_i,\phi_i,\vartheta_{i})}f(t,y,z)\\
=&f_t\!+\!\left[{\kappa}(\bar{Z}\!-\!z)+{\nu}\sqrt{z}({\rho}\varphi_i\!+\!\sqrt{1-{\rho}^2}\chi_i)\right]f_{z}\!+\!\frac{1}{2}\nu^2zf_{zz}\!+\!\nu\sqrt{z}\rho\sigma((1\!-\!\frac{\theta_{i}}{n})\pi_i\!-\!\frac{\theta_{i}}{n}\sum_{k\neq i}\pi_{k}(t))f_{yz}+ryf_y\\
&+(1-\frac{\theta_i}{n})\left[\eta_{i}\left(\lambda_{i}+\hat{\lambda}\right) \mu_{i 1}-\hat{\eta}(1-a_i(t))^2\left(\lambda_{i}+\hat{\lambda}\right) \mu_{i 2}+\pi_{i}(t)\left(m\sqrt{z}+{\varphi_i(t)}\right)\sigma\right] f_y\\
&-\frac{\theta_{i}}{n}\sum_{k\neq i}\left[\eta_{k}\left(\lambda_{k}+\hat{\lambda}\right) \mu_{k 1}-\hat{\eta}(1-a_k(t))^2\left(\lambda_{k}+\hat{\lambda}\right) \mu_{k 2}+\pi_{k}(t)\left(m\sqrt{z}+{\varphi_i(t)}\right)\sigma\right]f_y\\
&+(1-\frac{\theta_i}{n})a_{i}(t)\left(\sqrt{\hat{\lambda}}\mu_{i1}\phi_{i}(t)+\sqrt{(\hat{\lambda}+\lambda_{i})\mu_{i 2}-\hat{\lambda}\mu_{i 1}^2}\vartheta_{i,i}(t)\right)f_y\\
&-\frac{\theta_{i}}{n}\sum_{k\neq i}a_{k}(t) \left(\sqrt{\hat{\lambda}}\mu_{k1}\phi_{i}(t)+\sqrt{(\hat{\lambda}+\lambda_{k})\mu_{k 2}-\hat{\lambda}\mu_{k 1}^2}\vartheta_{i,k}(t) \right)f_y\\
&+\frac{1}{2}(1\!-\!\frac{\theta_i}{n})^2a_{i}^2(t) \left(\hat{\lambda}\!+\!\lambda_{i}\right) \mu_{i 2}f_{yy}\!+\!\frac{\theta_{i}^2}{2n^2}\sum_{k\neq i}(a_{k}(t))^2 \left(\left(\hat{\lambda}\!+\!\lambda_{k}\right) \mu_{k 2}-\hat{\lambda}\mu_{k1}^2\right)f_{yy}\!+\hat{\lambda}\frac{\theta_i^2}{2n^2}\left[\sum_{k\neq i}a_k(t)\mu_{k 1}\right]^2f_{yy}\\
&-\hat{\lambda}\frac{\theta_{i}}{n}(1-\frac{\theta_i}{n})a_i(t)\mu_{i 1}\sum_{k\neq i}a_{k}(t)\mu_{k 1}f_{yy}+\!\frac{1}{2}\sigma^2((1\!-\!\frac{\theta_{i}}{n})\pi_i\!-\!\frac{\theta_{i}}{n}\sum_{k\neq i}\pi_{k}(t))^2f_{yy},
\end{aligned}
\end{equation*} 
{where $ \sigma= a\sqrt{z}+\frac{b}{\sqrt{z}}$.}

Let $ (y,z) $ be  the value of process $ (Y_i(s), Z(s)) $ at time $ t $, denote the value function of agent $i$ at time $t$ by
$$
	V^{(i)}(t,  {y},z)=\sup _{(\pi_{i},a_i)\in\mathscr{U}_i} \inf _{\left(\varphi_i,\chi_i,\phi_i,\vartheta_{i}\right)\in\mathscr{A}} J_{i}\left(\left\{(\pi_{k}^*,a_k^*)_{k\neq i},(\pi_{i},a_i)\right\},\left(\varphi_i,\chi_i,\phi_i,\vartheta_{i}\right),t,y,z\right).
$$
\begin{definition}[HJBI equation]\label{def:hjbi}
For  $ (Y_i(t), Z(t))=(y,z) $, 
denote\begin{equation*}
	\begin{aligned}
	&	\mathcal{H}\left({\{(\pi_k,a_k)_{k=1}^n\},(\varphi,\chi,\phi,\vartheta)},f,(h_j)_{j=1}^4,(t,y,z)\right)\\=&\mathcal{A}^{\{(\pi_k,a_k)_{k=1}^n\},(\varphi,\chi,\phi,\vartheta)}f(t,y,z)+  \frac{\varphi^{2}(t)}{2 h_1(t,y,z)} +\frac{\chi^2(t)}{2 h_2(t,y,z)} +  \frac{\phi^{2}(t)}{2 h_3(t,y,z)} +\frac{\vartheta(t)^T\vartheta(t)}{2 h_4(t,y,z)} 
	\end{aligned}
\end{equation*}
then the HJBI equation of insurer $i$ is
\begin{equation}\label{HJBI}
	\begin{aligned}
		\sup _{(\pi_{i},a_i)\in\mathscr{U}_i} \inf _{\left(\varphi_i,\chi_i,\phi_i,\vartheta_{i}\right)\in\mathscr{A}}&\mathcal{H}\left({\left\{(\pi_{k}^*,a_k^*)_{k\neq i},(\pi_{i},a_i)\right\},(\varphi_i,\chi_i,\phi_i,\vartheta_{i})},V,(\Psi_j^{i})_{j=1}^4,(t,y,z)\right)  =0,
	\end{aligned}
\end{equation}
with boundary condition $V(T, y,z)= U_f\left(y; \delta_{i}\right)=-\frac{1}{\delta_i} e^{-\delta_i y}$.
\end{definition}
In the following, we first solve the HJBI equation \eqref{HJBI}. Then, we verify that the solution of the HJBI equation solves Problem \eqref{oi}.  
 Assume $ v^{(i)}(t,y,z) $ is the solution to the following equation:
\begin{equation}\label{HJBI-v}
	\begin{aligned}
		\sup _{(\pi_{i},a_i)\in\mathscr{U}_i} \inf _{\left(\varphi_i,\chi_i,\phi_i,\vartheta_{i}\right)\in\mathscr{A}}&\mathcal{H}\left({\left\{(\pi_{k}^*,a_k^*)_{k\neq i},(\pi_{i},a_i)\right\},(\varphi_i,\chi_i,\phi_i,\vartheta_{i})},V,(\beta_{i,j}V)_{j=1}^4,(t,y,z)\right)  =0,
	\end{aligned}
\end{equation}
{We call equation (\ref{HJBI-v}) the  HJBI equation and $ v^{(i)}(t,y,z) $ the candidate value function.}
And let the penalty terms that are scaled by $ \Psi_1^{i} $, $ \Psi^{i}_2 $, $ \Psi_3^{i} $ and $ \Psi^{i}_4 $  hereafter be \begin{equation*}
\Psi_j^{i}=\frac{\psi_{i,j}}{-\delta_i	v^{(i)}}:=\frac{\beta_{i,j}}{	v^{(i)}},
\footnote{
  
In our framework, the explicit solution is derived by linking the penalty term to the value function. This approach was first introduced in \cite{maenhout2004robust, maenhout2006robust}, where the robustness penalty depends simultaneously on the relative entropy and the continuation value. On one hand, this formulation introduces homogeneity and improves mathematical tractability. On the other hand, the results in \cite{maenhout2004robust} provide meaningful economic insights: a preference for robustness can significantly reduce the portfolio’s demand for equities and increase the demand for hedging-type assets. Since then, the assumption of homothetic preferences for the penalty term has been widely adopted in robust portfolio selection for its mathematical tractability, as seen in \cite{uppal2003model}, \cite{branger2013robust}, \cite{escobar2015robust}, \cite{kikuchi2024age}, and others.
}
\end{equation*}where $\psi_{i,j}={-\delta_i}{\beta_{i,j}}>0,j=1,2,3,4,$  ($\beta_{i,j}<0$) represent the ambiguity aversion coefficients of insurer $i$. Insurer $i$ is more ambiguity aversion for a larger $\psi_{i,j}$.  $\psi_{i,1}$, $\psi_{i,2}$, $\psi_{i,3}$ and $\psi_{i,4}$ represent ambiguity aversion coefficients of insurer $i$  over the equity risk, volatility risk, common insurance risk and idiosyncratic insurance risk, respectively.

{Then we see the  candidate value function $ v^{(i)} $ is also the solution to the HJBI equation \eqref{HJBI}}. We will show that under specific conditions, $ v^{(i)}(t,  {y},z) $ equals to the value function $ V^{(i)}(t,  {y},z)  $, which also indicates that $ \Psi_j^{i} ,(j=1,2,3,4) $ are inversely proportional to the value function.

For  $ (Y_i(t), Z(t))=(y,z) $, {let}
\begin{equation}\label{HJBI-opt}
	\begin{aligned}
		(\pi^\circ_i,a^\circ_i)=\arg	\sup _{(\pi_i,a_i)\in\mathscr{U}_i} \inf _{(\varphi_i,\chi_i,\phi_i,\vartheta_{i})\in\mathscr{A}}&\mathcal{H}\left({\left\{(\pi_{k}^*,a_k^*)_{k\neq i},(\pi_{i},a_i)\right\},(\varphi_i,\chi_i,\phi_i,\vartheta_{i})},v^{(i)},(\Psi_j^{i})_{j=1}^4,(t,y,z)\right),\\
		( \varphi^\circ_{i},\chi^\circ_{i},\phi_i^\circ,\vartheta_{i}^\circ)=\arg	 \inf _{(\varphi_i,\chi_i,\phi_i,\vartheta_{i})\in\mathscr{A}}&\mathcal{H}\left({\left\{(\pi_{k}^*,a_k^*)_{k\neq i},(\pi_{i}^\circ,a_i^\circ)\right\},(\varphi_i,\chi_i,\phi_i,\vartheta_{i})},v^{(i)},(\Psi_j^{i})_{j=1}^4,(t,y,z)\right),
	\end{aligned}
\end{equation}
{which we call the candidate robust optimal strategies of insurer $ i $ and the worst-case measures associated with the candidate  strategies  of insurer $ i $ respectively.}

Initially, we will derive the candidate robust optimal strategies, worst-case measures, and candidate value function in Theorem \ref{solution-HJBI}. Subsequently, we will demonstrate in Theorem \ref{Verification} that the candidate robust optimal strategies are indeed the robust optimal strategies.

We first list some compatible conditions about the parameters.
\begin{itemize}
	\item [\textbf{Condition (I)}]\begin{equation*}
		\sup_{1\leqslant i\leqslant n}\frac{\beta_{i,1}^2}{\left(1-\beta_{i,1}\right)^2}m^2<\frac{\kappa^2}{2\nu^2}.
	\end{equation*}
	\item [\textbf{Condition (II)}]\begin{equation*}
		\sup_{1\leqslant i\leqslant n}c_{i,1}^2\frac{\left(1-e^{c_{i,2} T}\right)^2}{\left(1 + c_{i,3}e^{c_{i,2} T}\right)^2}\nu^2\beta_{i,2}^2(1-\rho^2)<\frac{\kappa^2}{2\nu^2},
	\end{equation*}
	
\end{itemize}where $ c_{i,1}=\frac{{\kappa}+m\nu\rho +\sqrt{({\kappa}+m\nu\rho )^2+\frac{1-\beta_{i,2}}{1-\beta_{i,1}}\nu^2(1-\rho^2)m^2}}{\nu^2(1-\rho^2)(1-\beta_{i,2})}$, $c_{i,3}=2(1-\beta_{i,1})({\kappa}+m\nu\rho )c_{i,1}+1$, $c_{i,2}=\frac{c_{i,3}+1}{2(1-\beta_{i,1})c_{i,1}}$, $\beta_{i,1}=-\frac{\psi_{i,1}}{\delta_i}$, $\beta_{i,2}=-\frac{\psi_{i,2}}{\delta_i} $.

{ These conditions help ensure that the candidate strategies in Theorem \ref{solution-HJBI} are admissible, as demonstrated in the proof of Theorem \ref{solution-HJBI} in Appendix \ref{proof-solution-HJBI}.}

\begin{theorem}\label{solution-HJBI}
$ \forall 1\leqslant i\leqslant n $, if the parameters satisfy the compatible conditions (I) and (II),  
 when other insurer $k$ ($k\neq i$) adopts strategy $\left({\pi}^*_k, {a}^*_k\right)$, the  candidate robust optimal investment and reinsurance strategies are
 	\begin{equation}\label{equ:pia1}
		\left\{\begin{aligned}
			&{\pi}^\circ_i(t)=\frac{\theta_i}{n-\theta_i}\sum_{k\neq i}\pi^*_k(t)+\frac{n}{n-\theta_i}(\frac{m}{1-\beta_{i,1}}+\nu\rho h_i(t))\frac{\sqrt{Z(t)}}{ \delta_{i}g_i(t)\Sigma(t)},\\
			&a_i^\circ(t)=\left(\frac{Q_i(t)}{R_i(t)}\frac{1}{n}\sum_{k\neq i}\mu_{k 1}a_k^*(t)+\frac{P_i(t)}{R_i(t)}\right)\wedge1,
		\end{aligned}\right.
	\end{equation}
   and the {associated} worst-case measures are given as follows,
	\begin{equation*}
	\left\{\begin{aligned}
		&{\varphi^\circ_{i}(t)}=\frac{\beta_{i,1}}{1-\beta_{i,1}}m\sqrt{Z(t)},\\
		&{\chi^\circ_{i}(t)}=-\nu\sqrt{1-\rho^2} \beta_{i,2}h_i(t)\sqrt{Z(t)},\\
		&\phi^\circ_{{i}}(t)=\sqrt{\hat{\lambda}}\left[(1-\frac{\theta_i}{n})\mu_{i1}a^\circ_{i}(t)-\frac{\theta_{i}}{n}\sum_{k\neq i}\mu_{k1}a_{k}^*(t)\right]\delta_{i}g_i(t){\beta_{i,3}},\\
		&\vartheta^\circ_{{i,i}}(t)=(1-\frac{\theta_i}{n})a^\circ_{i}(t)\sqrt{(\hat{\lambda}+\lambda_{i})\mu_{i 2}-\hat{\lambda}\mu_{i 1}^2}\delta_{i}g_i(t){\beta_{i,4}},\\
		&\vartheta^\circ_{{i,k}}(t)=-\frac{\theta_i}{n}a_{k}^*(t)\sqrt{(\hat{\lambda}+\lambda_{k})\mu_{k 2}-\hat{\lambda}\mu_{k 1}^2}\delta_{i}g_i(t){\beta_{i,4}},~k\neq i
	\end{aligned}\right.
\end{equation*}
where 
\begin{equation*}
	\begin{aligned}
		&	P_i(t)=\frac{2\hat{\eta}}{ \delta_{i}g_i(t)}(\lambda_i+\hat{\lambda})\mu_{i 2},\quad Q_i(t)=\hat{\lambda}\theta_i\mu_{i 1}(1-\beta_{i,3}),\\& R_i(t)=\frac{2}{\delta_{i}g_i(t)}\hat{\eta}(\lambda_i+\hat{\lambda})\mu_{i 2}+ (1-\frac{\theta_i}{n})\left[(\lambda_i+\hat{\lambda})\mu_{i2}(1-{\beta_{i,4}})+\hat{\lambda}\mu_{i 1}^2({\beta_{i,4}}-{\beta_{i,3}})\right].
\end{aligned}	\end{equation*}
And the candidate value function of  insurer $i$ $ v^{(i)}(t,y,z) \in C^{1,2,2}([0,T]\times\mathbb{R}^2_+)$  is given by \begin{equation*}
	v^{(i)}(t, y,z)= -\frac{1}{\delta_i} \exp\left(f_i(t)-\delta_i yg_i(t)+h_i(t)z\right),
\end{equation*}where 
\begin{equation}\label{fgh}
	\left\{\begin{aligned}
		g_i(t)=&e^{r(T-t)},\\
		h_i(t)=&c_{i,1}\frac{e^{c_{i,2} t}-e^{c_{i,2} T}}{e^{c_{i,2} t} + c_{i,3}e^{c_{i,2} T}},\\
		f_i(t)=&\int_{t}^{T}\!\left\{{\kappa}\bar{Z}h_i(s)-(1-\frac{\theta_i}{n})\left[\eta_{i}\left(\lambda_{i}+\hat{\lambda}\right) \mu_{i 1}-\hat{\eta}(1-a^\circ_i(s))^2\left(\lambda_{i}+\hat{\lambda}\right) \mu_{i 2}\right] \delta_{i}g_i(s)\right.\\&\quad+\frac{\theta_{i}}{n}\sum_{k\neq i}\left[\eta_{k}\left(\lambda_{k}+\hat{\lambda}\right) \mu_{k 1}-\hat{\eta}(1-a^*_k(s))^2\left(\lambda_{k}+\hat{\lambda}\right) \mu_{k 2}\right]\delta_{i}g_i(s)\\
		&\quad+\frac{1}{2}(1-\frac{\theta_i}{n})^2(a^\circ_{i}(s))^2((\hat{\lambda}+\lambda_{i})\mu_{i 2}(1-{\beta_{i,4}})+\hat{\lambda}\mu_{i 1}^2({\beta_{i,4}}-{\beta_{i,3}}))\delta_{i}^2e^{2r(T-s)}\\
		&\quad+\frac{\theta_i^2}{2n^2}\sum_{k\neq i}(a_{k}^*(s))^2((\hat{\lambda}+\lambda_{k})\mu_{k 2}-\hat{\lambda}\mu_{k 1}^2)(1-{\beta_{i,4}})\delta_{i}^2e^{2r(T-s)} \\
		&\quad+\hat{\lambda}\frac{\theta_i^2}{2n^2}\left[\sum_{k\neq i}a^*_k(s) \mu_{k 1}\right]^2 \delta_{i}^2e^{2r(T-s)}(1-{\beta_{i,3}})\\
		&\quad\left.-\hat{\lambda}\frac{\theta_{i}}{n}(1-\frac{\theta_i}{n})a^\circ_i(s)\sum_{k\neq i}a^*_{k}(s)\mu_{i 1}\mu_{k 1}\delta_{i}^2e^{2r(T-s)}(1-{\beta_{i,3}})\right\}\rd s,
	\end{aligned}\right.
\end{equation}
and $ c_{i,1}=\frac{{\kappa}+m\nu\rho +\sqrt{({\kappa}+m\nu\rho )^2+\frac{1-\beta_{i,2}}{1-\beta_{i,1}}\nu^2(1-\rho^2)m^2}}{\nu^2(1-\rho^2)(1-\beta_{i,2})}$, $c_{i,3}=2(1-\beta_{i,1})({\kappa}+m\nu\rho )c_{i,1}+1$, $c_{i,2}=\frac{c_{i,3}+1}{2(1-\beta_{i,1})c_{i,1}} $.
\end{theorem}
\begin{proof}
See Appendix \ref{proof-solution-HJBI}.
\end{proof}
\eqref{equ:pia1} shows that the candidate robust optimal investment and reinsurance strategies are all composed of two parts: one part associated with others' strategies, and one part to hedge insurance risk or financial risk. When the insurer is not concerned with others' wealth, i.e., $\theta_i=0$, the first part disappears and the insurer is not influenced by others' strategies.  We can also see that \begin{equation*}
	\esssup_{0\leqslant t\leqslant T}v^{(i)}(t,Y(t),Z(t))<\infty
\end{equation*}

Theorem  \ref{solution-HJBI} presents the candidate robust optimal strategies, the associated worst-case measures and the candidate value function.
Next, we should verify that $(\pi^*_i,a^*_i)=(\pi^\circ_i,a^\circ_i)$, $ \left(\varphi_{i}^*, \chi_{i}^*,\phi_i^*,\vartheta_{i}^*\right)= \left((\varphi_i^\circ,\chi_i^\circ,\phi_i^\circ,\vartheta_{i}^\circ)\right)$ and $V^{(i)}(t, y,z)=v^{(i)}(t, y,z)$, i.e., the candidate robust optimal strategies, the associated worst-case measures and candidate value function are the optimal response strategies, the associated worst-case measures and value function indeed respectively, which will be shown in Theorem \ref{Verification}.

Before that, we present the following property of $(\pi^\circ_{i},a^\circ_i)$. 

\begin{proposition}\label{pi-circ} For $
	\Psi_{j}^i=\frac{\beta_{i,j}}{v^{(i)}}
	$,  $ (Y_i(t), Z(t))=(y,z) $, and $(\pi^\circ_i,a^\circ_i)$, $ \left(\varphi_{i}^\circ, \chi_{i}^\circ,\phi_i^\circ,\vartheta_{i}^\circ\right)$  defined by \eqref{HJBI-opt} solves the following problem
\begin{equation*}
(\pi^\circ_{i},a^\circ_i)=\arg\sup _{(\pi_i,a_i)\in\mathscr{U}_i}\mathcal{H}\left({\left\{(\pi_{k}^*,a_k^*)_{k\neq i},(\pi_{i},a_i)\right\},(\varphi_i^\circ,\chi_i^\circ,\phi_i^\circ,\vartheta_{i}^\circ)},v^{(i)},(\Psi_j^{i})_{j=1}^4,(t,y,z)\right).
\end{equation*}
\end{proposition}
\begin{proof}
See Appendix \ref{proof-pi-circ}.
\end{proof}
Proposition \ref{pi-circ} shows that $(\pi^\circ_{i},a^\circ_i)$  also solves the optimization problem when the probability measure  is given in advance by $ \left(\varphi_{i}^\circ, \chi_{i}^\circ,\phi_i^\circ,\vartheta_{i}^\circ\right)$. 
\begin{theorem}[Verification theorem]\label{Verification}
	For $v^{(i)}\in C^{1,2,2}([0,T]\times\mathbb{R}^2_+) $ solving  \eqref{HJBI-v}, $ (\pi^\circ_i,a^\circ_i) $, $(\varphi_i^\circ,\chi_i^\circ,\phi_i^\circ,\vartheta_{i}^\circ)$  defined by \eqref{HJBI-opt}, if the parameters satisfy the compatible conditions (I) and (II), then equations $(\pi^*_i,a^*_i)=(\pi^\circ_i,a^\circ_i)$, $ \left(\varphi_{i}^*, \chi_{i}^*,\phi_i^*,\vartheta_{i}^*\right)= \left((\varphi_i^\circ,\chi_i^\circ,\phi_i^\circ,\vartheta_{i}^\circ)\right)$, $ V^{(i)}(t, y,z)=v^{(i)}(t, y,z) $ hold $  \forall 1\leqslant i\leqslant n $. 
\end{theorem}\begin{proof}
See Appendix \ref{proof-verification}.\end{proof}

Based on Theorem \ref{solution-HJBI} and Theorem \ref{Verification}, we obtain the insurer $i$'s robust optimal response strategy. Next, we derive the robust equilibrium strategies  explicitly. Replacing $(\pi^\circ_i,a^\circ_i) $ by $(\pi^*_i,a^*_i)$ in \eqref{equ:pia1} for $ 1\leqslant i\leqslant n $, we obtain $n$ equations of $(\pi^*_i,a^*_i)$ expressed by $(\pi^*_k,a^*_k), k\neq i  $. Solving these $n$ equations simultaneously, we can derive the robust equilibrium reinsurance and investment strategies in the $n$-insurer game. The derivations are  simple and we omit them. The results are summarized in the following theorem.
\begin{theorem}\label{thm:result}
If $n\neq\sum_{k=1}^n\theta_k$, the robust $n$-insurer equilibrium  exists. The equilibrium investment strategies are given by $ \pi^*_i$, $ \forall 1\leqslant i\leqslant n $ are given by  
\begin{equation}\label{n-opt}
	\pi^*_i(t)=\left[\theta_i\frac{\sum_{k=1}^n\frac{\frac{m}{1-\beta_{k,1}}+\nu\rho h_k(t)}{\delta_{k}g_k(t)}}{n-\sum_{k=1}^n\theta_k}+\frac{\frac{m}{1-\beta_{i,1}}+\nu\rho h_i(t)}{\delta_{i}g_i(t)}\right]\frac{{Z(t)}}{a{Z(t)}+b},
\end{equation}the equilibrium reinsurance proportions $ a^*_i $, $ \forall 1\leqslant i\leqslant n $, are  determined by 
\begin{equation}\label{a_star}
	a_i^*(t)=\left(\frac{Q_i(t)}{R_i(t)}\frac{1}{n}\sum_{k\neq i}\mu_{k 1}a_k^*(t)+\frac{P_i(t)}{R_i(t)}\right)\wedge1,
\end{equation}
and the associated worst-case model uncertainty functions $ \left(\varphi_{i}^*, \chi_{i}^*,\phi_i^*,\vartheta_{i}^*\right)$,  $ \forall 1\leqslant i\leqslant n $, are given by 
\begin{equation*}
	\left\{\begin{aligned}
		&{\varphi^*_{i}(t)}=\frac{\beta_{i,1}}{1-\beta_{i,1}}m\sqrt{Z(t)},\\
		&{\chi^*_{i}(t)}=-\nu\sqrt{1-\rho^2} \beta_{i,2}h_i(t)\sqrt{Z(t)}\\
		&\phi^*_{{i}}(t)=\sqrt{\hat{\lambda}}\left[(1-\frac{\theta_i}{n})\mu_{i1}a^*_{i}(t)-\frac{\theta_{i}}{n}\sum_{k\neq i}\mu_{k1}a_{k}^*(t)\right]\delta_{i}g_i(t){\beta_{i,3}},\\
		&\vartheta^*_{{i,i}}(t)=(1-\frac{\theta_i}{n})a^*_{i}(t)\sqrt{(\hat{\lambda}+\lambda_{i})\mu_{i 2}-\hat{\lambda}\mu_{i 1}^2}\delta_{i}g_i(t){\beta_{i,4}},\\
		&\vartheta^*_{{i,k}}(t)=-\frac{\theta_i}{n}a_{k}^*(t)\sqrt{(\hat{\lambda}+\lambda_{k})\mu_{k 2}-\hat{\lambda}\mu_{k 1}^2}\delta_{i}g_i(t){\beta_{i,4}},~k\neq i.
	\end{aligned}\right.
\end{equation*}
where $ f_i(t) $, $ g_i(t) $ and  $ h_i(t) $ are given by \eqref{fgh}.

Moreover, if $n=\sum_{k=1}^n\theta_k$, the robust $n$-insurer equilibrium does not exist.
\end{theorem}

The robust investment strategy \eqref{n-opt} can be represented as
\begin{equation}\label{dn-opt}
	\pi^*_i(t)=\frac{e^{-r(T-t)}{Z(t)}}{a{Z(t)}+b}
	\left[\frac{m}{\delta_i(1-\beta_{i,1})}+\frac{\nu\rho h_i(t)}{\delta_i}
	+\frac{\theta_i}{n-\sum_{k=1}^n\theta_k}\sum_{k=1}^n\frac{m}{\delta_{k}(1-\beta_{k,1})}+\frac{\theta_i}{n-\sum_{k=1}^n\theta_k}\sum_{k=1}^n\frac{\nu\rho h_k(t)}{\delta_{k}}\right].	
\end{equation}

In \eqref{dn-opt}, the robust investment strategy manifests as a composite of four distinct components. The first component represents a myopic demand characterized by constant volatility. The second component is introduced to hedge against volatility risk. Both these components are optimal for the robust problem in the absence of any relative performance concerns and solely rely on individual risk aversion and ambiguity aversion coefficients. The third and fourth terms correspond to the myopic and hedging demands induced by competition, respectively.
 In the Black-Scholes financial market, \cite{lacker2019mean} and \cite{MFE2020} identify the first and third components. Here, each insurer not only hedges against exogenous changes in volatility but also responds to the endogenous hedging activities of their peers. Consequently, competition introduces a new hedging term in the robust equilibrium strategy. As shown in \eqref{dn-opt}, the third and fourth components exhibit non-linear dependence on the competition, risk aversion, and ambiguity aversion coefficients of all insurers. Naturally, these components vanish in the absence of competition. We see from \eqref{dn-opt} that when $\rho=0$, equity risk and volatility risk are independent and the hedging demands (second and fourth components) become zero. 

The robust reinsurance proportions are implicitly determined by \eqref{a_star}. As revealed in \eqref{a_star}, the robust equilibrium reinsurance proportions consist of two parts: one part disregarding competition and another part responsive to the activities of other insurers. Importantly, both these parts are influenced by the risk aversion and ambiguity aversion coefficients.

\begin{remark}\label{Randomness}
	We also see that the randomness of $ \pi^*_{i}(t) $ only depends on $ \frac{{Z(t)}}{a{Z(t)}+b} $,  the randomness of $ {\varphi^*_{i}(t)} $ and $ {\chi^*_{i}(t)} $ depends on $ Z(t) $ and $ {\varphi^*_{i}(t)} $  is  proportional to $ \sqrt{Z(t)} $.
  $ a_i^*(t) $ and $ \phi_i^*(t),\vartheta_{i}^*(t) $ are deterministic.	However, $ a_i^*(t) $ is implicit and $ \phi_i^*(t)$ and $\vartheta_{i}^*(t) $ rely on $ a_i^*(t) $,  we can get the approximate  numerical solution by the following procedure.
	
First solve \eqref{a_i}:
	\begin{equation}\label{a_i}
		\check{a}_i(t)=\frac{Q_i(t)}{R_i(t)}\frac{1}{n}\sum_{k\neq i}\mu_{k 1}\check{a}_k(t)+\frac{P_i(t)}{R_i(t)},\quad \forall1\leqslant i\leqslant n, t\in[0,T].
	\end{equation}
	\eqref{a_i} is equivalent to 
	\begin{equation}\label{a_i_2}
		\mu_{i 1}	\check{a}_i(t)=\frac{n\mu_{i 1}{Q_i(t)}}{{\mu_{i1}Q_i(t)}+nR_i(t)}\frac{1}{n}\sum_{k=1}^n\mu_{k 1}\check{a}_k(t)+\frac{n\mu_{i 1}{P_i(t)}}{{\mu_{i1}Q_i(t)}+nR_i(t)},\quad \forall1\leqslant i\leqslant n, t\in[0,T].
	\end{equation}
{As $ \beta_{i,3}<0, \beta_{i,4}<0, \mu_{i2}\geqslant\mu_{i1}^2$, we see 
\begin{equation*}
	R_i(t)>(1-\frac{\theta_i}{n})\hat{\lambda}\mu_{i1}^2(1-\beta_{i3}),
\end{equation*}
which indicates that \begin{equation*}
n-\sum_{k=1}^n	\frac{n\mu_{k 1}{Q_k(t)}}{{\mu_{k1}Q_k(t)}+nR_k(t)}>n-\sum_{k=1}^n\frac{n}{1+n-\theta_k}\geqslant n-\sum_{k=1}^n1=0.
\end{equation*}
}

	Then \eqref{a_i_2} admits a unique positive solution,  given by
\begin{equation*}
	\check{a}_i(t)=\frac{n{Q_i(t)}}{{\mu_{i1}Q_i(t)}+nR_i(t)}\frac{\sum_{k=1}^n\frac{n\mu_{k 1}{P_k(t)}}{{\mu_{k1}Q_k(t)}+nR_k(t)}}{n-\sum_{k=1}^n\frac{n\mu_{k 1}{Q_k(t)}}{{\mu_{k1}Q_k(t)}+nR_k(t)}}+\frac{n{P_i(t)}}{{\mu_{i1}Q_i(t)}+nR_i(t)}, \forall1\leqslant i\leqslant n, t\in[0,T].
\end{equation*}
	$ \forall1\leqslant i\leqslant n, t\in[0,T] $, denote $$ a_i^{(0)}(t)=\check{a}_i(t),\quad a_i^{(l+1)}(t)=\left(\frac{Q_i(t)}{R_i(t)}\frac{1}{n}\sum_{k\neq i}\mu_{k 1}a_k^{(l)}(t)+\frac{P_i(t)}{R_i(t)}\right)\wedge1,\quad\forall l\geqslant1,$$
	then we see the sequence $ \{a_i^{(l)}(t)\}_{l=1}^{\infty} $ is nonegative and monotone nonincreasing, so the sequence converges based on  monotone convergence theorem. Let $\displaystyle\mathring{a}_i(t)=\lim_{l\to\infty}a_i^{(l)}(t)  $, then we see 
	\begin{equation*}
		\mathring{a}_i(t)=\left(\frac{Q_i(t)}{R_i(t)}\frac{1}{n}\sum_{k\neq i}\mu_{k 1}\mathring{a}_k(t)+\frac{P_i(t)}{R_i(t)}\right)\wedge1,\quad \forall1\leqslant i\leqslant n, t\in[0,T] 
	\end{equation*}
	thus $a_i^*(t)= \mathring{a}_i(t), \forall1\leqslant i\leqslant n, t\in[0,T] $.
	  
	  In practice, when there exists $ j $ such that $ a_i^{(j+1)}(t) = a_i^{(j)}(t)$ or $ |a_i^{(j+1)}(t) -a_i^{(j)}(t)| $ is sufficiently small $\forall 1\leqslant i\leqslant n $, then we can use $ a_i^{(j+1)}(t) $ as a   numerical approximation for $ a^*_i(t) $.
\end{remark}

{\begin{remark}
		Denote $ \bar{\mu}_1=\frac{1}{n}\sum_{k=1}^n\mu_{k 1} $, if $ \frac{n-1}{n}\mu_{i2}\geqslant\bar{\mu}_1\mu_{ i 1} $ and $ \beta_{i 3}\geqslant \beta_{i 4}$ for $ i=1,2,\cdots, n $, then \begin{equation*}
			\frac{Q_i(t)}{R_i(t)}\frac{1}{n}\sum_{k\neq i}\mu_{k 1}{a}^*_k(t)+\frac{P_i(t)}{R_i(t)}\leqslant\frac{Q_i(t)}{R_i(t)}\frac{1}{n}\sum_{k\neq i}\mu_{k 1}+\frac{P_i(t)}{R_i(t)}<\frac{Q_i(t)\bar{\mu}_1+P_i(t)}{R_i(t)}\leqslant 1
		\end{equation*}which means that \eqref{equ:pia1} becomes \begin{equation}\label{equ:pia22}
		\left\{\begin{aligned}
			&{\pi}^\circ_i(t)=\frac{\theta_i}{n-\theta_i}\sum_{k\neq i}\pi^*_k(t)+\frac{n}{n-\theta_i}(\frac{m}{1-\beta_{i,1}}+\nu\rho h_i(t))\frac{\sqrt{Z(t)}}{ \delta_{i}g_i(t)\Sigma(t)},\\
			&a_i^\circ(t)=\frac{Q_i(t)}{R_i(t)}\frac{1}{n}\sum_{k\neq i}\mu_{k 1}a_k^*(t)+\frac{P_i(t)}{R_i(t)},
		\end{aligned}\right.
	\end{equation}and  \eqref{a_star} will be given by \begin{equation}\label{a_star2}
	a_i^*(t)=\frac{Q_i(t)}{R_i(t)}\frac{1}{n}\sum_{k\neq i}\mu_{k 1}a_k^*(t)+\frac{P_i(t)}{R_i(t)}, \forall1\leqslant i\leqslant n, t\in[0,T],
\end{equation}then we can get the explicit form of $ {a}^*_i $ given by: \begin{equation*}
{a}^*_i(t)=\frac{n{Q_i(t)}}{{\mu_{i1}Q_i(t)}+nR_i(t)}\frac{\sum_{k=1}^n\frac{n\mu_{k 1}{P_k(t)}}{{\mu_{k1}Q_k(t)}+nR_k(t)}}{n-\sum_{k=1}^n\frac{n\mu_{k 1}{Q_k(t)}}{{\mu_{k1}Q_k(t)}+nR_k(t)}}+\frac{n{P_i(t)}}{{\mu_{i1}Q_i(t)}+nR_i(t)}, \forall1\leqslant i\leqslant n, t\in[0,T].
\end{equation*}
\end{remark}}

\section{\bf  Robust mean-field game}
In this section, we study the limit of the robust $ n $-insurer  game as $ n\to\infty $. In the case of robust $ n $-insurer  game, we can use a type vector $ U=(x^0,\lambda,\mu_{1},\mu_{ 2},\eta,\theta,\delta,\psi_{1},\psi_{2},\psi_{3},\psi_{4}) $ from value space $ \mathbb{U} $ to identify the preference or information of the insurer, the type vectors of all these insurers also induce an empirical measure, which is the probability measure on
$\mathbb{U}$
given by
$$
\mathcal{D}_n(A)=\frac{1}{n} \sum_{i=1}^n 1_A\left(u_i\right), \text { for Borel sets } A \subset\mathbb{U}.
$$
Assume  that $ \mathcal{D}_n $ has a weak limit $ \mathcal{D} $, in the sense that $ \int_{\mathbb{U}} f\rd \mathcal{D}_n\to\int_{\mathbb{U}} f\rd \mathcal{D}$ for every bounded continuous function $ f $. As we have obtained the robust equilibrium of the robust $ n $-insurer game, we now search for the robust mean-field equilibrium given that the type distribution is $ \mathcal{D} $. 

In the light of \eqref{n-opt} and \eqref{a_star}, we should expect that
\begin{equation}\label{m-pi}
	\lim_{n\to\infty}\pi^*_i(t)=\left[\theta_i\frac{M_1}{1-\bar{\theta}}+\frac{\frac{m}{1-\beta_{i,1}}+\nu\rho h_i(t)}{\delta_{i}g_i(t)}\right]\frac{{Z(t)}}{a{Z(t)}+b},
\end{equation}
where \begin{equation*}
M_1=\lim_{n\to\infty}\frac{1}{n}\sum_{k=1}^n\frac{\frac{m}{1-\beta_{k,1}}+\nu\rho h_k(t)}{\delta_{k}g_k(t)}=\mathbb{E}[(\frac{m}{1-\beta_{1}}+\nu\rho h_{U}(t))\frac{1}{\delta }]\frac{1}{e^{r(T-t)}},~~\bar{\theta}=\lim_{n\to\infty}\frac{1}{n}\sum_{k=1}^n\theta_k=\mathbb{E}[\theta],
\end{equation*}
and \begin{equation}\label{m-a}
	\lim_{n\to\infty}a_i^*(t)=\left(\frac{Q_i(t)}{R_i(t)}M_2+\frac{P_i(t)}{R_i(t)}\right)\wedge1,
\end{equation}
where \begin{equation*}
M_2=\lim_{n\to\infty}\frac{1}{n}\sum_{k=1}^n\mu_{k 1}a_k^*(t)
\end{equation*}
is solved by \begin{equation*}
M_2=	\lim_{n\to\infty}\frac{1}{n}\sum_{i=1}^n\mu_{i 1}a_i^*(t)=	\lim_{n\to\infty}\frac{1}{n}\sum_{i=1}^n\mu_{i 1}\left(\frac{Q_i(t)}{R_i(t)}M_2+\frac{P_i(t)}{R_i(t)}\right)\wedge1,
\end{equation*}
\eqref{m-pi} and \eqref{m-a} represent the robust equilibrium reinsurance and investment strategies with infinite insurers in form. Next, we formulate the mean-field game and verify that \eqref{m-pi} and \eqref{m-a} solves the mean-field game.

We assume that the distribution $ \mathcal{D} $ is a discrete  distribution and $\left(\Omega, \mathcal{F}, \mathbb{P}\right)$ supports  a random vector $ U=(x^0,\lambda,\mu_{1},\mu_{ 2},\eta,\theta,\delta,\psi_{1},\psi_{2},\psi_{3},\psi_{4}):\Omega\to\mathbb{U} $ with distribution $ \mathcal{D} $ and all the involved Brownian motions below. { Assume that $U$ is independent of all the Brownian motions below}. $ \hat{\mathbb{F}}=\{\hat{\mathcal{F}}_t,t\in[0,T]\} $ is the smallest filtration satisfying the usual assumptions for which $U$ is $ \hat{\mathcal{F}}_0 $-measurable and all the involved Brownian motions below are adapted. Let $ \mathbb{F}^{\tilde{W},{W},{B}}=\{\mathcal{F}_t^{\tilde{W},{W},{B}},t\in[0,T]\} $ denote the natural filtration generated by the Brownian motions $ \tilde{W},{W}$ and ${B} $.

 Under the measure of  the insurer with type $ u_0=(x_0^0,\lambda_0,\mu_{ 01},\mu_{ 02},\eta_0,\theta_0,\delta_0,\psi_{01},\psi_{02},\psi_{03},\psi_{04}) $,
the surplus process of a representative insurer with  type $ U\sim\mathcal{D} $ involves as:
\begin{equation*}
	\begin{aligned}
		\mathrm{d} X_{U}(t)
		=&\left[r X_{U}(t)+\eta\left(\lambda+\hat{\lambda}\right) \mu_{ 1}-\hat{\eta}(1-a_{U}(t))^2\left(\lambda+\hat{\lambda}\right) \mu_{ 2}+\pi_{U}(t)m({a}{Z(t)}+{{b}})\right] \mathrm{d} t\\&+a_{U}(t)\left(\sqrt{\hat{\lambda}}\mu_{1}(\rd  \tilde{W}^{\mathbb{Q}^{u_0}}(t)+\phi_{u_0}(t)\rd t)+\sqrt{(\hat{\lambda}+\lambda)\mu_{ 2}-\hat{\lambda}\mu_{ 1}^2}(\rd \hat{W}_U^{\mathbb{Q}^{u_0}}(t)+\vartheta_{u_0}^U(t)\rd t)\right)\\&+\pi_{U}(t){\Sigma}(t)(\rd {W}^{\mathbb{Q}^{u_0}}(t)+\varphi_{u_0}(t)\rd t).
	\end{aligned}
\end{equation*}
\begin{equation*}
	\rd Z(t)={\kappa}(\bar{Z}-Z(t))\rd t+{\nu}\sqrt{Z(t)}\left[{\rho}(\rd {W}^{\mathbb{Q}^{u_0}}(t)+\varphi_{u_0}(t)\rd t)+\sqrt{1-{\rho}^2}(\rd{B}^{\mathbb{Q}^{u_0}}(t)+\chi_{u_0}(t)\rd t)\right]
\end{equation*}
Similar to the robust $ n $-insurer game, we first define the admissible set of the strategies and probability measure transformation functions. Some notations are similar to the last section. Define
\begin{equation}
	\begin{aligned}
		\mathscr{U}=&\Big\{(\pi_{U},a_{U}):U\sim\mathcal{D}, (\pi_{U},a_{U}) \text{ is progressively measurable w.r.t. }\hat{\mathbb{F}}\text{ and }0\leqslant a_{U}(t)\leqslant1,\\
		&\exists \mathcal{C}_{U}\in\mathbb{R}_+, \sup_{U}\mathcal{C}_{U}<\infty, \text{ such that } \pi_{U}(t)=\ell_{U}(t)\frac{Z(t)}{aZ(t)+b},|\ell_{U}(t)|\leqslant \mathcal{C}_{U}, \forall t\in[0,T]\Big\},\\
		\mathscr{U}_{u_0}=&\Big\{(\pi_{u_0},a_{u_0}):(\pi_{u_0},a_{u_0}) \text{ is progressively measurable w.r.t. }{\mathbb{F}}\text{ and }0\leqslant a_{u_0}(t)\leqslant1,\\
		&\exists \mathcal{C}_{u_0}\in\mathbb{R}_+, \text{ such that } \pi_{u_0}(t)=\ell_{u_0}(t)\frac{Z(t)}{aZ(t)+b},|\ell_{u_0}(t)|\leqslant \mathcal{C}_{u_0}, \forall t\in[0,T]\Big\},\\
		\mathscr{U}_{-u_0}=&\Big\{\{(\pi_{u},a_{u})_{{u}\in\mathbb{U},u\neq u_0}\}:(\pi_{u},a_{u}) \text{ is progressively measurable w.r.t. }{\mathbb{F}}\text{ and }0\leqslant a_{u}(t)\leqslant1, \forall {u}\in\mathbb{U},u\neq u_0,\\
		&\exists \{\mathcal{C}_{u}\}_{{u}\in\mathbb{U},u\neq {u_0}}\subseteq\mathbb{R}_+, \sup_{{u}\in\mathbb{U},u\neq {u_0}}\mathcal{C}_{u}<\infty, \text{ such that } \pi_{u}(t)=\ell_{u}(t)\frac{Z(t)}{aZ(t)+b},|\ell_{u}(t)|\leqslant \mathcal{C}_{u}, \forall t\in[0,T]\Big\},\\
	\mathscr{A}=&\Big\{\left(\varphi, \chi,\phi,\vartheta\right):\left(\varphi, \chi,\phi,\vartheta\right) \text{ is progressively measurable w.r.t. }\mathbb{F}\text{ and }\left(\varphi(t), \chi(t)\right)=\left(\hslash(t),\hbar(t) \right)\sqrt{Z(t)},\\&\quad
		|\hslash(t)|, |\hbar(t)|\leqslant \mathcal{C},\sup|\phi|<\infty,\sup|\vartheta|<\infty,~ \forall t\in[0,T]\Big\}.\\
	\end{aligned}
\end{equation} 

{ 
In this paper, we study the mean-field game by taking the limit as \(n \to \infty\). The solutions to the \(n\)-insurer game and the mean-field game are shown to be consistent. In social optimal control problems, the optimization is typically high-dimensional, and finding the exact socially optimal solution is challenging. To address this, \cite{huang2012social} and \cite{wang2020social} solve an auxiliary problem and construct a mean-field approximation. The optimization rule in the mean-field game differs from that in the social player game and depends on the specific optimization goal. The results from the constructed mean-field game are shown to be asymptotically optimal for the social control problem, as studied in \cite{huang2012social} and \cite{wang2020social}.
}

The insurers {search for} the robust mean-field equilibrium strategies within the admissible set. We present definition of the robust mean-field equilibrium in the mean-field game as follows.
\begin{definition}\label{def:mfe}
We call $ (\pi_{U}^*,a_{U}^* ), U\sim \mathcal{D}$ a \textit{\textbf{robust mean-field equilibrium}} if $(\pi_{U}^*,a_{U}^*)\in\mathscr{U}$, and  \begin{equation}\label{rmf-s}
(\pi_{u_0}^*,a_{u_0}^* )=\arg	\sup _{(\pi_{u_0},a_{u_0} )\in\mathscr{U}_{u_0}} \inf _{\left(\varphi_{u_0}, \chi_{u_0},\phi_{u_0},\vartheta_{u_0}\right)\in\mathscr{A}} J_{u_0}\left((\pi_{u_0},a_{u_0}),\left(\varphi_{u_0}, \chi_{u_0},\phi_{u_0},\vartheta_{u_0}\right),\bar{X}_{u_0}^*(T),0,Y_{u_0}(0),z^0\right), \forall  u_0\in\mathbb{U} .
\end{equation}
where \begin{equation*}
	\begin{aligned}
		&J_{u_0}\left((\pi_{u_0},a_{u_0}),\left(\varphi_{u_0}, \chi_{u_0},\phi_{u_0},\vartheta_{u_0}\right),\bar{X}_{u_0}^*(T),s,y,z\right)\\
		=&\mathbb{E}^{\mathbb{Q}^{u_0}}_{s,y,z}\left[ U_f\left(X_{u_0}(T) {-\theta_{0}}\bar{X}_{u_0}^*(T) ; \delta_{0}\right)+\int_{s}^{T}  \frac{\varphi_{u_0}^{2}(t)}{2 \Psi_1^{u_0}} +\frac{\chi_{u_0}^2(t)}{2 \Psi^{u_0}_2}+ \frac{\phi_{u_0}^{2}(t)}{2 \Psi_3^{u_0}} +\frac{\mathbb{E}^{\mathbb{Q}^{u_0}}\left[(\vartheta_{u_0}^U(t))^2\Big|\mathcal{F}_T^{\tilde{W},{W},{B}}\right]}{2 \Psi^{u_0}_4} \mathrm{d} t \right],
	\end{aligned}
\end{equation*}
\begin{equation}\label{mX_u0}
	\bar{X}_{u_0}^*(T)=\mathbb{E}^{\mathbb{Q}^{u_0}}\left[X^*_{U}(T)\Big|\mathcal{F}_T^{\tilde{W},{W},{B}}\right],
\end{equation}\begin{equation}\label{xi0xi*}
\begin{aligned}
	\mathrm{d} X^*_{U}(t)
	=&\left[r X^*_{U}(t)+\eta\left(\lambda+\hat{\lambda}\right) \mu_{ 1}-\hat{\eta}(1-a^*_{U}(t))^2\left(\lambda+\hat{\lambda}\right) \mu_{ 2}+\pi^*_{U}(t)m({a}{Z(t)}+{{b}})\right] \mathrm{d} t\\&+a_{U}^*(t)\left(\sqrt{\hat{\lambda}}\mu_{1}(\rd  \tilde{W}^{\mathbb{Q}^{u_0}}(t)+\phi_{u_0}(t)\rd t)+\sqrt{(\hat{\lambda}+\lambda)\mu_{ 2}-\hat{\lambda}\mu_{ 1}^2}(\rd \hat{W}_U^{\mathbb{Q}^{u_0}}(t)+\vartheta_{u_0}^U(t)\rd t)\right)\\
	&+\pi_{U}^*(t){\Sigma}(t)(\rd {W}^{\mathbb{Q}^{u_0}}(t)+\varphi_{u_0}(t)\rd t).
\end{aligned}
\end{equation}
$ \Psi_1^{u_0} $, $ \Psi_2^{u_0} $, $ \Psi_3^{u_0} $ and $ \Psi_4^{u_0} $ should be determined later to get closed-form solutions. 
\end{definition}

Definition \ref{def:mfe} indicates that the robust mean-field equilibrium is formulated by finding a fixed point. For an insurer with type vector $ u_0 $, with a given mean of the insurers' wealth $ \bar{X}_{u_0}^*(T) $, she/he first searches for the robust optimal reinsurance-investment strategy  $(a^*_{u_0},\pi_{u_0}^*)$ under the worst-case scenario. Every insurer in the system behaves like this. For insurer with random type vector  $ U \sim\mathcal{D}$, under the robust optimal reinsurance-investment strategy $(a^*_{U},\pi_{U}^*)$ and the probability measure  given in type vector $ u_0 $, the  surplus process $ X^*_{U}(t) $ involves as \eqref{xi0xi*}.   If we have a fixed point in the sense that the consistency condition $ \bar{X}_{u_0}^*(T)=\mathbb{E}^{\mathbb{Q}^{u_0}}\left[X^*_{U}(T)\Big|\mathcal{F}_T^{\tilde{W},{W},{B}}\right] $ and $ U \sim \mathcal{D} $ holds, then  $ (\pi_{U}^*,a_{U}^* ), U \sim \mathcal{D}$ is called robust mean-field equilibrium.

%

Based on \eqref{m-pi} and \eqref{m-a}, we should expect that
\begin{equation}\label{m-opt}
	\left\{\begin{aligned}
		\pi_{U}^*(t)&=\theta\frac{\sqrt{Z(t)}}{\sigma e^{r(T-t)}}\frac{\mathbb{E}[(\frac{m}{1-\beta_{1}}+\nu\rho h_{U}(t))\frac{1}{\delta }]}{1-\mathbb{E}[\theta]}+(\frac{m}{1-\beta_{1}}+\nu\rho h_{U}(t))\frac{\sqrt{Z(t)}}{\delta\sigma e^{r(T-t)}},\\
		a_U^*(t)&=\left(\frac{Q_U(t)}{R_U(t)}\bar{\Omega}(t)+\frac{P_U(t)}{R_U(t)}\right)\wedge1,
	\end{aligned}\right.
\end{equation}
where $\beta_1=\frac{\psi_{1}}{-\delta}  $, $\beta_2=\frac{\psi_{2}}{-\delta}  $, $\beta_3=\frac{\psi_{3}}{-\delta}  $, $\beta_4=\frac{\psi_{4}}{-\delta}  $ and   $\overline{\Omega}(t)=\mathbb{E}[\mu_{ 1}a_{U}^*(t)] $	is obtained by $\overline{\Omega}(t)=\mathbb{E}[\mu_{ 1}\left(\frac{Q_U(t)}{R_U(t)}\bar{\Omega}(t)+\frac{P_U(t)}{R_U(t)}\right)\wedge1] $ and \begin{equation*}\begin{aligned}
		&	P_U(t)=\frac{2\hat{\eta}}{ \delta e^{r(T-t)}}(\lambda +\hat{\lambda})\mu_{  2},\quad Q_U(t)=\hat{\lambda}\theta \mu_{  1}(1-\beta_{ 3}),\\& R_U(t)=\frac{2}{\delta e^{r(T-t)}}\hat{\eta}(\lambda +\hat{\lambda})\mu_{ 2}+ (1-\frac{\theta }{n})\left[(\lambda +\hat{\lambda})\mu_{ 2}(1-{\beta_{ 4}})+\hat{\lambda}\mu_{  1}^2({\beta_{ 4}}-{\beta_{ 3}})\right].
\end{aligned}	\end{equation*}
\eqref{m-opt} is the robust mean-field equilibrium obtained by letting $n\rightarrow\infty$ in the competition system.

To derive the robust mean-field equilibrium $ (\pi_{U}^*,a_{U}^* ), U\sim \mathcal{D}$, we first define $ Y_{u_0}(t)=X_{u_0}(t) {-\theta_{0}}\bar{X}^*_{u_0}(t),$  $ \forall  u_0\in\mathbb{U} $, where $ \bar{X}^*_{u_0}(t) $ is defined as in \eqref{mX_u0}. Then we need to solve   Problem \eqref{Yxi} for all $u_0\in\mathbb{U}  $ simultaneously.
\begin{equation}\label{Yxi}
	\sup _{(\pi_{u_0},a_{u_0} )\in\mathscr{U}_{u_0}} \inf _{\left(\varphi_{u_0}, \chi_{u_0},\phi_{u_0},\vartheta_{u_0}\right)\in\mathscr{A}}\!\mathbb{E}^{\mathbb{Q}^{u_0}}\!\left[\! U_f\left(Y_{u_0}(T)  ; \delta_{0}\right)\!+\!\int_{0}^{T} \! \frac{\varphi_{u_0}^{2}(t)}{2 \Psi_1^{u_0}} \!+\!\frac{\chi_{u_0}^2(t)}{2 \Psi^{u_0}_2} \!+\! \frac{\phi_{u_0}^{2}(t)}{2 \Psi_3^{u_0}} \!+\!\frac{\mathbb{E}\left[(\vartheta_{u_0}^U(t))^2\Big|\mathcal{F}_T^{\tilde{W},{W},{B}}\right]}{2 \Psi^{u_0}_4}\mathrm{d} t \right]
\end{equation}

Based on the dynamic of $ X_{U}(t) $, we have
\begin{equation}
	\begin{aligned}
			&\mathrm{d} Y_{u_0}(t)\\
		=&r Y_{u_0}(t)\mathrm{d} t+\left(\eta_0\left(\lambda_0+\hat{\lambda}\right) \mu_{0 1}-\theta_0\mathbb{E}\left[\eta\left(\lambda+\hat{\lambda}\right) \mu_{ 1}\right]\right)\rd t\\
		&-\hat{\eta}(1-a_{u_0}(t))^2\left(\lambda_0+\hat{\lambda}\right) \mu_{ 02}\rd t+\theta_0\hat{\eta}\mathbb{E}\left[(1-a^*_{U}(t))^2\left(\lambda+\hat{\lambda}\right) \mu_{ 2}\Big|\mathcal{F}_T^{\tilde{W},{W},{B}}\right]\rd t\\
		&+\left(\pi_{u_0}(t)-\theta_0\mathbb{E}\left[\pi^*_{U}(t)\Big|\mathcal{F}_T^{\tilde{W},{W},{B}}\right]\right)\left({\Sigma}(t)\rd {W}^{\mathbb{Q}^{u_0}}(t)+\left(m\sqrt{Z(t)}+{\varphi_{u_0}(t)}\right)\Sigma(t) \mathrm{d} t\right)\\
		&+\sqrt{\hat{\lambda}}\left(a_{u_0}(t)\mu_{01}-\theta_0\mathbb{E}\left[a^*_{U}(t)\mu_{1}\Big|\mathcal{F}_T^{\tilde{W},{W},{B}}\right]\right)(\rd  \tilde{W}^{\mathbb{Q}^{u_0}}(t)+\phi_{u_0}(t)\rd t)\\
		&+a_{u_0}(t)\sqrt{(\hat{\lambda}+\lambda_0)\mu_{0 2}-\hat{\lambda}\mu_{0 1}^2}(\rd \hat{W}_U^{\mathbb{Q}^{u_0}}(t)+\vartheta_{u_0}^U(t)\rd t)
		\\&-\theta_0\mathbb{E}\left[a^*_{U}(t)\sqrt{(\hat{\lambda}+\lambda)\mu_{ 2}-\hat{\lambda}\mu_{ 1}^2}\vartheta_{u_0}^U(t)\Big|\mathcal{F}_T^{\tilde{W},{W},{B}}\right]\rd t.
	\end{aligned}
\end{equation}
To solve Problem \eqref{Yxi}, we first define the  infinitesimal generator based on the dynamic of $ Y_{u_0}(t) $
\begin{equation*}
	\begin{aligned}
&\mathcal{A}^{u_0}f(t,y,z)\\
=&f_t+\left[{\kappa}(\bar{Z}-z)+{\nu}\sqrt{z}({\rho}\varphi_{u_0}+\sqrt{1-{\rho}^2}\chi_{u_0})\right]f_{z}+\frac{1}{2}\nu^2zf_{zz}\\
&+\rho\nu (az+b)\left(\pi_{u_0}(t)-\theta_0\mathbb{E}\left[\pi^*_{u}(t)\Big|\mathcal{F}_T^{\tilde{W},{W},{B}}\right]\right)f_{yz}\\
&+r yf_y+\left(\eta_0\left(\lambda_0+\hat{\lambda}\right) \mu_{0 1}-\theta_0\mathbb{E}\left[\eta\left(\lambda+\hat{\lambda}\right) \mu_{ 1}\right]\right)f_y\\
&-\hat{\eta}(1-a_{u_0}(t))^2\left(\lambda_0+\hat{\lambda}\right) \mu_{0 2}f_y+\theta_0\mathbb{E}\left[\hat{\eta}(1-a^*_{u}(t))^2\left(\lambda+\hat{\lambda}\right) \mu_{ 2}\Big|\mathcal{F}_T^{\tilde{W},{W},{B}}\right]f_y\\
&+\left(\pi_{u_0}(t)-\theta_0\mathbb{E}\left[\pi^*_{u}(t)\Big|\mathcal{F}_T^{\tilde{W},{W},{B}}\right]\right)\left(m\sqrt{z}+{\varphi_{u_0}(t)}\right)\sigma f_y\\
&+\sqrt{\hat{\lambda}}\left(a_{u_0}(t)\mu_{01}-\theta_0\mathbb{E}\left[a^*_{U}(t)\mu_{1}\Big|\mathcal{F}_T^{\tilde{W},{W},{B}}\right]\right)\phi_{u_0}(t)f_y\\
&+a_{u_0}(t)\sqrt{(\hat{\lambda}+\lambda_0)\mu_{0 2}-\hat{\lambda}\mu_{0 1}^2}\vartheta_{u_0}^U(t)f_y
-\theta_0\mathbb{E}\left[a^*_{U}(t)\sqrt{(\hat{\lambda}+\lambda)\mu_{ 2}-\hat{\lambda}\mu_{ 1}^2}\vartheta_{u_0}^U(t)\Big|\mathcal{F}_T^{\tilde{W},{W},{B}}\right]f_y\\
&+\frac{1}{2}{\hat{\lambda}}\left(a_{u_0}(t)\mu_{01}-\theta_0\mathbb{E}\left[a^*_{u}(t)\mu_{1}\Big|\mathcal{F}_T^{\tilde{W},{W},{B}}\right]\right)^2f_{yy}+\frac{1}{2}\left(\pi_{u_0}(t)-\theta_0\mathbb{E}\left[\pi^*_{u}(t)\Big|\mathcal{F}_T^{\tilde{W},{W},{B}}\right]\right)^2{\sigma}^2f_{yy}\\
&+\frac{1}{2}a_{u_0}^2(t){\left((\hat{\lambda}+\lambda_0)\mu_{0 2}-\hat{\lambda}\mu_{ 01}^2\right)}f_{yy},
	\end{aligned}
\end{equation*}
where $ \sigma=a\sqrt{z}+\frac{b}{\sqrt{z}} $.

Similarly, let $ (y,z) $ be  the value of process $ (Y_{u_0}(s), Z(s)) $ at time $ t $, denote the value function of insurer with type $ u_0$ at time $t$ by
$$
\begin{aligned}
	V^{u_0}(t,  {y},z)=\sup _{(\pi_{u_0},a_{u_0} )\in\mathscr{U}_{u_0}} \inf _{\left(\varphi_{u_0}, \chi_{u_0},\phi_{u_0},\vartheta_{u_0}\right)\in\mathscr{A}}&\mathbb{E}_{t,y,z}^{\mathbb{Q}^{u_0}}\!\left[\!\int_{t}^{T}  \!\frac{\varphi_{u_0}^{2}(s)}{2 \Psi_1^{u_0}} \!+\!\frac{\chi_{u_0}^2(s)}{2 \Psi^{u_0}_2}  \!+\! \frac{\phi_{u_0}^{2}(s)}{2 \Psi_3^{u_0}} \!+\!\frac{\mathbb{E}\left[(\vartheta_{u_0}^U(s))^2\Big|\mathcal{F}_T^{\tilde{W},{W},{B}}\right]}{2 \Psi^{u_0}_4}\mathrm{d} s  \right].
\\&+\mathbb{E}_{t,y,z}^{\mathbb{Q}^{u_0}}U_f\left(Y_{u_0}(T) ; \delta_{0}\right)
\end{aligned}$$
\begin{definition}[HJBI equation]\label{def:hjbi2}
The HJBI equation of insurer with type $ u_0=(x^0_0,\lambda_0,\mu_{01},\mu_{0 2},\eta_0,\theta_0,\delta_0,\psi_{01},\psi_{02},\psi_{03},\psi_{04})$ in  Problem \eqref{Yxi} is
\begin{equation}\label{MHJBI}
\sup _{(\pi_{u_0},a_{u_0} )\in\mathscr{U}_{u_0}} \inf _{\left(\varphi_{u_0}, \chi_{u_0},\phi_{u_0},\vartheta_{u_0}\right)\in\mathscr{A}}\left\{\!\mathcal{A}^{u_0}V(t,y,z)\!+ \! \frac{\varphi_{u_0}^{2}(t)}{2 \Psi_1^{u_0}} \!+\!\frac{\chi_{u_0}^2(t)}{2 \Psi^{u_0}_2}\! +\! \frac{\phi_{u_0}^{2}(t)}{2 \Psi_3^{u_0}} \!+\!\frac{\mathbb{E}\left[(\vartheta_{u_0}^U(t))^2\Big|\mathcal{F}_T^{\tilde{W},{W},{B}}\right]}{2 \Psi^{u_0}_4} \right\}\! =\!0,
\end{equation}
with boundary condition $V(T, y,z)= U_f\left(y; \delta_{0}\right)=-\frac{1}{\delta_0} e^{-\delta_0  y}$.
\end{definition}

Let  \begin{equation*}
	\Psi_1^{u_0}=\frac{\beta_{01}}{	v^{u_0}},\quad\Psi_2^{u_0}=\frac{\beta_{02}}{	v^{u_0}},\quad	\Psi_3^{u_0}=\frac{\beta_{03}}{	v^{u_0}},\quad\Psi_4^{u_0}=\frac{\beta_{04}}{	v^{u_0}},
\end{equation*}
where $ \beta_{01}=\frac{\psi_{01}}{-\delta_{0}} $, $ \beta_{02}=\frac{\psi_{02}}{-\delta_{0}} $, $ \beta_{03}=\frac{\psi_{03}}{-\delta_{0}} $, $ \beta_{04}=\frac{\psi_{04}}{-\delta_{0}} $ and  $ v^{u_0}(t,y,z) $ is the solution to the following equation:
\begin{equation}\label{HJBI-v-m}
	\begin{aligned}
		\sup _{(\pi_{u_0},a_{u_0} )\in\mathscr{U}_{u_0}} \inf _{\left(\varphi_{u_0}, \chi_{u_0},\phi_{u_0},\vartheta_{u_0}\right)\in\mathscr{A}}&\left\{\!\mathcal{A}^{u_0}V(t,y,z) \!+\!\left(\!\frac{\varphi_{u_0}^{2}(t)}{2 \beta_{01}} \!+\!\frac{\chi_{u_0}^2(t)}{2 \beta_{02}} \!+\! \frac{\phi_{u_0}^{2}(t)}{2 \beta_{03}} \!+\!\frac{\mathbb{E}\left[(\vartheta_{u_0}^U(t))^2\Big|\mathcal{F}_T^{\tilde{W},{W},{B}}\right]}{2 \beta_{04}}\!\right)	\!V(t,  {y},z)\! \right\} \!=\!0,
	\end{aligned}
\end{equation}
Then $ v^{u_0}(t,y,z) $ is also the solution to \eqref{MHJBI}.

Similar to Theorem \ref{solution-HJBI} and Theorem \ref{Verification}, we have the following results for  Problem \eqref{Yxi}. The proofs of the following results are similar to the last section and we omit them. 
\begin{theorem}\label{solution-MHJBI}
	$ \forall u_0\in\mathbb{U} $, write $ u_0=(x^0_0,\lambda_0,\mu_{01},\mu_{0 2},\eta_0,\theta_0,\delta_0,\psi_{01},\psi_{02}) $,	if $ 	\mathcal{C}^2<\frac{\kappa^2}{2\nu^2} $,  then $ v^{u_0}(t,y,z) \in C^{1,2,2}([0,T]\times\mathbb{R}^2_+)$ and is given by \begin{equation*}
		v^{u_0}(t, y,z)= -\frac{1}{\delta} \exp\left(f_{u_0}(t)-\delta_{u_0} g_{u_0}(t)y+h_{u_0}(t)z\right),
	\end{equation*}where \begin{equation*}
	\left\{\begin{aligned}
		& g_{u_0}(t)=e^{r(T-t)},\\
		&h_{u_0}(t)=c_{0,1}\frac{e^{c_{0,2}t}-e^{c_{0,2}T}}{e^{c_{0,2}t}+c_{0,3}e^{c_{0,2}T}},\\
		&0=f'_{u_0}(t)+{\kappa}\bar{Z}h_{u_0}(t)-\left(\eta_0\left(\lambda_0+\hat{\lambda}\right) \mu_{0 1}-\theta_0\mathbb{E}\left[\eta\left(\lambda+\hat{\lambda}\right) \mu_{ 1}\right]\right)\delta_{0}e^{r(T-t)}\\
		&+\hat{\eta}(1-a_{u_0}^\circ(t))^2\left(\lambda_0+\hat{\lambda}\right) \mu_{0 2}\delta_{0}e^{r(T-t)}-\theta_0\hat{\eta}\mathbb{E}\left[(1-a_{U}^*(t))^2\left(\lambda+\hat{\lambda}\right) \mu_{ 2}\right]\delta_{0}e^{r(T-t)}\\
		&+\frac{1}{2}(a^\circ_{u_0}(t))^2((\hat{\lambda}+\lambda_{0})\mu_{0 2}(1-{\beta_{04}})+\hat{\lambda}\mu_{0 1}^2({\beta_{04}}-{\beta_{03}}))\delta_{0}^2e^{2r(T-t)}\\
		&+\hat{\lambda}\frac{\theta_0^2}{2} \mathbb{E}\left[a^*_U(t) \mu_{ 1}\right]^2 \delta_{0}^2e^{2r(T-t)}(1-{\beta_{03}})\\
		&-\hat{\lambda}\theta_0a^\circ_{u_0}(t)\mu_{0 1}\mathbb{E}\left[a^*_U(t)\mu_{ 1}\right]\delta_{0}^2e^{2r(T-t)}(1-{\beta_{03}})
	\end{aligned}\right.
\end{equation*}
and\begin{equation*}
a_{u_0}^\circ=	\left(\frac{Q_{u_0}(t)}{R_{u_0}(t)}\mathbb{E}\left[a_{U}^*(t)\mu_{1}\right]+\frac{P_{u_0}(t)}{R_{u_0}(t)}\right)\wedge1
\end{equation*}

where $ c_{0,1}=\frac{{\kappa}+m\nu\rho +\sqrt{({\kappa}+m\nu\rho )^2+\frac{1-\beta_{0,2}}{1-\beta_{0,1}}\nu^2(1-\rho^2)m^2}}{\nu^2(1-\rho^2)(1-\beta_{0,2})}$, $c_{0,3}=2(1-\beta_{0,1})({\kappa}+m\nu\rho )c_{0,1}+1$, $c_{0,2}=\frac{c_{0,3}+1}{2(1-\beta_{0,1})c_{0,1}} $.
\end{theorem}
Theorem \ref{solution-MHJBI} presents the solutions to the HJBI equations and the following theorem verifies that solutions in Theorem \ref{solution-MHJBI} solves Problem \eqref{Yxi}.
\begin{theorem}[Verification theorem]\label{mVerification}
	$ \forall u_0\in\mathbb{U} $, $v^{u_0}\in C^{1,2,2}([0,T]\times\mathbb{R}^2_+) $ solving  \eqref{MHJBI},  the equation  $ V^{u_0}(t, y,z)=v^{u_0}(t, y,z) $ holds.
\end{theorem}
Based on Theorems \ref{solution-MHJBI} and \ref{mVerification}, we obtain the robust mean-field equilibrium reinsurance and investment strategies as follows.
\begin{theorem}	
	If $1\neq \mathbb{E}[\theta]$, the robust mean-field equilibrium strategy $ (\pi^*_{U},a^*_{U}) $,  $ U\sim\mathcal{D} $ exists and is given by \eqref{m-opt}. 
	 Furthermore,  the associated probability measure transformation functions  are $ \left(\varphi_{U}^*, \chi_{U}^*,\phi_{U}^*,\vartheta_{U}^*\right)$,  $ U\sim\mathcal{D} $, when $ U=u_0 $, the explicit forms  are
\begin{equation*}
	\left\{\begin{aligned}
		&{\varphi^*_{u_0}(t)}=\frac{m\beta_{01}}{1- {\beta_{01}}}\sqrt{Z(t)},\\
		&{\chi^*_{u_0}(t)}=-\nu\sqrt{1-\rho^2}\beta_{02} h_{u_0}(t)\sqrt{Z(t)}\\
		&\phi^*_{u_0}(t)=\sqrt{\hat{\lambda}}\left[\mu_{01}a^*_{u_0}(t)-\theta\mathbb{E}\left[a_{U}^*(t)\mu_{1}\right]\right]{\beta_{03}}\delta_{0} e^{r(T-t)},\\
		&\vartheta^{u_0,*}_{u_0}(t)=a^*_{u_0}(t)\sqrt{(\hat{\lambda}+\lambda_{0})\mu_{0 2}-\hat{\lambda}\mu_{0 1}^2}{\beta_{04}}\delta_{0}e^{r(T-t)},\\
		&\vartheta^{u_1,*}_{u_0}(t)=0,~u_1\neq u_0.
	\end{aligned}\right.
\end{equation*}
If $1= \mathbb{E}[\theta]$, the robust  mean-field equilibrium does not exist.
\end{theorem}
From the above paragraph, we see that the results in the mean-field and $n$-insurer cases are consistent. Consequently, we are able to obtain closed-form solutions for the robust equilibrium investment strategy, while the robust equilibrium reinsurance strategy can be computed numerically.

	\section{\bf Numerical analysis}
	In this section, we show the impacts of different parameters on the insurers' optimal portfolio processes and optimal reinsurance proportions. As explained in Remark \ref{Randomness}, we will use the deterministic coefficients $ \pi^*_i(t)/\frac{{Z(t)}}{a{Z(t)}+b} $ to illustrate the portfolio behavior of the insurers and we use   $ a^*_i(t) $ to illustrate the reinsurance behavior of the insurers. Unless otherwise stated, the parameters we adopt are real-data estimates as reported
	in \cite{cheng2021optimal}, which is based on S\&P 500 and VIX data from January 2010 to the last
	day of December 2019: $ T=5 $, $ r=0.02 $, 
	$ \kappa=7.3479 $, $\bar{Z}=0.0328$, $ \nu=0.6612 $, $ z^0=0.04 $, $ \rho=-0.7689 $, $ m=2.9428, $ parameters $ a $ and $ b $ are chosen such that $ a\sqrt{z^0}+\frac{b}{\sqrt{z^0}}=\sqrt{z^0} $, let $ a = 0.9051 $, $ b=0.0023 $.
	
	For simplicity, we consider two insurers. The payment of  insurer 1 is low frequency, high payment and 
 the payment of   insurer 2  is high frequency, low payment, that is,  we set 
	$ \lambda_1 = 0.9,$  $
	\lambda_2 = 2.4,$  $\hat{\lambda} = 0.6,$  $ \eta_1 = 0.2,$  $
	\eta_2 = 0.2,$ $	\hat{\eta} = 0.25,  $  $
	\mu_{1,1} = 1,$  $
	\mu_{1,2} = 2,$  $
	\mu_{2,1} = 1/2,$  $
	\mu_{2,2} = 1/2,$  we then let
	$
	\delta_1 = 1.5,$  $
	\delta_2 = 1.3,$  $
	\theta_1 = 0.7,$  $
	\theta_2 = 0.7,$  $
	\psi_{1,1} = 5,$  $
	\psi_{1,2} = 7,$  $
	\psi_{1,3} = 5,$  $
	\psi_{1,4} = 7,$  $
	\psi_{2,1} = 5,$  $
	\psi_{2,2} = 7 ,$ $
	\psi_{2,3} = 5,$  $
	\psi_{2,4} = 7 $ be the default  values of the parameters for the two insurers.

\subsection{Robust equilibrium reinsurance strategy}

We are mainly concerned with the effects of competition, risk aversion parameters, and ambiguity aversion parameters. Figs.~\ref{lambda_hat_Reinsurance}-\ref{eta_hat_Reinsurance} shows the effects of $\hat{\lambda}$ and $\hat{\eta}$ on the reinsurance strategies, respectively. $\hat{\lambda}$ represents the intensity of the common insurance business. Because the insurance premium is calculated by the expected value principle, the insurers can earn profits by undertaking more insurance business. As such, the reinsurance strategy increases with $\hat{\lambda}$, which is illustrated in Fig.~\ref{lambda_hat_Reinsurance}. $\hat{\eta}$ characterizes the reinsurance premium. For a larger $\hat{\eta}$, the insurers have to pay more premiums to divide the insurance risk. Then the insurers will undertake more insurance risk by themselves. Fig.~\ref{eta_hat_Reinsurance} depicts the positive relationship between $\hat{\eta}$ and the reinsurance strategy. Because the financial market and insurance market are independent, we see from \eqref{a_star} that the ambiguity aversion parameters $\psi_{\cdot,1}$ and  $\psi_{\cdot,2}$  over financial risks do not influence the robust equilibrium reinsurance strategy. Figs.~\ref{theta_1_Reinsurance} and \ref{theta_2_Reinsurance} illustrate the effects of competition on insurer 1's reinsurance strategy. When insurer 1 is more cared about relative performance, insurer 1 will be more risk-seeking to perform better than others. Thus, in Fig.~\ref{theta_1_Reinsurance}, insure 1's reinsurance strategy increases with $\theta$, which means that insurer 1 undertakes more insurance risk when she/he is more concerned with relative performance. Fig.~\ref{theta_2_Reinsurance} shows insurer 2's competition attitude toward insurer 1's reinsurance strategy. We see that $a_1^*$ has a positive relationship with $\theta_2$. When $\theta_2$ increases, insurer 2 will choose a larger $a_2^*$. In response to insurer 2's action, insurer 1 will also adopt a larger $a_1^*$ in the competition game.

Fig.~\ref{delta_1_Reinsurance} shows the effect of $\delta_1$ on insurer 1's robust equilibrium strategy. $\delta_1$ represents insurer 1's risk aversion attitude. For a larger $\delta_1$, insurer 1 becomes more risk aversion and intends to undertake less insurance risk by itself. Then insurer 1 divides more insurance risk to the insurer and adopts a less retention level $a_1^*$, which is well illustrated in Fig.~\ref{delta_1_Reinsurance}.  The effect of insurer 2's risk aversion parameter $\delta_2$ on insurer 1's robust equilibrium reinsurance strategy is depicted in Fig.~\ref{delta_2_Reinsurance}. Insurer 2 will divide more insurance risk with a higher $\delta_2$. Fig.~\ref{delta_2_Reinsurance} shows that insurer 1 will also take similar action in the competition, i.e., $a_1^*$ decreases with insurer 2's risk aversion parameter.  

$\psi_{\cdot,3}$ and $\psi_{\cdot,4}$ characterize the ambiguity aversion attitudes towards common and idiosyncratic insurance risks. When insurer 1 is more ambiguity aversion, she/he has less faith in the financial model. Thus, to handle higher model uncertainty, insurer 1 retains less insurance risk and decreases $a^*_1$ whether $\psi_{1,3}$ or $\psi_{2,4}$ increases, which are revealed in   Figs.~\ref{psi1_3_Reinsurance} and \ref{psi1_4_Reinsurance}. Besides, we are also interested in the effect of insurer 2's ambiguity attitude on insurer 1's reinsurance strategy. Figs.~\ref{psi2_3_Reinsurance} and \ref{psi2_4_Reinsurance} show that insurer 1's reinsurance strategy has a positive relationship with insurer 2's ambiguity attitude. Similar to the explanations in  Figs.~\ref{theta_2_Reinsurance} and \ref{delta_2_Reinsurance}, when  $\psi_{2,3}$ or $\psi_{2,4}$ increases, insurers 2 and 1 will all undertake less insurance risk in the competition game. From Figs.~\ref{theta_1_Reinsurance}-\ref{psi2_4_Reinsurance}, we see that the behaviors of different insurers tend to be convergent in the competition system. When others are more aggressive and undertake more risks, the insurer will also retain more insurance risk in order to perform better than others. 

\begin{figure}[htbp]
	\centering
	\begin{minipage}[t]{0.45\linewidth} 
		\centering
		\includegraphics[width=0.9\linewidth]{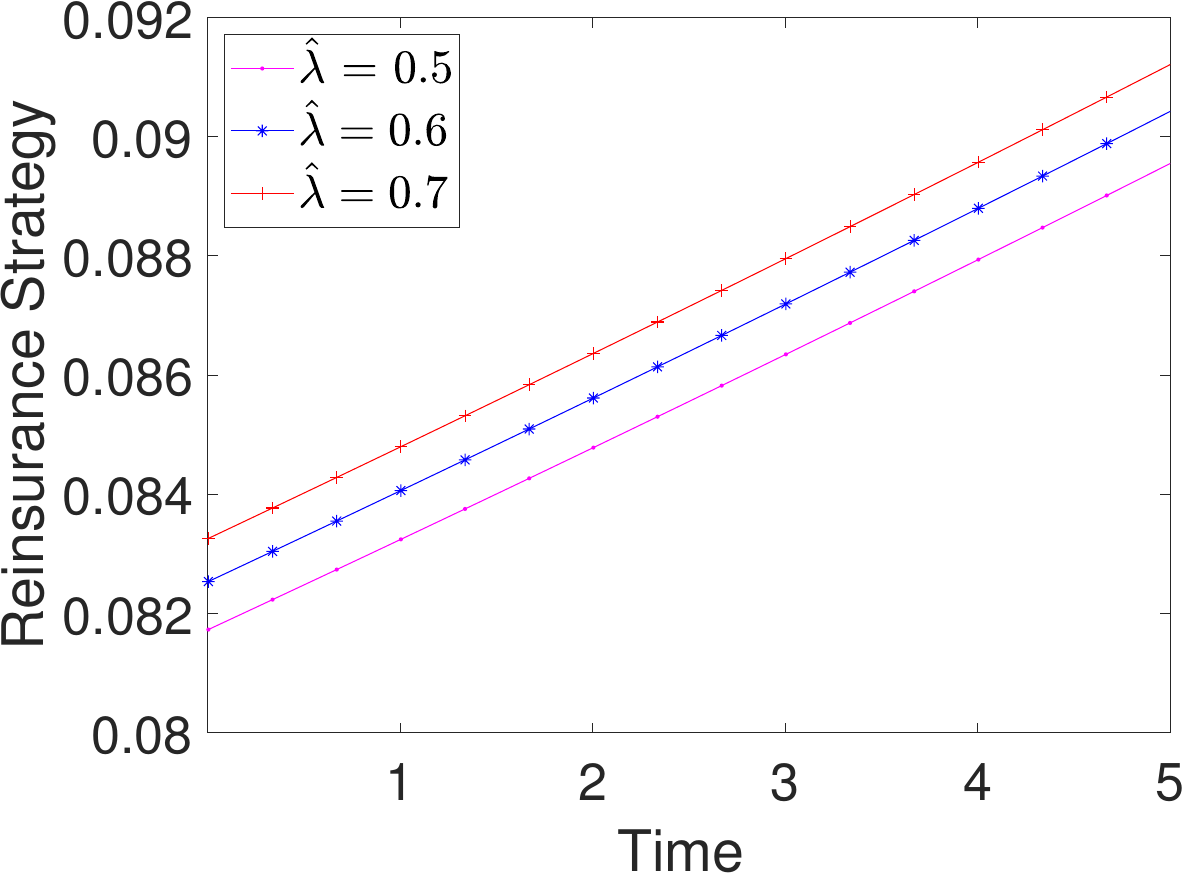}
		\caption{Effects of $ \hat{\lambda} $ on $ a^*_1(t) $}
		\label{lambda_hat_Reinsurance}
	\end{minipage}
	\centering
   \begin{minipage}[t]{0.45\linewidth} 
	\centering
	\includegraphics[width=0.9\linewidth]{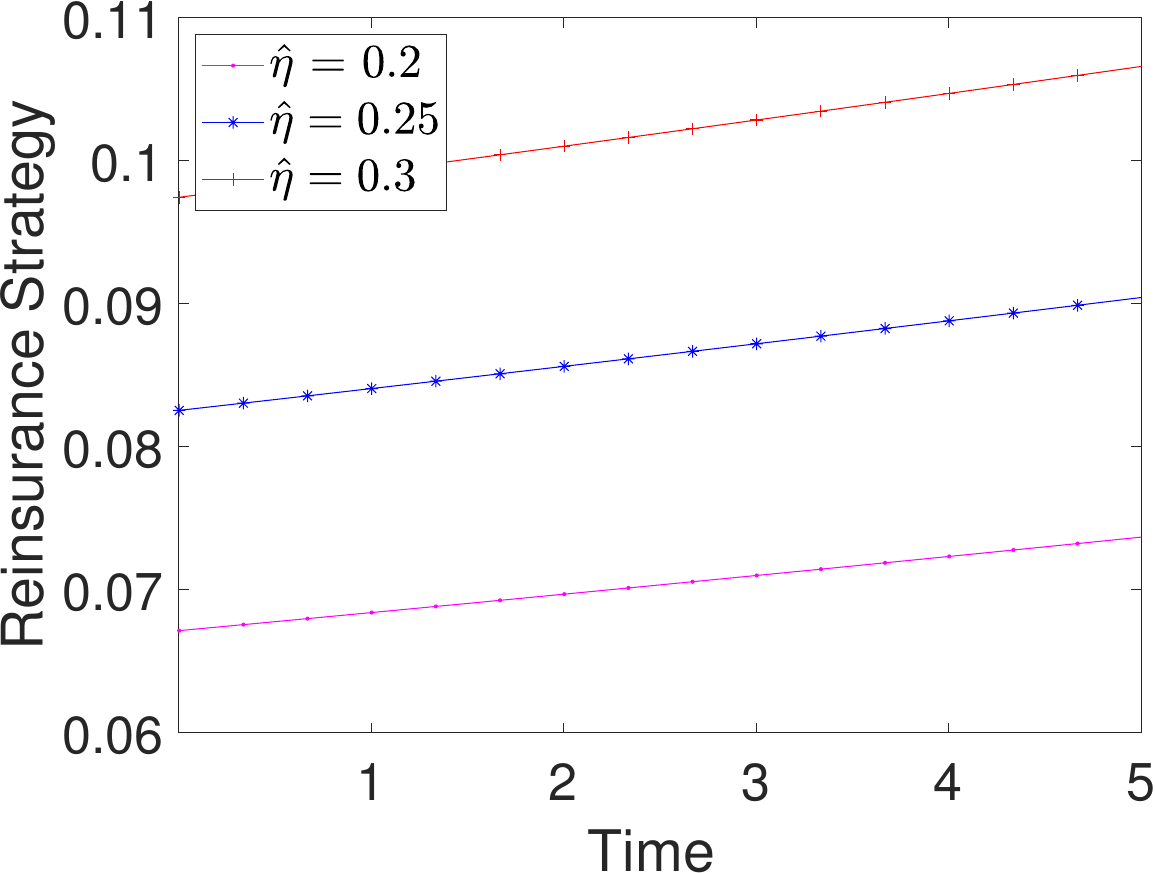}
	\caption{Effects of $ \hat{\eta} $ on $ a^*_1(t) $}
	\label{eta_hat_Reinsurance}
   \end{minipage}
\end{figure}

\begin{figure}[htbp]
	\centering
\begin{minipage}[t]{0.45\linewidth} 
	\centering
	\includegraphics[width=0.9\linewidth]{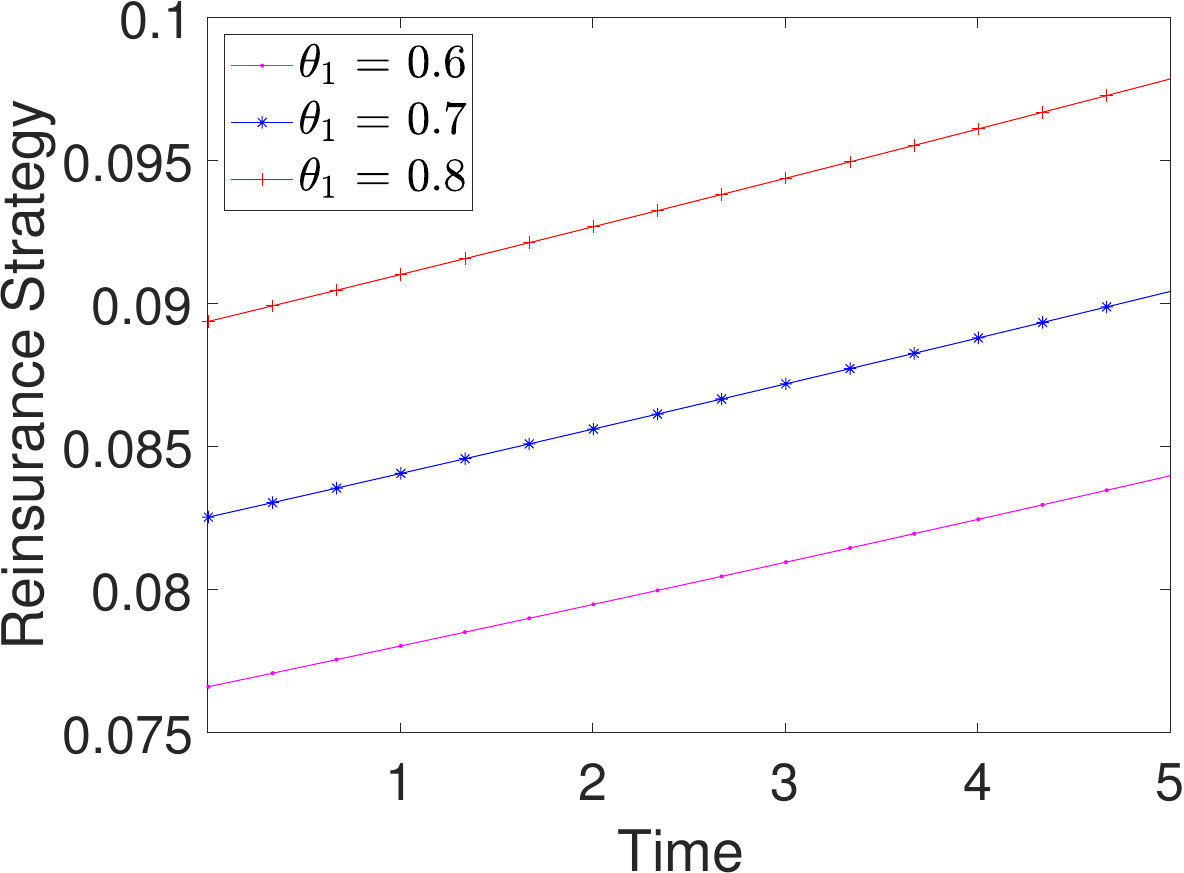}
	\caption{Effect of $ \theta_1 $ on $ a^*_1(t) $.}
	\label{theta_1_Reinsurance}
\end{minipage}
\begin{minipage}[t]{0.45\linewidth} 
	\centering
	\includegraphics[width=0.9\linewidth]{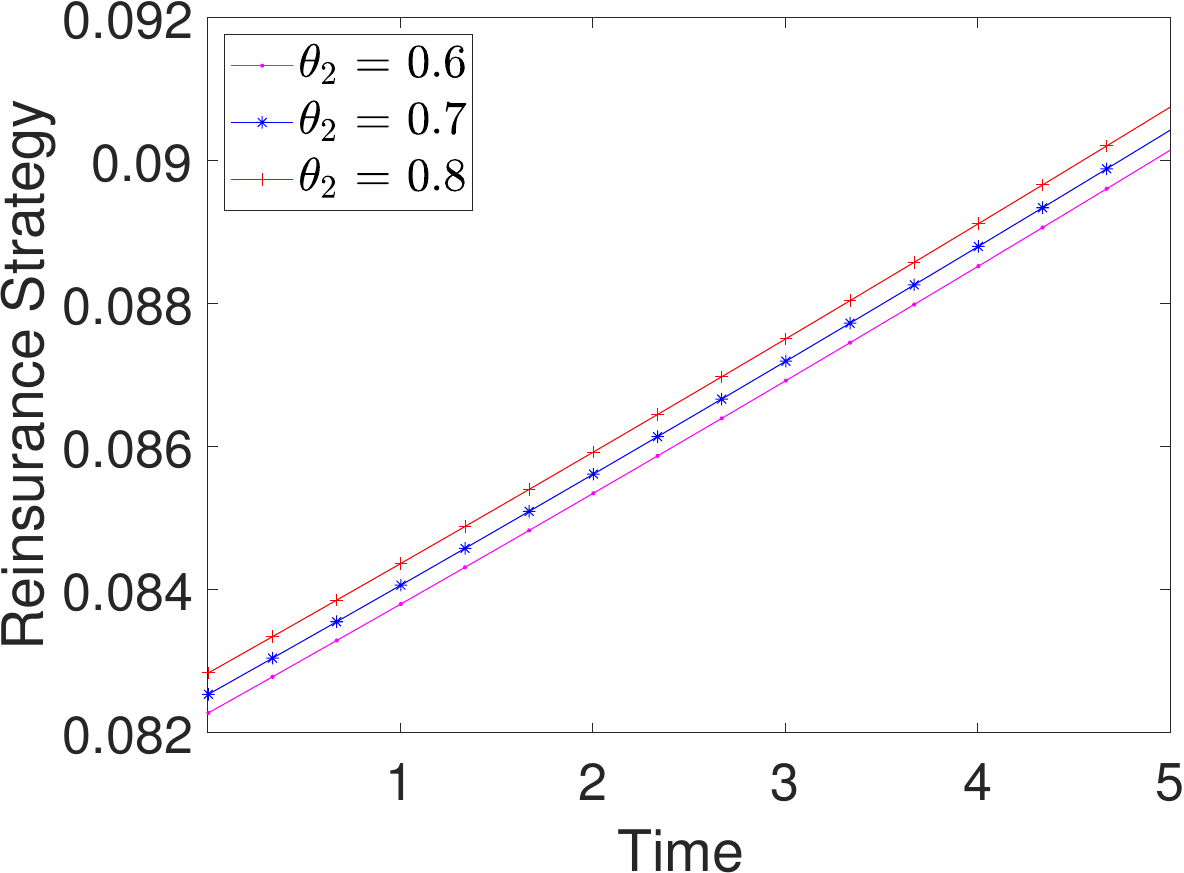}
	\caption{Effect of $ \theta_2$ on $ a^*_1(t) $.}
	\label{theta_2_Reinsurance}
\end{minipage}
\begin{minipage}[t]{0.45\linewidth} 
	\centering
	\includegraphics[width=0.9\linewidth]{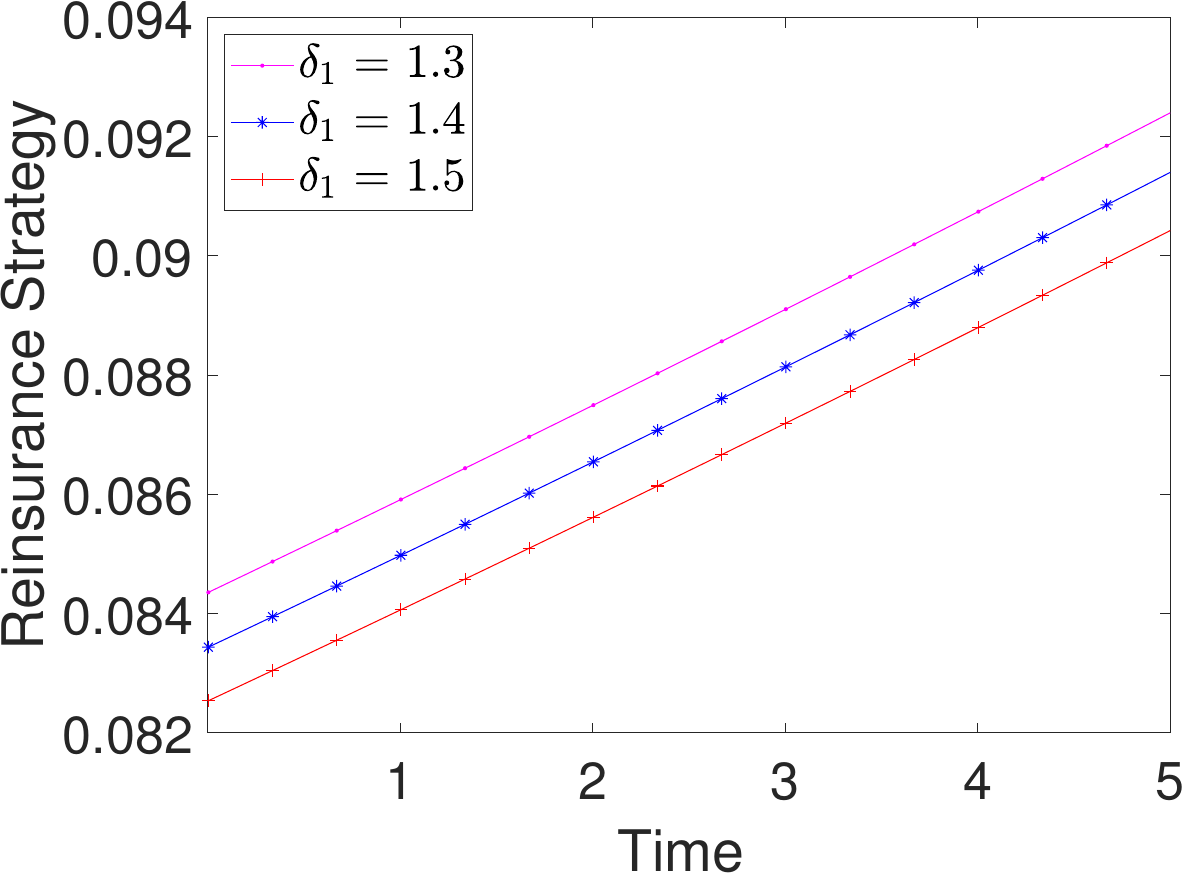}
	\caption{Effect of $ \delta_1 $ on $ a^*_1(t) $.}
	\label{delta_1_Reinsurance}
\end{minipage}
\begin{minipage}[t]{0.45\linewidth} 
	\centering
	\includegraphics[width=0.9\linewidth]{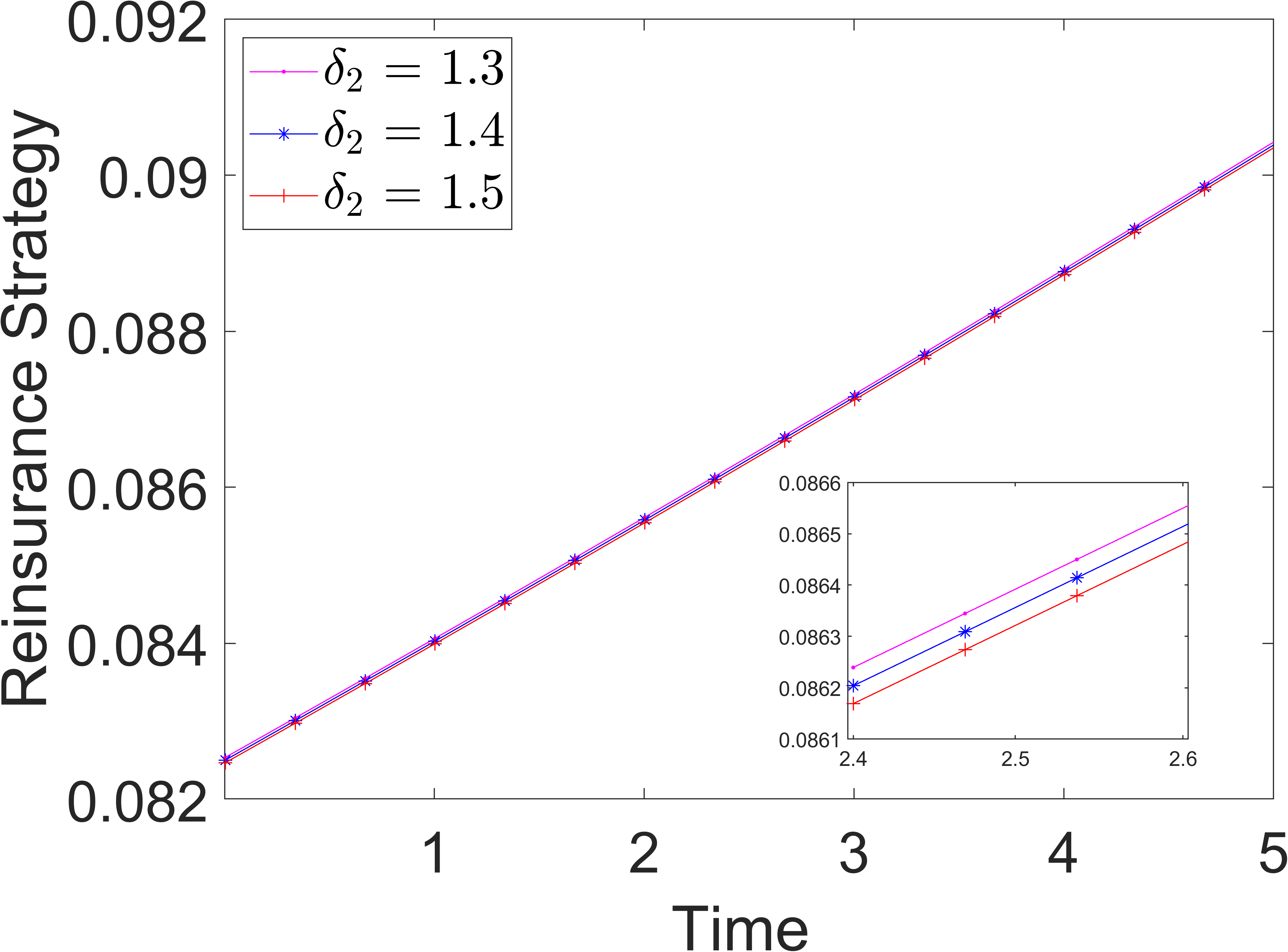}
	\caption{Effect of $ \delta_2$ on $ a^*_1(t) $.}
	\label{delta_2_Reinsurance}
\end{minipage}
\end{figure}

\begin{figure}[htbp]
	\centering
\begin{minipage}[t]{0.45\linewidth} 
	\centering
	\includegraphics[width=0.9\linewidth]{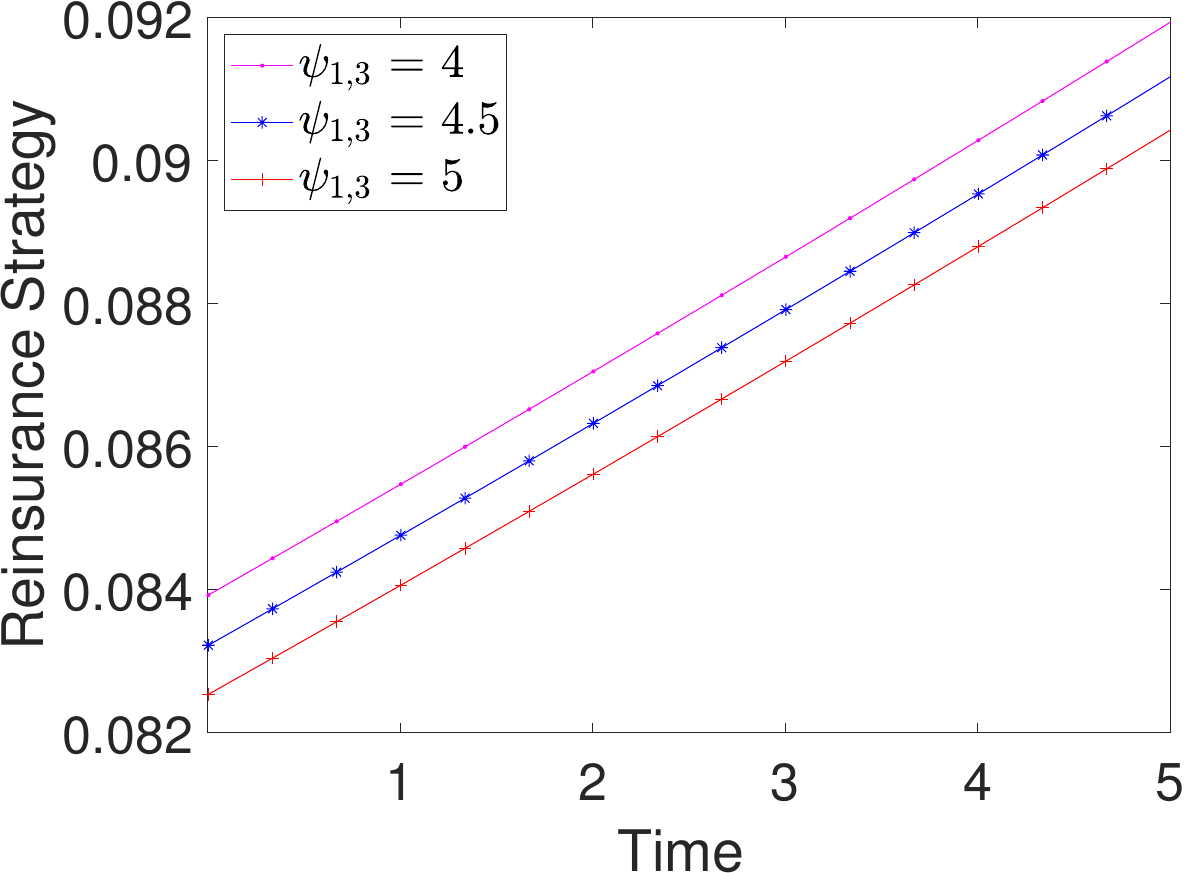}
	\caption{Effect of $ \psi_{1,3}$ on $ a^*_1(t) $.}
	\label{psi1_3_Reinsurance}
\end{minipage}
	\begin{minipage}[t]{0.45\linewidth} 
	\centering
	\includegraphics[width=0.9\linewidth]{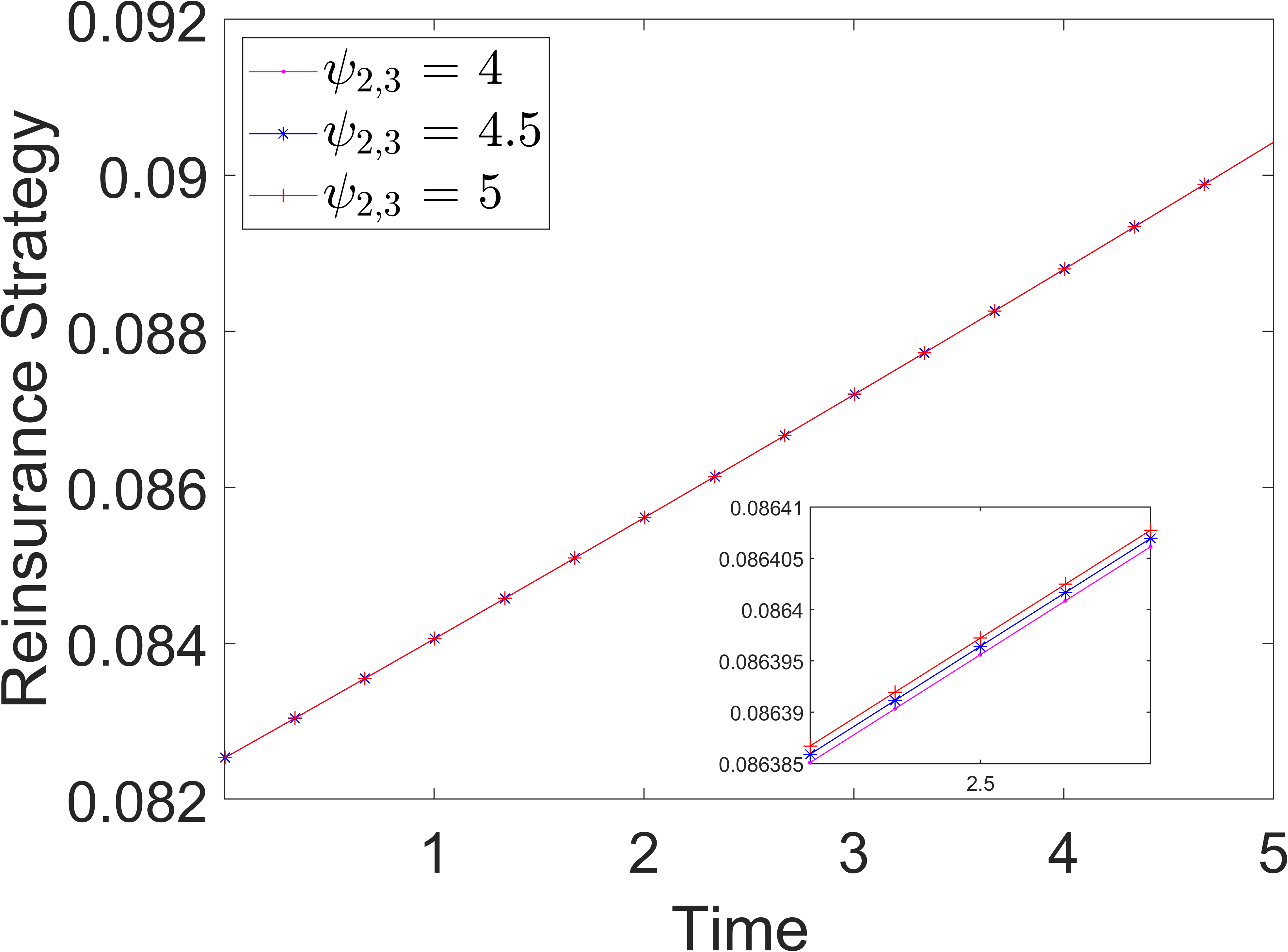}
	\caption{Effect of $ \psi_{2,3}$ on $ a^*_1(t) $.}
	\label{psi2_3_Reinsurance}
\end{minipage}
\begin{minipage}[t]{0.45\linewidth}
	\centering
	\includegraphics[width=0.9\linewidth]{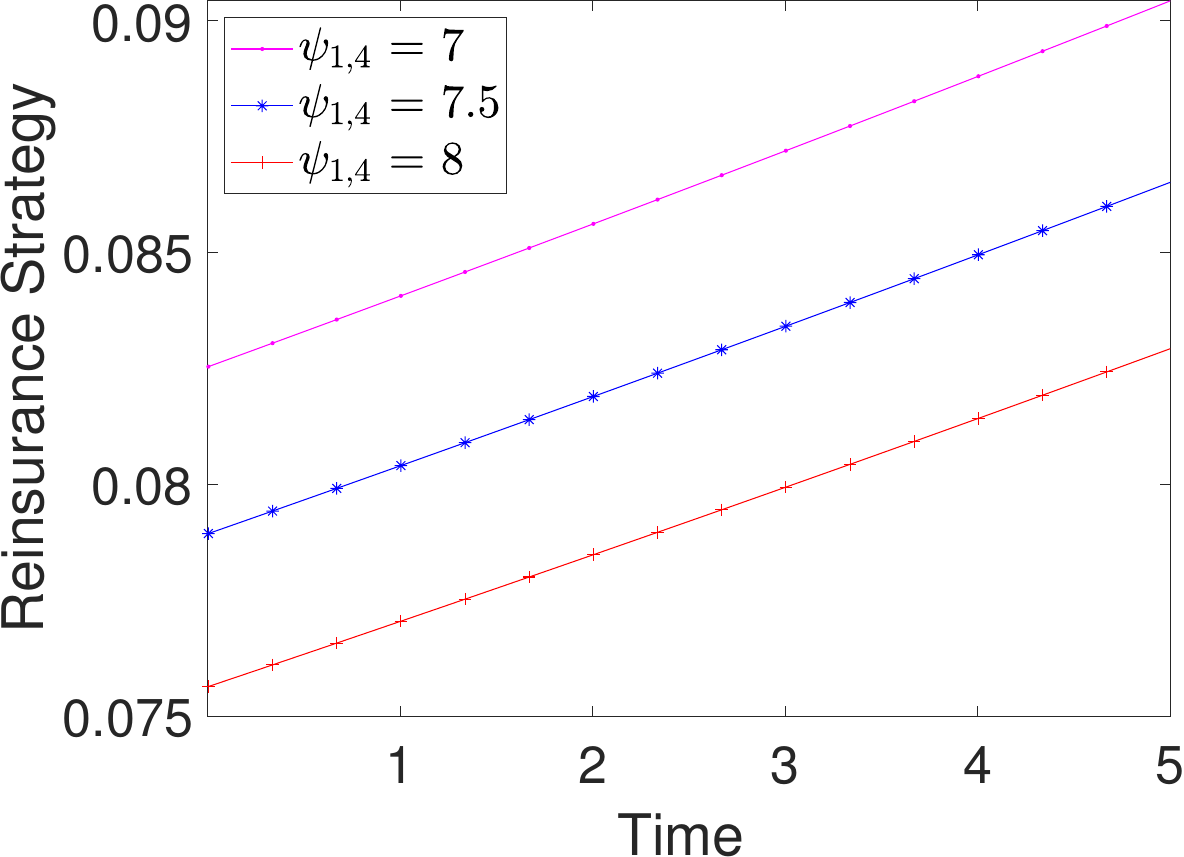}
	\caption{Effect of $ \psi_{1,4}$ on $ a^*_1(t) $.}
	\label{psi1_4_Reinsurance}
\end{minipage}
	\begin{minipage}[t]{0.45\linewidth} 
	\centering
	\includegraphics[width=0.9\linewidth]{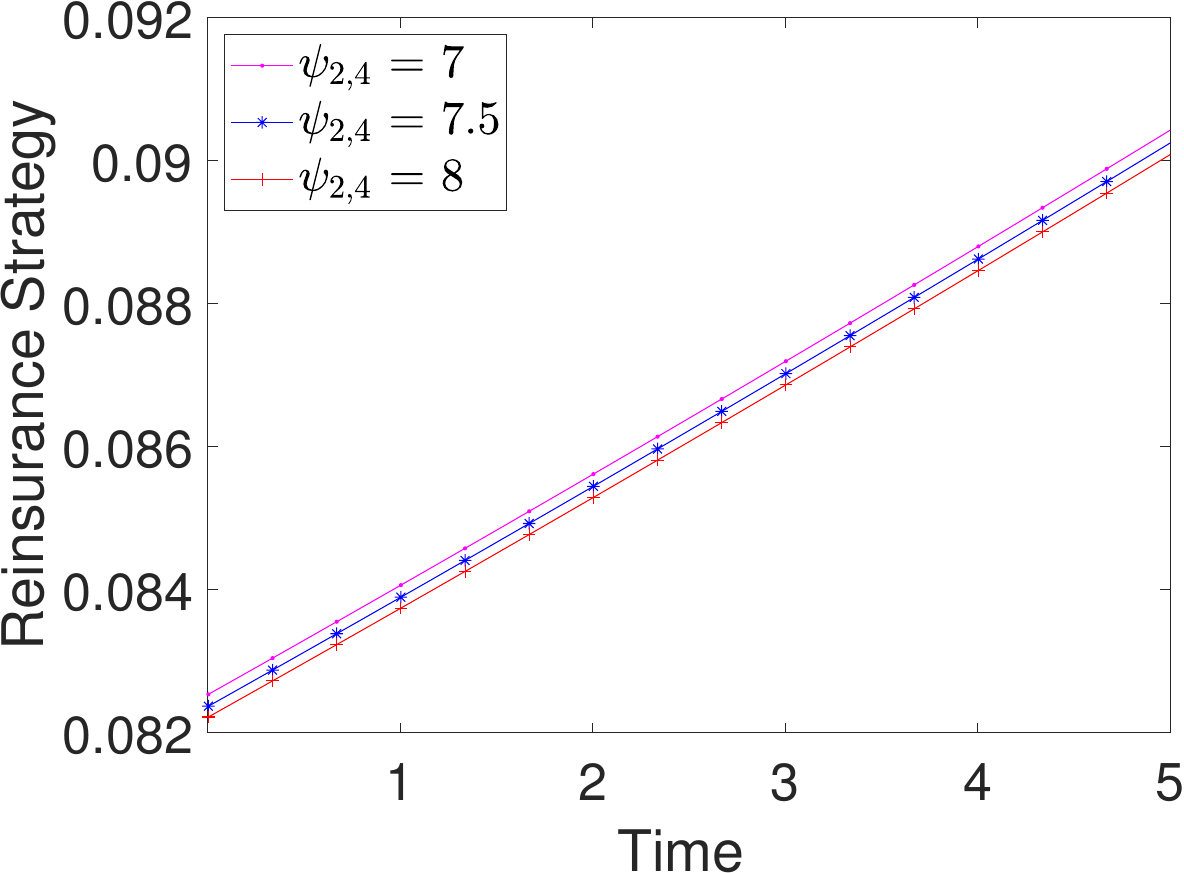}
	\caption{Effect of $ \psi_{2,4}$ on $ a^*_1(t) $.}
	\label{psi2_4_Reinsurance}
\end{minipage}
\end{figure}

\subsection{Robust equilibrium investment strategy}
Next, we present the sensitivity analyses of the robust equilibrium investment strategy. Due to the independence of the financial market and insurance market, the ambiguity attitudes over the insurance market have no effects on the robust equilibrium investment strategy, which can be observed in \eqref{n-opt}. 

The effects of competition on the robust equilibrium investment strategies are depicted in Figs.~\ref{theta_1_Portfolio} and \ref{theta_2_Portfolio}. Similar with Figs.~\ref{theta_1_Reinsurance} and \ref{theta_2_Reinsurance}, we observe that $\pi^*_1$ increases with $\theta_1$ and $\theta_2$. In the competition game, if one of the insurers in the system is more cared about the relative performance and takes risk-seeking strategy, the others will also allocate more to the risky asset. This phenomenon reveals that either in the insurance market or the financial market, competition makes the insurers more aggressive.  

Fig.~\ref{delta_1_Portfolio} describes the relationship between $\delta_1$ and $\pi_1^*$. Insurer 1 is more risk aversion when $\delta_1$ increases. Then insurer 1 will become conservative in the financial market and decrease allocation in the risky asset. Therefore, $\delta_1$ has a negative effect on insurer 1's robust equilibrium investment strategy. Besides, the effect of insurer 2's risk aversion attitude on insurer 1's investment strategy is presented in Fig.~\ref{delta_1_Portfolio}. Fig.~\ref{delta_2_Portfolio} shows that when insurer 2 is more aggressive/conservative, insurer 1 will also become more aggressive/conservative. The effects of $\delta_1$ and $\delta_2$ on $\pi^*_1$ (Figs.~\ref{delta_1_Portfolio} and \ref{delta_2_Portfolio}) are similar to Fig.~\ref{delta_1_Reinsurance} and \ref{delta_2_Reinsurance}. In the financial market, the behaviors of different insurers also tend to be convergent.  
	\begin{figure}[htbp]
	\centering
	\begin{minipage}[t]{0.45\linewidth}
		\centering
		\includegraphics[width=0.9\linewidth]{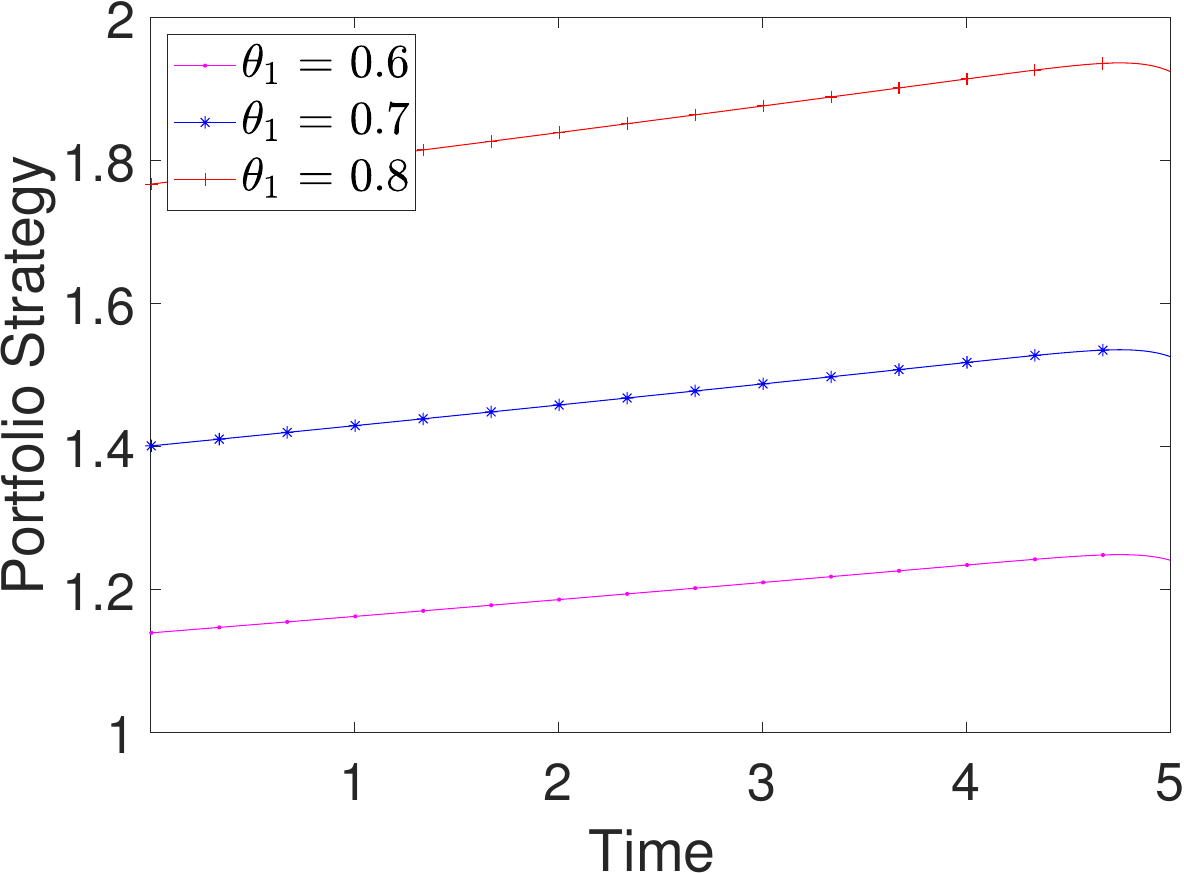}
		\caption{Effect of $ \theta_1$ on $ \pi^*_1(t) $.}
		\label{theta_1_Portfolio}
	\end{minipage}
\begin{minipage}[t]{0.45\linewidth}
	\centering
	\includegraphics[width=0.9\linewidth]{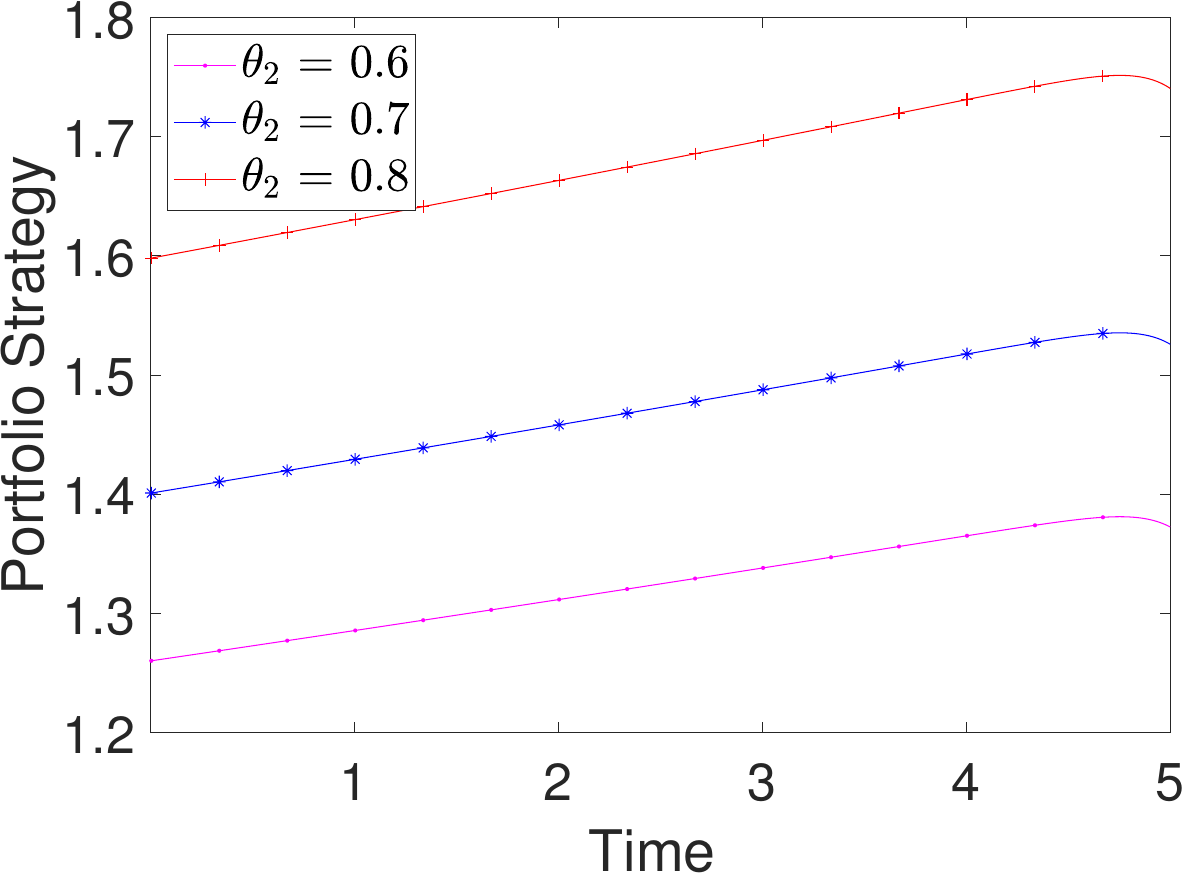}
	\caption{Effect of $\theta_2$ on $ \pi^*_1(t) $.}
	\label{theta_2_Portfolio}
\end{minipage}
\begin{minipage}[t]{0.45\linewidth}
	\centering
	\includegraphics[width=0.9\linewidth]{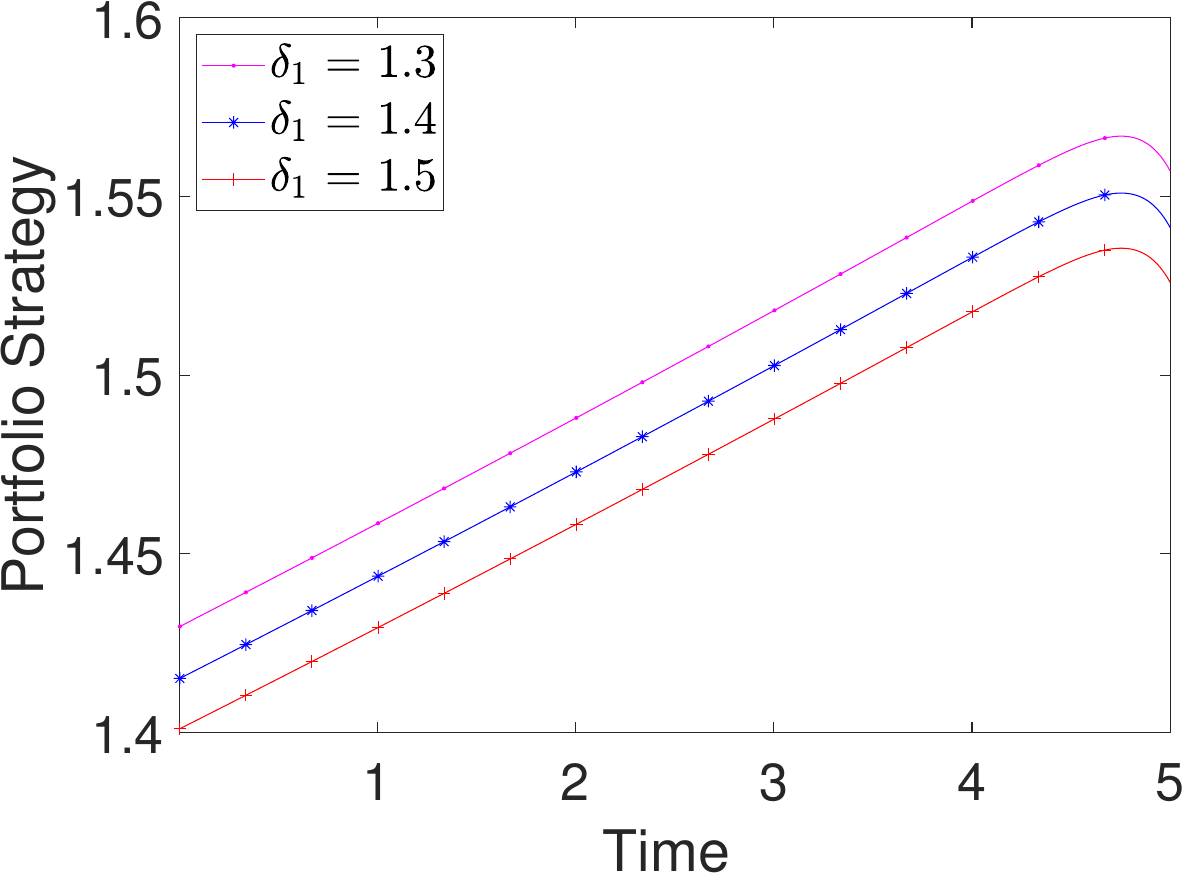}
	\caption{Effect of $ \delta_1$ on $ \pi^*_1(t) $.}
	\label{delta_1_Portfolio}
\end{minipage}
	\begin{minipage}[t]{0.45\linewidth}
	\centering
	\includegraphics[width=0.9\linewidth]{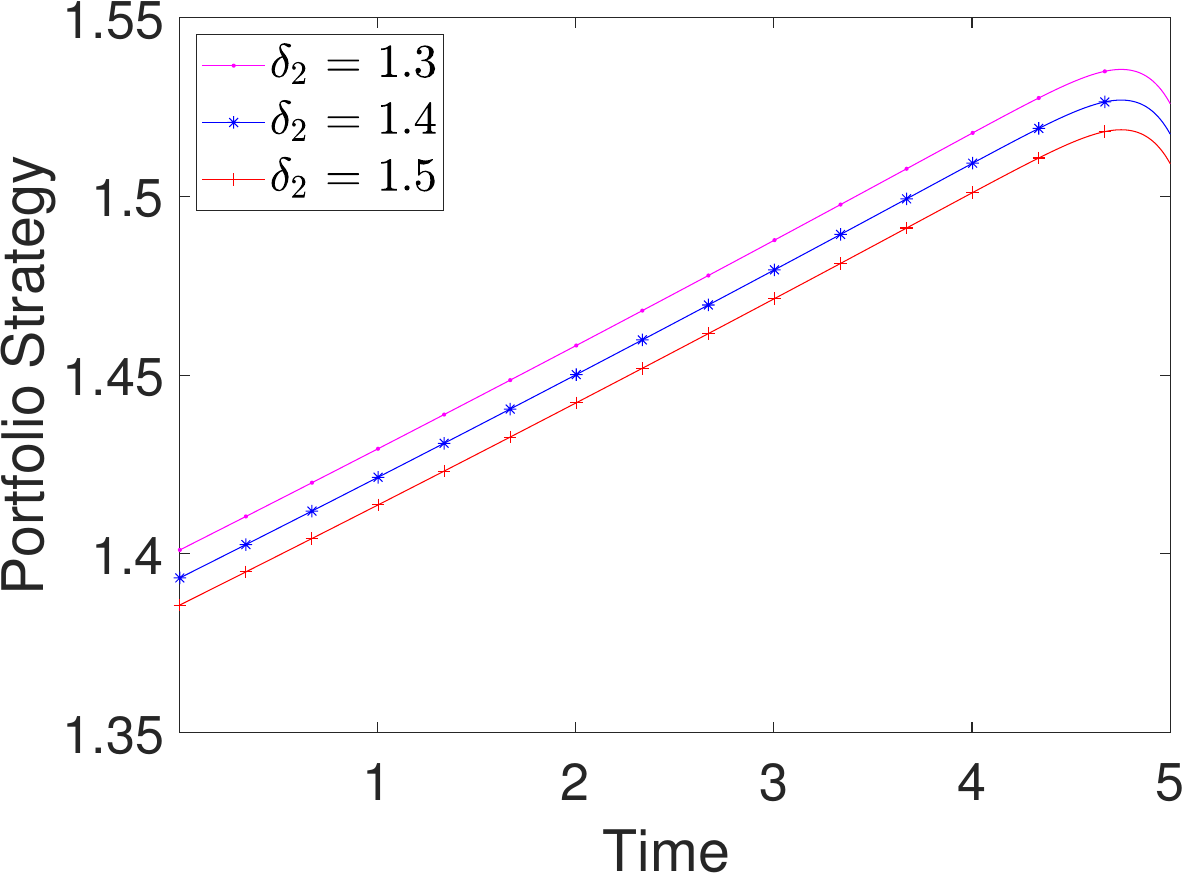}
	\caption{Effect of $\delta_2$ on $ \pi^*_1(t) $.}
	\label{delta_2_Portfolio}
\end{minipage}
\end{figure}

Fig.~\ref{psi1_1_Portfolio} illustrates the effect of ambiguity attitude towards equity risk on $\pi_1^*$. Obviously, if insurer 1 is more ambiguity aversion, she/he will have less faith in the financial market. Therefore, $\pi_1^*$ decreases with $\psi_{1,1}$. Meanwhile, insurer 2's ambiguity aversion coefficient $\psi_{2,1}$ has a similar influence on $\pi_1^*$. When $\psi_{2,1}$ increases, insurer 2 will undertake less equity risk and decrease the risky allocation. Simultaneously, insurer 1 will also decrease investment in the stock, which can be observed in Fig.~\ref{psi2_1_Portfolio}. Figs.~\ref{psi1_2_Portfolio} and \ref{psi2_2_Portfolio} depict the effects of ambiguity aversion attitudes towards volatility risk. We can observe from Fig.~\ref{psi1_2_Portfolio} that when insurer 1 is more ambiguity-averse interest risk, she/he will decrease the investment in the risky asset. Meanwhile, the ambiguity attitude of insurer 2 also has a similar effect on insurer 1's investment strategy. Fig.~\ref{psi2_2_Portfolio} shows that insurer 1's investment strategy decreases with $\psi_{2,2}$. This phenomenon again confirms the convergent behaviors of insurers in the competition environment, i.e., when others increase/decrease risk exposures, the insurer will also adjust the strategy in the same direction.
\begin{figure}[htbp]
	\centering
	\begin{minipage}[t]{0.45\linewidth} 
	\centering
	\includegraphics[width=0.9\linewidth]{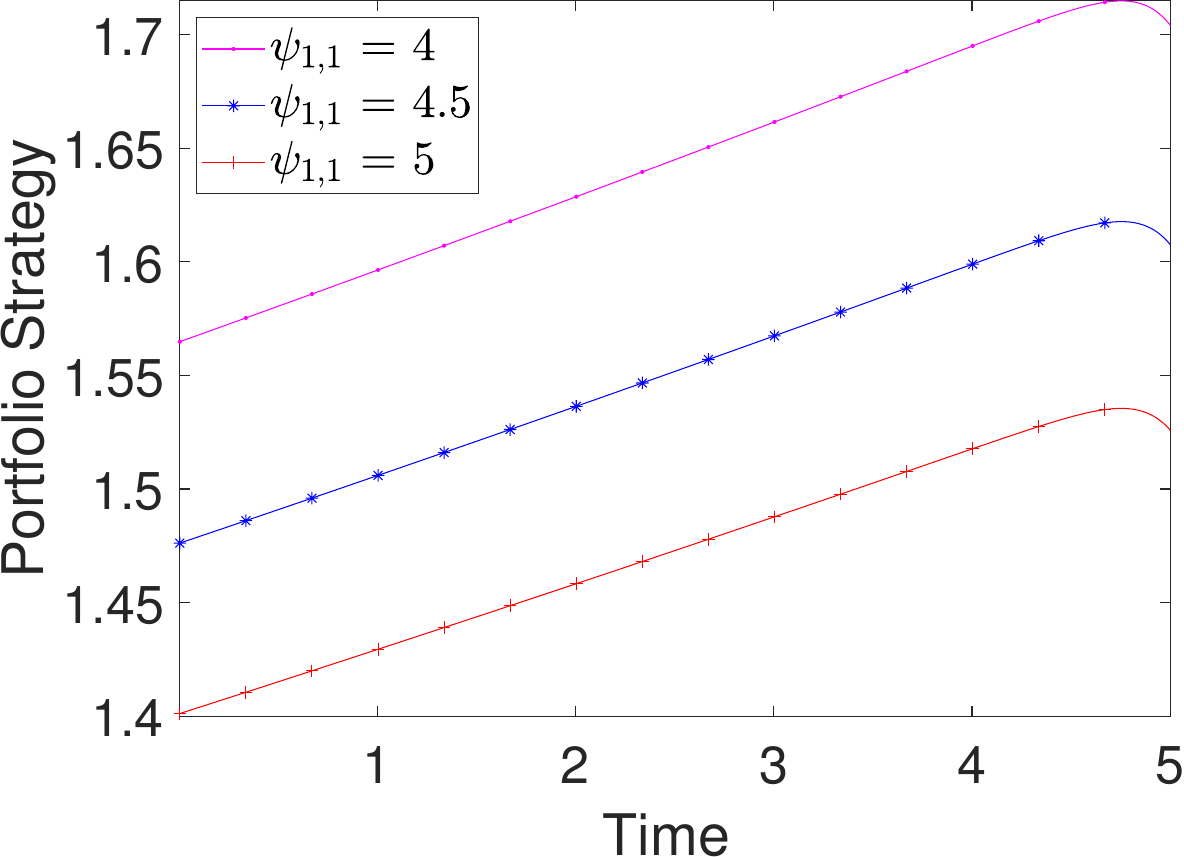}
	\caption{Effect of $ \psi_{1,1}$ on $ \pi^*_1(t) $.}
	\label{psi1_1_Portfolio}
\end{minipage}
	\begin{minipage}[t]{0.45\linewidth} 
	\centering
	\includegraphics[width=0.9\linewidth]{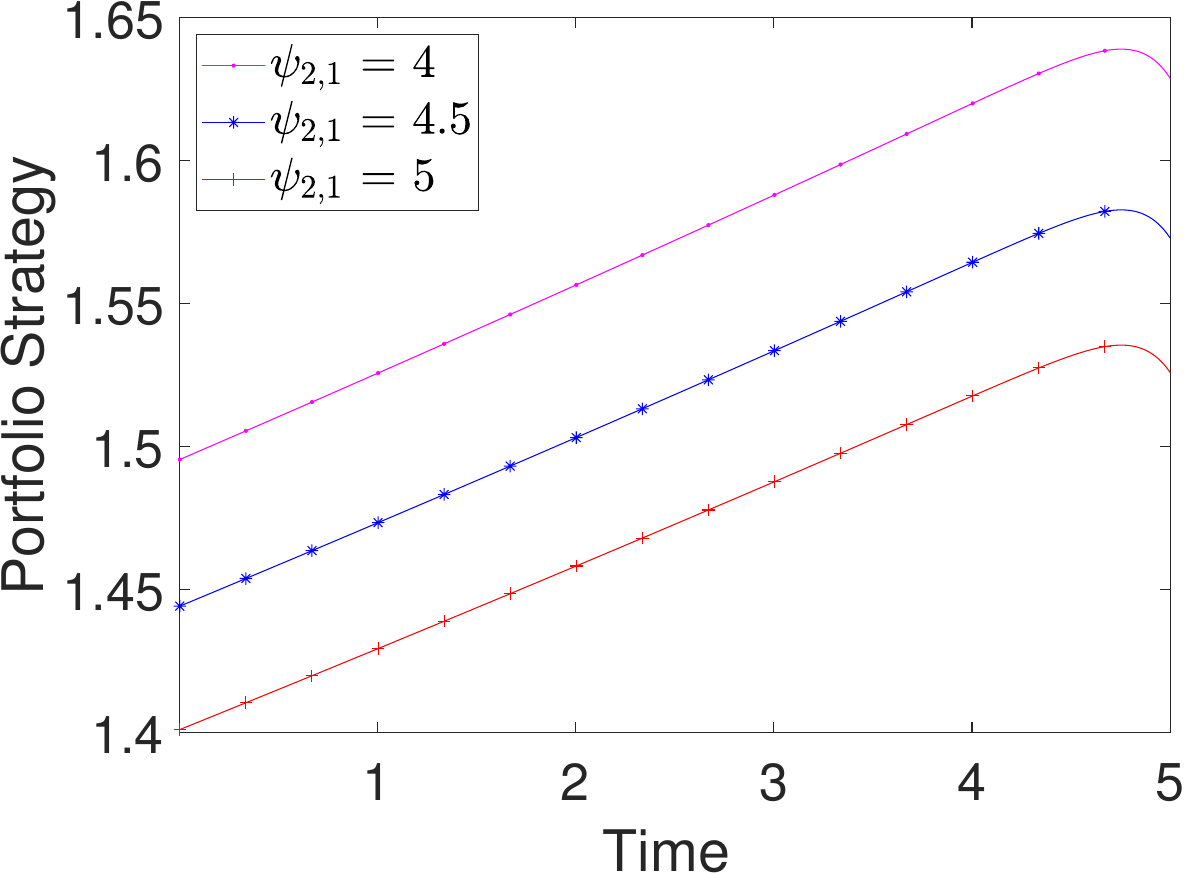}
	\caption{Effect of $ \psi_{2,1}$ on $ \pi^*_1(t) $.}
	\label{psi2_1_Portfolio}
\end{minipage}
	\begin{minipage}[t]{0.45\linewidth} 
	\centering
	\includegraphics[width=0.9\linewidth]{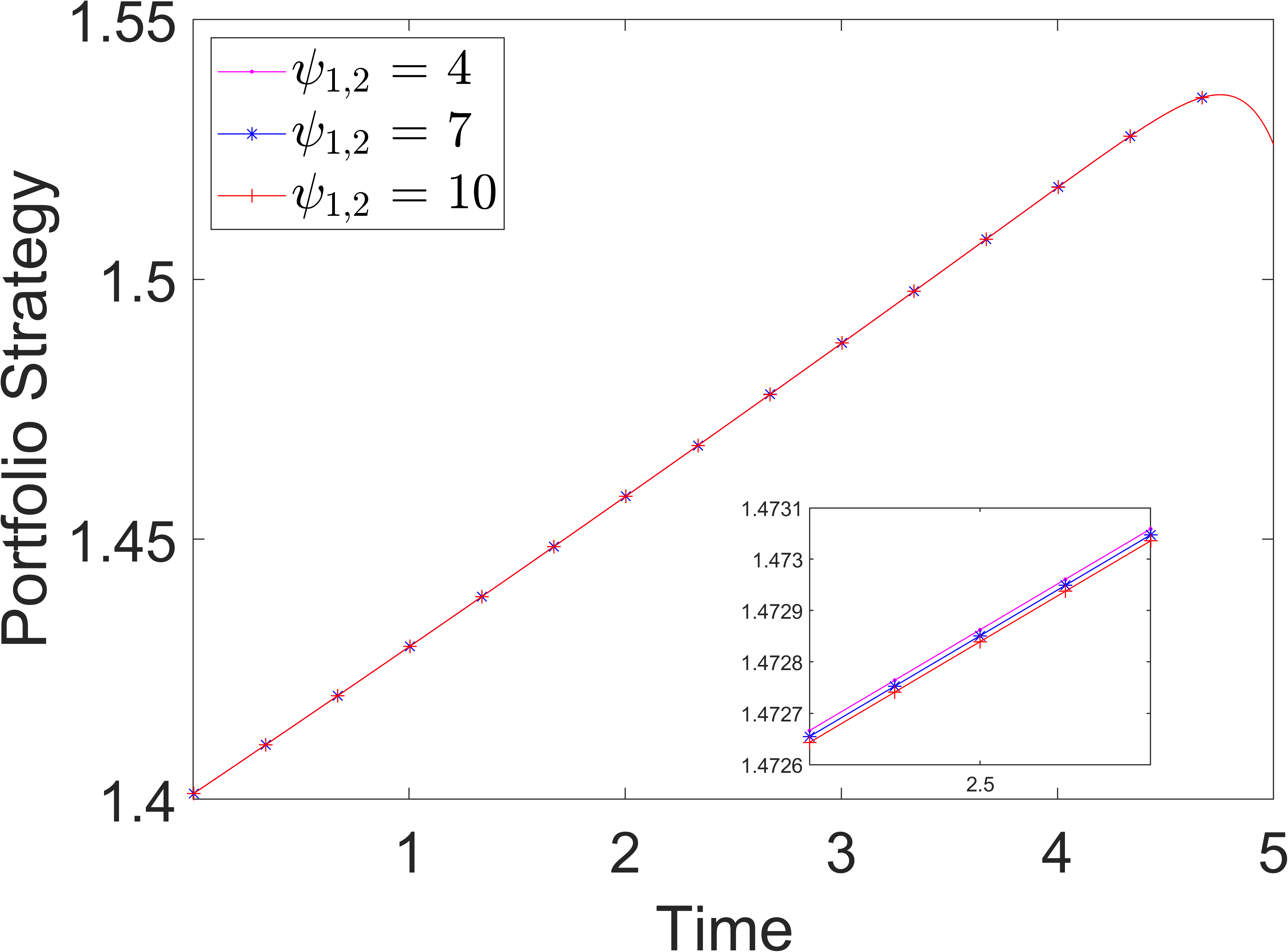}
	\caption{Effect of $ \psi_{1,1}$ on $ \pi^*_1(t) $.}
	\label{psi1_2_Portfolio}
\end{minipage}
\begin{minipage}[t]{0.45\linewidth} 
	\centering
	\includegraphics[width=0.9\linewidth]{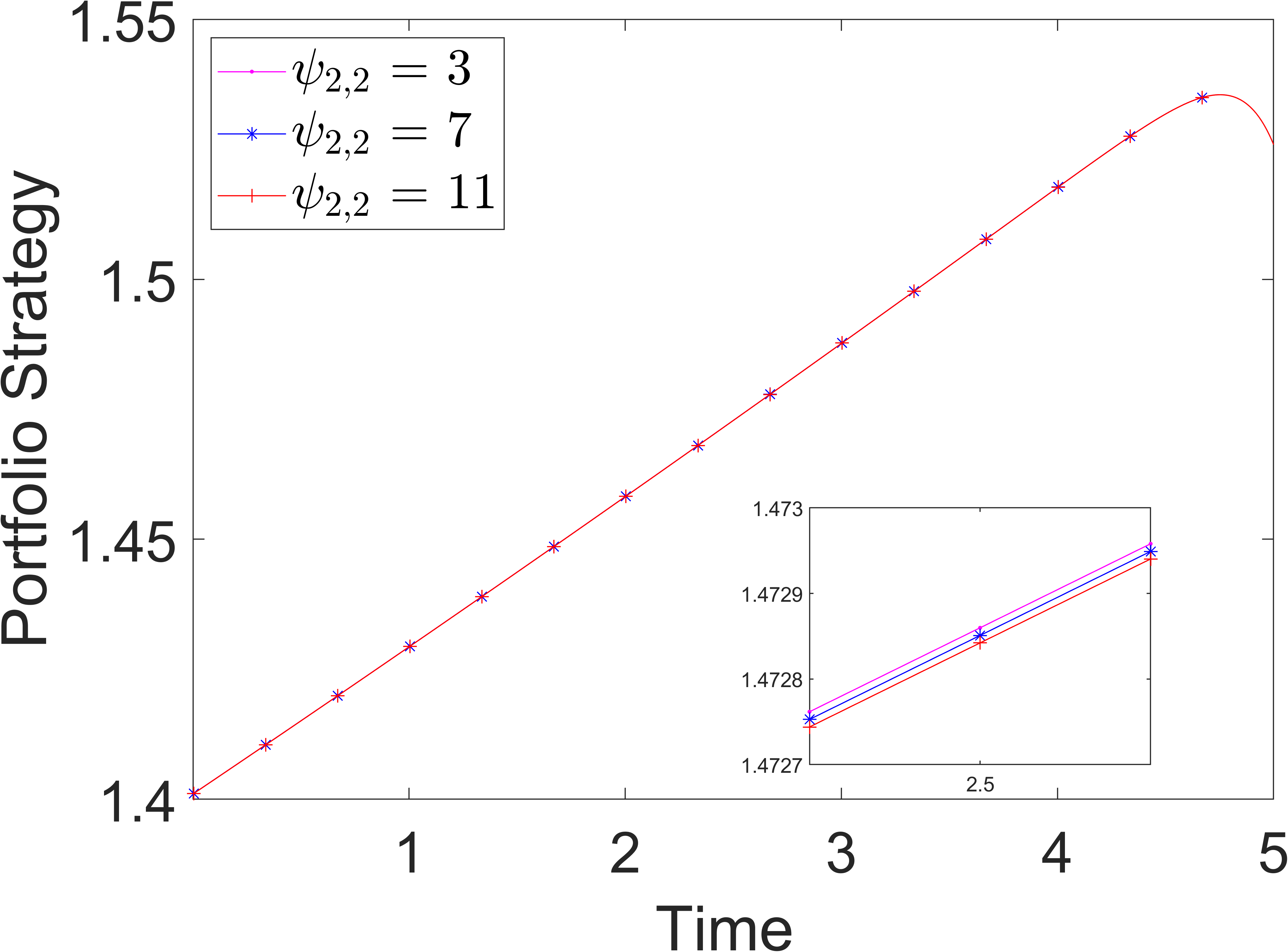}
	\caption{Effect of $ \psi_{2,2}$ on $ \pi^*_1(t) $.}
	\label{psi2_2_Portfolio}
\end{minipage}
\end{figure}

		\section{\bf Conclusion}
		
		In this paper, we investigate robust $n$-insurer and mean-field games for insurers under model uncertainty within incomplete markets. In the financial market, insurers allocate cash and stocks, with stock dynamics governed by the 4/2 stochastic volatility model. These insurers exhibit ambiguity aversion in the market and are driven by concerns about relative performance.
		
		Our work presents a solvable robust mean-field game in a non-linear financial system with a non-quadratic goal. The robust equilibrium strategies are determined through $n$-dimensional HJBI equations. We derive and rigorously verify the $n$-insurer equilibrium and mean-field equilibrium robust reinsurance and investment strategies. Notably, as $n$ approaches infinity, the $n$-insurer equilibrium converges to the mean-field equilibrium. Closed-form solutions for the $n$-agent and robust mean-field games are presented. Suitable conditions are outlined in the verification theorem.
		
		The robust equilibrium investment strategies comprise four components, all exhibiting nonlinear dependence on the risk aversion and ambiguity aversion coefficients. Two of these components stem from competition and are influenced by competition coefficients. Simultaneously, the robust reinsurance proportions consist of two components. Numerical results highlight the herd effect of competition on insurers, demonstrating their tendency to emulate each other's behaviors.
		
%

	\vskip 15pt
	{\bf Acknowledgements.} The authors acknowledge the support from the National Natural Science Foundation of China (Grant No.12271290, No.11901574, No.11871036), the MOE Project of Key Research Institute of Humanities and Social Sciences (22JJD910003). The authors thank the members of the group of Actuarial Science and Mathematical Finance at the Department of Mathematical Sciences, Tsinghua University for their feedbacks and useful conversations.
	\vskip 15pt
	\appendix
	\renewcommand{\theequation}{\thesection.\arabic{equation}}
	\section{\bf Proof of Proposition \ref{solution-HJBI}}\label{proof-solution-HJBI}
For  $ (Y_i(t), Z(t))=(y,z) $, \begin{equation}\label{HJBI1}
		\mathcal{A}^{\left\{(\pi_{k}^*,a_k^*)_{k\neq i},(\pi_{i},a_i)\right\},(\varphi_i,\chi_i,\phi_i,\vartheta_{i})}V(t,y,z)+  \left[ \frac{\varphi_{i}^{2}(t)}{2 \beta_{i,1}}	 +\frac{\chi_{i}^2(t)}{2 \beta_{i,2}}+   \frac{\phi_{i}^{2}(t)}{2 \beta_{i,3}} +\frac{\vartheta_{i}(t)^T\vartheta_{i}(t)}{2 \beta_{i,4}}\right]	V(t,  {y},z)
	\end{equation} is a quadratic function of $ \varphi_i$, $\chi_{i}$,  $\phi_i $ and  $ \vartheta_i $.   
The first-order condition of \eqref{HJBI-opt} is 
	\begin{equation*}
		\left\{\begin{aligned}
			&\nu\sqrt{z}\rho v^{(i)}_{z}+((1-\frac{\theta_{i}}{n})\pi_i-\frac{\theta_{i}}{n}\sum_{k\neq i}\pi^*_{k})\sigma v^{(i)}_y+\frac{\varphi^\circ_{{i,\pi_i,a_i}}}{\beta_{i,1}}v^{(i)}=0,\\
			&\nu\sqrt{z}\sqrt{1-\rho^2} v^{(i)}_{z}+\frac{\chi^\circ_{i,\pi_i,a_i}}{\beta_{i,2}}v^{(i)}=0,\\
			&\sqrt{\hat{\lambda}}(1-\frac{\theta_i}{n})\mu_{i1}a_{i}(t)v^{(i)}_y-\sqrt{\hat{\lambda}}\frac{\theta_{i}}{n}\sum_{k\neq i}\mu_{k1}a_{k}^*(t)v^{(i)}_y+\frac{\phi^\circ_{{i,\pi_i,a_i}}}{\beta_{i,3}}v^{(i)}=0,\\
			&(1-\frac{\theta_i}{n})a_{i}(t)\sqrt{(\hat{\lambda}+\lambda_{i})\mu_{i 2}-\hat{\lambda}\mu_{i 1}^2}v^{(i)}_y+\frac{\vartheta^\circ_{{i,i,\pi_i,a_i}}}{\beta_{i,4}}v^{(i)}=0,\\&-\frac{\theta_i}{n}a_{k}^*(t)\sqrt{(\hat{\lambda}+\lambda_{k})\mu_{k 2}-\hat{\lambda}\mu_{k 1}^2}v^{(i)}_y+\frac{\vartheta^\circ_{{i,k,\pi_i,a_i}}}{\beta_{i,4}}v^{(i)}=0,~k\neq i.
		\end{aligned}\right.
	\end{equation*}
By the last equations,  
we obtain $\left(\varphi_{i,\pi_i,a_i}^\circ, \chi_{i,\pi_i,a_i}^\circ,\phi_{i,\pi_i,a_i}^\circ, \vartheta_{i,\pi_i,a_i}^\circ\right)$ as follows. (which is expressed by $\left\{(\pi_{k}^*,a_k^*)_{k\neq i}\right\},(\pi_{i},a_i), v^{(i)}$, etc.)
\begin{equation*}
		\left\{\begin{aligned}
			&{\varphi^\circ_{i,\pi_i,a_i}}=-\nu\sqrt{z}\rho v^{(i)}_{z}/v^{(i)}\beta_{i,1}-((1-\frac{\theta_{i}}{n})\pi_i-\frac{\theta_{i}}{n}\sum_{k\neq i}\pi^*_{k})\sigma v^{(i)}_y/v^{(i)}\beta_{i,1},\\
			&{\chi^\circ_{i,\pi_i,a_i}}=-\nu\sqrt{z}\sqrt{1-\rho^2} v^{(i)}_{z}/v^{(i)}\beta_{i,2},\\
			&\phi^\circ_{{i,\pi_i,a_i}}=-\sqrt{\hat{\lambda}}\left[(1-\frac{\theta_i}{n})\mu_{i1}a_{i}(t)-\frac{\theta_{i}}{n}\sum_{k\neq i}\mu_{k1}a_{k}^*(t)\right]v^{(i)}_y/v^{(i)}{\beta_{i,3}},\\
			&\vartheta^\circ_{{i,i,\pi_i,a_i}}=-(1-\frac{\theta_i}{n})a_{i}(t)\sqrt{(\hat{\lambda}+\lambda_{i})\mu_{i 2}-\hat{\lambda}\mu_{i 1}^2}v^{(i)}_y/v^{(i)}{\beta_{i,4}},\\
			&\vartheta^\circ_{{i,k,\pi_i,a_i}}=\frac{\theta_i}{n}a_{k}^*(t)\sqrt{(\hat{\lambda}+\lambda_{k})\mu_{k 2}-\hat{\lambda}\mu_{k 1}^2}v^{(i)}_y/v^{(i)}{\beta_{i,4}},~k\neq i.
		\end{aligned}\right.
	\end{equation*}

	Substituting the expressions of $\left(\varphi_{i,\pi_i,a_i}^\circ, \chi_{i,\pi_i,a_i}^\circ,\phi_{i,\pi_i,a_i}^\circ, \vartheta_{i,\pi_i,a_i}^\circ\right)$ into \eqref{HJBI1}.  \eqref{HJBI1}  is also quadratic about $ \pi_{i}$ and $a_i $, so we obtain $\left(\hat{\pi}^\circ_i, \hat{a}^\circ_i\right)$ by the first-order conditions as follows
\begin{equation}\label{equ:a1}
		\left\{	\begin{aligned}
			&\nu\rho\sqrt{z}( v^{(i)}_{yz}-v^{(i)}_yv^{(i)}_{z}/v^{(i)}\beta_{i,1})+m\sqrt{z} v^{(i)}_y+\sigma((1-\frac{\theta_i}{n})\hat{\pi}^\circ_i-\frac{\theta_i}{n}\sum_{k\neq i}\pi^*_k)(v^{(i)}_{yy}-(v^{(i)}_y)^2/v^{(i)}\beta_{i,1})=0,\\
			&2\hat{\eta}(\lambda_i+\hat{\lambda})\mu_{i 2}(1-\hat{a}^\circ_i(t))v^{(i)}_y+\hat{\lambda}\frac{\theta_{i}}{n}\sum_{k\neq i}a^*_{k}(t)\mu_{i 1}\mu_{k 1}(\beta_{i,3}(v^{(i)}_y)^2/v^{(i)}-v^{(i)}_{yy})\\&+ (1-\frac{\theta_i}{n})\hat{a}^\circ_i(t)\left[(\lambda_i+\hat{\lambda})\mu_{i2}v^{(i)}_{yy}-({(\hat{\lambda}+\lambda_{i})\mu_{i 2}{\beta_{i,4}}+\hat{\lambda}\mu_{i 1}^2}({\beta_{i,3}}-{\beta_{i,4}}))(v^{(i)}_y)^2/v^{(i)}\right]=0
		\end{aligned}\right.
	\end{equation}

	We guess that\begin{equation*}
		v^{(i)}(t, y,z)= -\frac{1}{\delta_i} \exp\left(f_i(t)-\delta_i yg_i(t)+h_i(t)z\right)
	\end{equation*}where $ f_i(T)=h_i(T)=0$, $ g_i(T)=1 $.
	Then $ v^{(i)}_y=-\delta_{i}g_i(t)v^{(i)}$, $v^{(i)}_z=h_i(t)v^{(i)}$, $v^{(i)}_t=(f'_i(t)+h'_i(t)z-\delta_i yg'_i(t))v^{(i)} $, and \eqref{equ:a1} becomes
\begin{equation*}
	\left\{	\begin{aligned}
		&\nu\rho\sqrt{z}(1-\beta_{i,1})(-\delta_i g_i(t)h_i(t))+m\sqrt{z} (-\delta_i g_i(t))+\sigma((1-\frac{\theta_i}{n})\hat{\pi}^\circ_i-\frac{\theta_i}{n}\sum_{k\neq i}\pi^*_k)(1-\beta_{i,1})(\delta_i^2 g_i^2(t))=0,\\
		&\left(\frac{2}{\delta_{i}g_i(t)}\hat{\eta}(\lambda_i+\hat{\lambda})\mu_{i 2}+ (1-\frac{\theta_i}{n})\left[(\lambda_i+\hat{\lambda})\mu_{i2}(1-{\beta_{i,4}})+\hat{\lambda}\mu_{i 1}^2({\beta_{i,4}}-{\beta_{i,3}})\right]\right)\hat{a}^\circ_i(t)\\&=\frac{2}{\delta_{i}g_i(t)}\hat{\eta}(\lambda_i+\hat{\lambda})\mu_{i 2}+\hat{\lambda}\frac{\theta_{i}}{n}\sum_{k\neq i}a^*_{k}(t)\mu_{i 1}\mu_{k 1}(1-\beta_{i,3}).
	\end{aligned}\right.
\end{equation*}
By the last equation, we obtain
	\begin{equation*}
		\left\{\begin{aligned}
			&(1-\frac{\theta_i}{n})\hat{\pi}^\circ_i(t)-\frac{\theta_i}{n}\sum_{k\neq i}\pi^*_k(t)=(\frac{m}{1-\beta_{i,1}}+\nu\rho h_i(t))\frac{\sqrt{z}}{\sigma \delta_{i}g_i(t)},\\
			&\hat{a}^\circ_i(t)R_i(t)=Q_i(t)\frac{1}{n}\sum_{k\neq i}\mu_{k 1}a^*_k(t)+P_i(t),
		\end{aligned}\right.
	\end{equation*}
	where
	\begin{equation*}
	\begin{aligned}
		&	P_i(t)=\frac{2\hat{\eta}}{ \delta_{i}g_i(t)}(\lambda_i+\hat{\lambda})\mu_{i 2},\quad Q_i(t)=\hat{\lambda}\theta_i\mu_{i 1}(1-\beta_{i,3}),\\& R_i(t)=\frac{2}{\delta_{i}g_i(t)}\hat{\eta}(\lambda_i+\hat{\lambda})\mu_{i 2}+ (1-\frac{\theta_i}{n})\left[(\lambda_i+\hat{\lambda})\mu_{i2}(1-{\beta_{i,4}})+\hat{\lambda}\mu_{i 1}^2({\beta_{i,4}}-{\beta_{i,3}})\right].
\end{aligned}	\end{equation*}
	As $ ( \pi_{i},a_i) $ are restricted in the admissible set, we should adjust  $\left(\hat{\pi}^\circ_i, \hat{a}^\circ_i\right)$ which are obtained by the first-order condition
	to  ensure $\left({\pi}^\circ_i, {a}^\circ_i\right)\in\mathscr{U}_i$, then we have \begin{equation*}
		\left\{\begin{aligned}
			&{\pi}^\circ_i(t)=\frac{\theta_i}{n-\theta_i}\sum_{k\neq i}\pi^*_k(t)+\frac{n}{n-\theta_i}(\frac{m}{1-\beta_{i,1}}+\nu\rho h_i(t))\frac{\sqrt{z}}{\sigma \delta_{i}g_i(t)},\\
			&a_i^\circ(t)=(\hat{a}^\circ_i(t)\vee0)\wedge1\\
			&\qquad=\left(\left(\frac{Q_i(t)}{R_i(t)}\frac{1}{n}\sum_{k\neq i}\mu_{k 1}a_k^*(t)+\frac{P_i(t)}{R_i(t)}\right)\vee0\right)\wedge1\\
			&\qquad=\left(\frac{Q_i(t)}{R_i(t)}\frac{1}{n}\sum_{k\neq i}\mu_{k 1}a_k^*(t)+\frac{P_i(t)}{R_i(t)}\right)\wedge1,
		\end{aligned}\right.
	\end{equation*}
and
 \begin{equation*}
	\left\{\begin{aligned}
		&{\varphi^\circ_{i,\pi_i,a_i}}=-\nu\sqrt{z}\rho v^{(i)}_{z}/v^{(i)}\beta_{i,1}-((1-\frac{\theta_{i}}{n})\pi_i-\frac{\theta_{i}}{n}\sum_{k\neq i}\pi^*_{k})\sigma v^{(i)}_y/v^{(i)}\beta_{i,1},\\
		&{\chi^\circ_{i,\pi_i,a_i}}=-\nu\sqrt{z}\sqrt{1-\rho^2} v^{(i)}_{z}/v^{(i)}\beta_{i,2},\\
		&\phi^\circ_{{i,\pi_i,a_i}}=-\sqrt{\hat{\lambda}}\left[(1-\frac{\theta_i}{n})\mu_{i1}a_{i}(t)-\frac{\theta_{i}}{n}\sum_{k\neq i}\mu_{k1}a_{k}^*(t)\right]v^{(i)}_y/v^{(i)}{\beta_{i,3}},\\
		&\vartheta^\circ_{{i,i,\pi_i,a_i}}=-(1-\frac{\theta_i}{n})a_{i}(t)\sqrt{(\hat{\lambda}+\lambda_{i})\mu_{i 2}-\hat{\lambda}\mu_{i 1}^2}v^{(i)}_y/v^{(i)}{\beta_{i,4}},\\
		&\vartheta^\circ_{{i,k,\pi_i,a_i}}=\frac{\theta_i}{n}a_{k}^*(t)\sqrt{(\hat{\lambda}+\lambda_{k})\mu_{k 2}-\hat{\lambda}\mu_{k 1}^2}v^{(i)}_y/v^{(i)}{\beta_{i,4}},~k\neq i.
	\end{aligned}\right.
\end{equation*}
	\begin{equation*}
		\left\{\begin{aligned}
			&{\varphi^\circ_{i}(t)}={\varphi^\circ_{i,\pi_i,a_i}}(t)=-\nu\sqrt{z}\rho \beta_{i,1}h_i(t)+((1-\frac{\theta_{i}}{n})\pi_i^\circ(t)-\frac{\theta_{i}}{n}\sum_{k\neq i}\pi^*_{k}(t))\sigma \beta_{i,1}\delta_{i}g_i(t)=\frac{\beta_{i,1}}{1-\beta_{i,1}}m\sqrt{z},\\
			&{\chi^\circ_{i}(t)}={\chi^\circ_{i,\pi_i,a_i}}(t)=-\nu\sqrt{z}\sqrt{1-\rho^2} \beta_{i,2}h_i(t),\\
			&\phi^\circ_{{i}}(t)=\sqrt{\hat{\lambda}}\left[(1-\frac{\theta_i}{n})\mu_{i1}a^\circ_{i}(t)-\frac{\theta_{i}}{n}\sum_{k\neq i}\mu_{k1}a_{k}^*(t)\right]\delta_{i}g_i(t){\beta_{i,3}},\\
			&\vartheta^\circ_{{i,i}}(t)=(1-\frac{\theta_i}{n})a^\circ_{i}(t)\sqrt{(\hat{\lambda}+\lambda_{i})\mu_{i 2}-\hat{\lambda}\mu_{i 1}^2}\delta_{i}g_i(t){\beta_{i,4}},\\
			&\vartheta^\circ_{{i,k}}(t)=-\frac{\theta_i}{n}a_{k}^*(t)\sqrt{(\hat{\lambda}+\lambda_{k})\mu_{k 2}-\hat{\lambda}\mu_{k 1}^2}\delta_{i}g_i(t){\beta_{i,4}}.~k\neq i
		\end{aligned}\right.
	\end{equation*}
 
 { Under the compatible conditions (I) and (II), we see
 	\begin{align*}
 		\left|\frac{\varphi^\circ_{i}(t)}{Z(t)}\right|^2\leqslant \sup_{1\leqslant i\leqslant n}\frac{\beta_{i,1}^2}{\left(1-\beta_{i,1}\right)^2}m^2<\frac{\kappa^2}{2\nu^2},
\end{align*}\begin{align*}
 		\left|\frac{\chi^\circ_{i}(t)}{Z(t)}\right|^2\leqslant 	\sup_{1\leqslant i\leqslant n}c_{i,1}^2\frac{\left(1-e^{c_{i,2} T}\right)^2}{\left(1 + c_{i,3}e^{c_{i,2} T}\right)^2}\nu^2\beta_{i,2}^2(1-\rho^2)<\frac{\kappa^2}{2\nu^2}.
 	\end{align*}
As  \begin{equation*}
	\sup_{t\in[0,T]}|\phi^\circ_{{i}}(t)|<\infty,\quad\sup_{t\in[0,T]}|\vartheta^\circ_{{i,i}}(t)|<\infty,\quad\sup_{t\in[0,T]}|\vartheta^\circ_{{i,i}}(t)|<\infty,
\end{equation*}	 we see $\left({\varphi}^\circ_i, {\chi}^\circ_i, \phi_i^\circ, \vartheta_i^\circ \right)\in\mathscr{A}$ and $ \left(\pi_i^\circ,a_i^\circ\right)\in\mathscr{U}_i $.}

	Plugging the explicit form of $\left({\pi}^\circ_i, {a}^\circ_i\right)$ and $\left(\varphi^\circ_{i},\chi^\circ_{i},\phi^\circ_{i},\vartheta^\circ_{i}\right)$ into   \eqref{HJBI-v},  we have

\begin{equation*}
\begin{aligned}
	&-\frac{\beta_{i,1}}{2(1-\beta_{i,1})^2}m^2{z}-\frac{1}{2}\nu^2{z}{(1-\rho^2)} \beta_{i,2}h_i(t)^2-\frac{1}{2}{\hat{\lambda}}\left[(1-\frac{\theta_i}{n})\mu_{i1}a^\circ_{i}(t)-\frac{\theta_{i}}{n}\sum_{k\neq i}\mu_{k1}a_{k}^*(t)\right]^2\delta_{i}^2g_i^2(t){\beta_{i,3}}\\
	&-\frac{1}{2}(1-\frac{\theta_i}{n})^2(a^\circ_{i}(t))^2((\hat{\lambda}+\lambda_{i})\mu_{i 2}-\hat{\lambda}\mu_{i 1}^2)\delta_{i}^2g_i^2(t){\beta_{i,4}}\!-\!\frac{\theta_i^2}{2n^2}\sum_{k\neq i}(a_{k}^*(t))^2((\hat{\lambda}+\lambda_{k})\mu_{k 2}-\hat{\lambda}\mu_{k 1}^2)\delta_{i}^2g_i^2(t){\beta_{i,4}} \\
	=&\mathcal{A}^{\{(\pi_k,a_k)\}_{k=1}^n,(\varphi_i,\chi_i,\phi_i,\vartheta_{i})}v(t,y,z)/v(t,y,z)\\
	\end{aligned}
\end{equation*} 
	\begin{equation*}
\begin{aligned}
	=&f'_i(t)+h'_i(t)z-\delta_i yg'_i(t)\!+\!\left[{\kappa}(\bar{Z}\!-\!z)+{\nu}{z}{\rho}\frac{\beta_{i,1}}{1-\beta_{i,1}}m\!\right]h_i(t)\!-{\nu}^2\!{(1-{\rho}^2)}\beta_{i,2}h_i^2(t)z+\!\nu^2\frac{1}{2}zh_i^2(t)\\
	&-ry\delta_{i}g_i(t)-(1-\frac{\theta_i}{n})\left[\eta_{i}\left(\lambda_{i}+\hat{\lambda}\right) \mu_{i 1}-\hat{\eta}(1-a^\circ_i(t))^2\left(\lambda_{i}+\hat{\lambda}\right) \mu_{i 2}\right] \delta_{i}g_i(t)\\
	&+\frac{\theta_{i}}{n}\sum_{k\neq i}\left[\eta_{k}\left(\lambda_{k}+\hat{\lambda}\right) \mu_{k 1}-\hat{\eta}(1-a^*_k(t))^2\left(\lambda_{k}+\hat{\lambda}\right) \mu_{k 2}\right]\delta_{i}g_i(t)\\
	&+\!(1-\frac{\theta_i}{n})a_{i}^\circ(t)\!\left(\!2{\hat{\lambda}}\mu_{i1}\frac{\theta_{i}}{n}\sum_{k\neq i}\mu_{k1}a_{k}^*(t){\beta_{i,3}}\!-\!(1\!-\!\frac{\theta_i}{n})a_{i}^\circ(t)({(\hat{\lambda}+\lambda_{i})\mu_{i 2}{\beta_{i,4}}+\hat{\lambda}\mu_{i 1}^2}({\beta_{i,3}}\!-\!{\beta_{i,4}}))\!\right)\!\delta_i^2g_i^2(t)\\
	&-\frac{\theta_{i}}{n}\sum_{k\neq i}a^*_{k}(t) \left({\hat{\lambda}}\mu_{k1}\left[\frac{\theta_{i}}{n}\sum_{k\neq i}\mu_{k1}a_{k}^*(t)\right]{\beta_{i,3}}+\frac{\theta_i}{n}a_{k}^*(t){((\hat{\lambda}+\lambda_{k})\mu_{k 2}-\hat{\lambda}\mu_{k 1}^2)}{\beta_{i,4}} \right)\delta_i^2g_i^2(t)\\
		&+\frac{1}{2}(1\!-\!\frac{\theta_i}{n})^2(a_{i}^\circ(t))^2 {\left(\hat{\lambda}\!+\!\lambda_{i}\right) \mu_{i 2}}\delta_{i}^2g_i^2(t)+\frac{\theta_{i}^2}{2n^2}\sum_{k\neq i}(a^*_{k}(t) )^2\left(\left(\hat{\lambda}\!+\!\lambda_{k}\right) \mu_{k 2}-\hat{\lambda}\mu_{k1}^2\right)\delta_{i}^2g_i^2(t)\\
	&+\hat{\lambda}\frac{\theta_i^2}{2n^2}\left[\sum_{k\neq i}a^*_k(t)\mu_{k 1}\right]^2\!\!\delta_{i}^2g_i^2(t)-\hat{\lambda}\frac{\theta_{i}}{n}(1-\frac{\theta_i}{n})a^\circ_i(t)\sum_{k\neq i}a^*_{k}(t)\mu_{i 1}\mu_{k 1}\delta_{i}^2g_i^2(t)-\frac{1}{2}(\frac{m}{1\!-\!\beta_{i,1}}\!+\!\nu\rho h_i(t))^2z.
\end{aligned}
\end{equation*} 

	Compare the coefficients of $ z $ and $ y $, and based on the values of $ g_i(T) $, $ h_i(T) $ and $ f_i(T) $,	 we have

 \begin{equation*}
	\left\{\begin{aligned}
		g_i(t)=&e^{r(T-t)},\\
		h_i(t)=&c_{i,1}\frac{e^{c_{i,2} t}-e^{c_{i,2} T}}{e^{c_{i,2} t} + c_{i,3}e^{c_{i,2} T}},\\
		f_i(t)=&\int_{t}^{T}\!\left\{{\kappa}\bar{Z}h_i(s)-(1-\frac{\theta_i}{n})\left[\eta_{i}\left(\lambda_{i}+\hat{\lambda}\right) \mu_{i 1}-\hat{\eta}(1-a^\circ_i(s))^2\left(\lambda_{i}+\hat{\lambda}\right) \mu_{i 2}\right] \delta_{i}g_i(s)\right.\\&\quad+\frac{\theta_{i}}{n}\sum_{k\neq i}\left[\eta_{k}\left(\lambda_{k}+\hat{\lambda}\right) \mu_{k 1}-\hat{\eta}(1-a^*_k(s))^2\left(\lambda_{k}+\hat{\lambda}\right) \mu_{k 2}\right]\delta_{i}g_i(s)\\
		&\quad+\frac{1}{2}(1-\frac{\theta_i}{n})^2(a^\circ_{i}(s))^2((\hat{\lambda}+\lambda_{i})\mu_{i 2}(1-{\beta_{i,4}})+\hat{\lambda}\mu_{i 1}^2({\beta_{i,4}}-{\beta_{i,3}}))\delta_{i}^2e^{2r(T-s)}\\
		&\quad+\frac{\theta_i^2}{2n^2}\sum_{k\neq i}(a_{k}^*(s))^2((\hat{\lambda}+\lambda_{k})\mu_{k 2}-\hat{\lambda}\mu_{k 1}^2)(1-{\beta_{i,4}})\delta_{i}^2e^{2r(T-s)} \\
		&\quad+\hat{\lambda}\frac{\theta_i^2}{2n^2}\left[\sum_{k\neq i}a^*_k(t) \mu_{k 1}\right]^2 \delta_{i}^2e^{2r(T-s)}(1-{\beta_{i,3}})\\
		&\quad\left.-\hat{\lambda}\frac{\theta_{i}}{n}(1-\frac{\theta_i}{n})a^\circ_i(s)\sum_{k\neq i}a^*_{k}(s)\mu_{i 1}\mu_{k 1}\delta_{i}^2e^{2r(T-s)}(1-{\beta_{i,3}})\right\}\rd s,
	\end{aligned}\right.
\end{equation*}
	where $ c_{i,1}=\frac{{\kappa}+m\nu\rho +\sqrt{({\kappa}+m\nu\rho )^2+\frac{1-\beta_{i,2}}{1-\beta_{i,1}}\nu^2(1-\rho^2)m^2}}{\nu^2(1-\rho^2)(1-\beta_{i,2})}$, $c_{i,3}=2(1-\beta_{i,1})({\kappa}+m\nu\rho )c_{i,1}+1$, $c_{i,2}=\frac{c_{i,3}+1}{2(1-\beta_{i,1})c_{i,1}} $
	\section{\bf Proof of Proposition \ref{pi-circ}}\label{proof-pi-circ}
We first prove the following lemma to prove Proposition \ref{pi-circ}.

\begin{lemma}\label{Quad}
	For an indefinite quadratic form $$  F(x,y)= x^2 +axy+bx+cy+\frac{d}{2} y^2, d<0,$$
it's easy to see  that \begin{equation*}
z(x):=\arg	\sup_{y\in\mathbb{R}}F(x,y)=-\frac{ax+c}{d}
\end{equation*}
and \begin{equation*}
	x_0:=\arg\inf_{x\in\mathbb{R}}F(x,z(x))=-\frac{bd-ac}{2d-a^2},
\end{equation*}
if $x_0>1  $, then \begin{equation*}
	\hat{x}:=\arg\inf_{x\in[0,1]}F(x,z(x))=1
\end{equation*}
Let $ z_0= z(x_0)$, $ \hat{z}=z(\hat{x}) $, $ \bar{x}_0=\arg\inf_{x\in\mathbb{R}}F(x,z_0) $, $ \bar{x}=\arg\inf_{x\in[0,1]}F(x,\hat{z}) $,
then $ \bar{x}_0={x}_0 $, $ \bar{x}= \hat{x}$.
\end{lemma}
\begin{proof}
	It's easy to verify that $ \frac{\partial }{\partial y}F(x,y)=dy+ax+c $, so based on the first-order condition, we see that \begin{equation*}
	\arg	\sup_{y\in\mathbb{R}}F(x,y)=-\frac{ax+c}{d},
	\end{equation*}
Moreover, we observe\begin{equation*}
	F(x,z(x))=(1-\frac{a^2}{2d})x^2+(b-\frac{ac}{d})x-\frac{c^2}{2d}.
\end{equation*}
as $ 1-\frac{a^2}{2d}>0 $. Based on the property of quadratic function, we have \begin{equation*}
\arg\inf_{x\in\mathbb{R}}F(x,z(x))=-\frac{bd-ac}{2d-a^2},
\end{equation*} and if  $x_0>1  $, then \begin{equation*}
\hat{x}:=\arg\inf_{x\in[0,1]}F(x,z(x))=1.
\end{equation*}Furthermore, we see \begin{equation*}
F(x,z_0) = x^2+axz_0+bx+cz_0+\frac{d}{2}z_0^2,
\end{equation*}
so $$  \bar{x}_0=-\frac{az_0+b}{2}=-\frac{-a\frac{ax_0+c}{d}+b}{2}=-\frac{bd-ac}{2d-a^2}=x_0.  $$
And \begin{equation*}
	F(x,\hat{z}) = x^2+ax\hat{z}+bx+c\hat{z}+\frac{d}{2}\hat{z}^2
\end{equation*} is monotonically decreasing for $ x<\tilde{x} $, where $ \tilde{x}=-\frac{a\hat{z}+b}{2} $. When $ x_0=-\frac{bd-ac}{2d-a^2}>1 $, we see $ \hat{z}=-\frac{a+c}{d} $, and $ \tilde{x}=-\frac{a\hat{z}+b}{2}=\frac{a^2+ac-bd}{2d} >1$, so we see $ \bar{x}=1=\hat{x} $.
\end{proof}
As \begin{equation*}
	\mathcal{H}\left({\left\{(\pi_{k}^*,a_k^*)_{k\neq i},(\pi_{i},a_i)\right\},(\varphi_i,\chi_i,\phi_i,\vartheta_{i})},V,(\beta_{i,j}V)_{j=1}^4,(t,y,z)\right)
\end{equation*} is an indefinite quadratic form about $ (\pi_{i},a_i) $ and $ (\varphi_i,\chi_i,\phi_i,\vartheta_{i}) $, and\begin{equation}\nonumber
\begin{aligned}
	(\pi^\circ_i,a^\circ_i)=\arg	\sup _{(\pi_i,a_i)\in\mathscr{U}_i} \inf _{(\varphi_i,\chi_i,\phi_i,\vartheta_{i})\in\mathscr{A}}&\mathcal{H}\left({\left\{(\pi_{k}^*,a_k^*)_{k\neq i},(\pi_{i},a_i)\right\},(\varphi_i,\chi_i,\phi_i,\vartheta_{i})},V,(\beta_{i,j}V)_{j=1}^4,(t,y,z)\right),\\
	( \varphi^\circ_{i},\chi^\circ_{i},\phi_i^\circ,\vartheta_{i}^\circ)=\arg	 \inf _{(\varphi_i,\chi_i,\phi_i,\vartheta_{i})\in\mathscr{A}}&\mathcal{H}\left({\left\{(\pi_{k}^*,a_k^*)_{k\neq i},(\pi_{i}^\circ,a_i^\circ)\right\},(\varphi_i,\chi_i,\phi_i,\vartheta_{i})},V,(\beta_{i,j}V)_{j=1}^4,(t,y,z)\right),
\end{aligned}
\end{equation} based on Lemma \ref{Quad}, we see that \begin{equation*}
(\pi^\circ_{i},a^\circ_i)=\arg\sup _{(\pi_i,a_i)\in\mathscr{U}_i}\mathcal{H}\left({\left\{(\pi_{k}^*,a_k^*)_{k\neq i},(\pi_{i},a_i)\right\},(\varphi_i^\circ,\chi_i^\circ,\phi_i^\circ,\vartheta_{i}^\circ)},V,(\beta_{i,j}V)_{j=1}^4,(t,y,z)\right).
\end{equation*}
	\section{\bf Proof of Theorem \ref{Verification}}\label{proof-verification}
	For $ \left(\varphi_i(t), \chi_i(t),\phi_i(t),\vartheta_{i}(t)\right)\in\mathscr{A} $, denote $ \left(\varphi_i(t), \chi_i(t)\right)=\left(\hslash_i(t)\sqrt{Z(t)},\hbar_i(t)\sqrt{Z(t)}\right) $, under the measure $  \mathbb{Q}^{\varphi_i,\chi_i,\phi_i,\vartheta_{i}} $,	using It\^{o}'s  formula, we have\begin{equation*}
		\begin{aligned}
			&v^{(i)}(T, Y(T),Z(T))\\
			=&v^{(i)}(s, Y(s),Z(s))+\int_{s}^{T}\mathcal{A}^{\left\{(\pi_{k}^*,a_k^*)_{k\neq i},(\pi_{i},a_i)\right\},(\varphi_i,\chi_{i},\phi_i,\vartheta_{i})}v^{(i)}(t, Y(t),Z(t))\rd t\\
			&+\int_{s}^{T}v_z^{(i)}(t, Y(t),Z(t))\nu \sqrt{Z(t)}\rho\rd W^{ \mathbb{Q}^{\varphi_i,\chi_i,\phi_i,\vartheta_{i}}}(t)+\int_{s}^{T}v_z^{(i)}(t, Y(t),Z(t))\nu \sqrt{Z(t)}\sqrt{1-\rho^2}\rd B^{ \mathbb{Q}^{\varphi_i,\chi_i,\phi_i,\vartheta_{i}}}(t)\\
			&+\int_{s}^{T}v_y^{(i)}(t, Y(t),Z(t))\sqrt{\hat{\lambda}}\left(a_{i}(t)\mu_{i 1}-\frac{\theta_{i}}{n}\sum_{k=1}^na^*_{k}(t)\mu_{k1}\right)\rd  \tilde{W}(t)\\
			&+\int_{s}^{T}v_y^{(i)}(t, Y(t),Z(t))(1-\frac{\theta_i}{n})a_{i}(t)\sqrt{(\hat{\lambda}+\lambda_{i})\mu_{i 2}-\hat{\lambda}\mu_{i 1}^2}\rd \hat{W}_i(t)\\
			&+\int_{s}^{T}v_y^{(i)}(t, Y(t),Z(t))\left((1-\frac{\theta_{i}}{n})\pi_{i}(t)-\frac{\theta_{i}}{n}\sum_{k\neq i}\pi^*_{k}(t)\right){\Sigma}(t)\rd {W}^{ \mathbb{Q}^{\varphi_i,\chi_i,\phi_i,\vartheta_{i}}}(t)\\
			&-\int_{s}^{T}v_y^{(i)}(t, Y(t),Z(t))\frac{\theta_{i}}{n}\sum_{k\neq i}^na^*_{k}(t) \sqrt{(\hat{\lambda}+\lambda_{k})\mu_{k 2}-\hat{\lambda}\mu_{k 1}^2}\rd \hat{W}_k(t)
		\end{aligned}
	\end{equation*}
	
 and \begin{equation*}
		\begin{aligned}
			\rd Z(t)\!=&{\kappa}(\bar{Z}-Z(t))\rd t+{\nu}\sqrt{Z(t)}\left[{\rho}(\rd W^{ \mathbb{Q}^{\varphi_i,\chi_i,\phi_i,\vartheta_{i}}}(t)+\varphi_i(t)\rd t)+\sqrt{1-{\rho}^2}(\rd B^{ \mathbb{Q}^{\varphi_i,\chi_i,\phi_i,\vartheta_{i}}}(t)+\chi_i(t)\rd t)\right]\\
			=&(\kappa\bar{Z}\!-\!\left(\!{\kappa}\!-\!\nu \rho\hslash_i(t)\!-\!\nu \sqrt{1-{\rho}^2}\hbar_i(t)\!\right)\!Z(t))\rd t\!+\!{\nu}\sqrt{Z(t)}\!\left[{\rho}\rd W^{ \mathbb{Q}^{\varphi_i,\chi_i,\phi_i,\vartheta_{i}}}\!(t)\!+\!\sqrt{1-{\rho}^2}\rd B^{ \mathbb{Q}^{\varphi_i,\chi_i,\phi_i,\vartheta_{i}}}(t)\right]
		\end{aligned}
	\end{equation*}
	let $ \bar{\kappa}= {\kappa}-\nu (\rho+\sqrt{1-{\rho}^2})\mathcal{C} $, as $ \mathcal{C}^2<\frac{\kappa^2}{2\nu^2} $, then $\mathcal{C}< \frac{{\kappa}}{\nu (\rho+\sqrt{1-{\rho}^2})} $, then $ 0<\bar{\kappa}\leqslant {\kappa}-\nu \rho\hslash_i(t)-\nu \sqrt{1-{\rho}^2}\hbar_i(t)$. $\forall s\in[0,T].$. Let $ \hat{Z}^s(t) $ be determined by the CIR model as follows \begin{equation*}
		\rd \hat{Z}^s(t)\!=\!(\kappa\bar{Z}-\bar{\kappa}\hat{Z}^s(t))\rd t+{\nu}\sqrt{\hat{Z}^s(t)}\left[{\rho}\rd W^{ \mathbb{Q}^{\varphi_i,\chi_i,\phi_i,\vartheta_{i}}}(t)\!+\!\sqrt{1-{\rho}^2}\rd B^{ \mathbb{Q}^{\varphi_i,\chi_i,\phi_i,\vartheta_{i}}}(t)\right]\!,~~ t\in[s,T],\hat{Z}^s(s)={Z}(s).
	\end{equation*}
	Then by the comparison theorem, we see that \begin{equation*}
		Z(t)\leqslant\hat{Z}^s(t), \forall t\in[s,T],\quad a.s.
	\end{equation*}
	Thus,  
	\begin{equation*}
		\mathbb{E}^{ \mathbb{Q}^{\varphi_i,\chi_i,\phi_i,\vartheta_{i}}}_{s,y,z}\left[\int_{s}^{T}{Z(t)}\rd t\right]\leqslant\mathbb{E}^{ \mathbb{Q}^{\varphi_i,\chi_i,\phi_i,\vartheta_{i}}}_{s,y,z}\left[\int_{s}^{T}{\hat{Z}^s(t)}\rd t\right]=\int_{s}^Tze^{-\bar{\kappa}(t-s)}+\frac{\kappa}{\bar{\kappa}}\bar{Z}(1-e^{-\bar{\kappa}(t-s)})\rd t<\infty.
	\end{equation*}
	While for $ (\pi_i,a_i)\in\mathscr{U}_i$, $ \left\{(\pi_{k}^*,a_k^*)_{k\neq i}\right\}\in\mathscr{U}_{-i} $,  $  \exists \{\mathcal{C}_j\}_{j=1}^n\subseteq\mathbb{R}_+$, $ \mathcal{M}:=\sup_{1\leqslant j\leqslant n}\mathcal{C}_j<\infty, $ such that  $ \pi_k^*(t)=\ell_k(t)\frac{Z(t)}{aZ(t)+b}$, $|\ell_k(t)|\leqslant \mathcal{C}_k$, $\forall k\neq i$, $ \pi_i(t)=\ell_i(t)\frac{Z(t)}{aZ(t)+b}$, $|\ell_i(t)|\leqslant \mathcal{C}_i$, $ \forall t\in[0,T]$, $0\leqslant a_i\leqslant1$, $0\leqslant a_k^*\leqslant1 $ and it's easy to see that $$ \sup_{t\in[s,T]} |v_z^{(i)}(t, Y(t),Z(t))|<\infty,\quad\sup_{t\in[s,T]} |v_y^{(i)}(t, Y(t),Z(t))|<\infty.$$Then \begin{equation*}
		\begin{aligned}
			&\mathbb{E}^{ \mathbb{Q}^{\varphi_i,\chi_i,\phi_i,\vartheta_{i}}}_{s,y,z}\left[\int_{s}^{T}\left(v_y^{(i)}(t, Y(t),Z(t))\right)^2\left((1-\frac{\theta_{i}}{n})\pi_{i}(t)-\frac{\theta_{i}}{n}\sum_{k\neq i}\pi^*_{k}(t)\right)^2({\Sigma}(t))^2\rd t\right]\\
			\leqslant&\mathcal{M}^2\sup_{t\in[s,T]} |v_y^{(i)}(t, Y(t),Z(t))|^2	\mathbb{E}^{ \mathbb{Q}^{\varphi_i,\chi_i,\phi_i,\vartheta_{i}}}_{s,y,z}\left[\int_{s}^{T}{Z(t)}\rd t\right]<\infty.
		\end{aligned}
	\end{equation*}
and
 \begin{equation*}
		\mathbb{E}^{ \mathbb{Q}^{\varphi_i,\chi_i,\phi_i,\vartheta_{i}}}_{s,y,z}\left[\int_{s}^{T}\left(v_z^{(i)}(t, Y(t),Z(t))\right)^2{Z(t)}\rd t\right]\leqslant\! \sup_{t\in[s,T]} |v_z^{(i)}(t, Y(t),Z(t))|^2\mathbb{E}^{ \mathbb{Q}^{\varphi_i,\chi_i,\phi_i,\vartheta_{i}}}_{s,y,z}\!\left[\int_{s}^{T}{Z(t)}\rd t\right]\!<\infty.
	\end{equation*}We also see \begin{equation*}
		\left\{\begin{aligned}
			&\mathbb{E}\left[\int_{s}^{T}\left(v_y^{(i)}(t, Y(t),Z(t))\right)^2\left((1-\frac{\theta_{i}}{n})a_{i}(t)\mu_{i 1}-\frac{\theta_{i}}{n}\sum_{k\neq i}a^*_{k}(t)\mu_{k1}\right)^2\rd t\right]<\infty,\\
			&\mathbb{E}\left[\int_{s}^{T}\left(v_y^{(i)}(t, Y(t),Z(t))\right)^2\frac{1}{n}\sum_{k\neq i}(a_{k}^*(t))^2 {((\hat{\lambda}+\lambda_{k})\mu_{k 2}-\hat{\lambda}\mu_{k 1}^2)}\rd t\right]<\infty,\\
			&\mathbb{E}\left[\int_{s}^{T}\left(v_y^{(i)}(t, Y(t),Z(t))\right)^2a_{i}^2(t) {((\hat{\lambda}+\lambda_{i})\mu_{i 2}-\hat{\lambda}\mu_{i1}^2)}\rd t\right]<\infty.
		\end{aligned}\right.
	\end{equation*}
	
	Then $\forall \left(\varphi_i(t), \chi_i(t),\phi_i(t),\vartheta_{i}(t)\right)\in\mathscr{A} $, $ (\pi_i,a_i)\in\mathscr{U}_i$ and $ \left\{(\pi_{k}^*,a_k^*)_{k\neq i}\right\}\in\mathscr{U}_{-i} $, the following equality  holds:  \begin{equation*}
		\mathbb{E}^{ \mathbb{Q}^{\varphi_i,\chi_i,\phi_i,\vartheta_{i}}}_{s,y,z}\!\left[v^{(i)}(T, Y(T),Z(T))\right]
		\!=\!v^{(i)}(s, y,z)+\mathbb{E}^{ \mathbb{Q}^{\varphi_i,\chi_i,\phi_i,\vartheta_{i}}}_{s,y,z}\!\!\int_{s}^{T}\!\!\!\mathcal{A}^{\left\{(\pi_{k}^*,a_k^*)_{k\neq i},(\pi_{i},a_i)\right\},(\varphi_i,\chi_{i},\phi_i,\vartheta_i)}(t, Y(t),Z(t))\rd t.
	\end{equation*}
	On the one hand, as \begin{equation*}
		\inf _{(\varphi_i,\chi_i,\phi_i,\vartheta_{i})}\left\{\mathcal{A}^{\left\{(\pi_{k}^*,a_k^*)_{k\neq i},(\pi^\circ_{i},a^\circ_i)\right\},(\varphi_i,\chi_{i},\phi_i,\vartheta_{i})}v^{(i)}(t,y,z)+  \frac{\varphi_{i}^{2}(t)}{2 \Psi_1^{i}} +\frac{\chi_{i}^2(t)}{2 \Psi^{i}_2} +\frac{\phi_{i}^{2}(t)}{2 \Psi_3^{i}} +\frac{\vartheta_{i}(t)^T\vartheta_{i}(t)}{2 \Psi^{i}_4}\right\} =0,
	\end{equation*}
	we see \begin{equation*}
		-\mathcal{A}^{\left\{(\pi_{k}^*,a_k^*)_{k\neq i},(\pi^\circ_{i},a^\circ_i)\right\},(\varphi_i,\chi_{i},\phi_i,\vartheta_{i})}v^{(i)}(t,y,z)\leqslant \frac{\varphi_{i}^{2}(t)}{2 \Psi_1^{i}} +\frac{\chi_{i}^2(t)}{2 \Psi^{i}_2}+\frac{\phi_{i}^{2}(t)}{2 \Psi_3^{i}} +\frac{\vartheta_{i}(t)^T\vartheta_{i}(t)}{2 \Psi^{i}_4}.
	\end{equation*}
	Then \begin{equation*}
		\begin{aligned}
			&v^{(i)}(s, y,z)\\
			=&	-\mathbb{E}^{ \mathbb{Q}^{\varphi_i,\chi_i,\phi_i,\vartheta_{i}}}_{s,y,z}\!\!\int_{s}^{T}\!\!\mathcal{A}^{\left\{(\pi_{k}^*,a_k^*)_{k\neq i},(\pi^\circ_{i},a^\circ_i)\right\},(\varphi_i,\chi_{i},\phi_i,\vartheta_{i})}v^{(i)}(t, Y(t),Z(t))\rd t+\mathbb{E}^{ \mathbb{Q}^{\varphi_i,\chi_i,\phi_i,\vartheta_{i}}}_{s,y,z}\left[v^{(i)}(T, Y(T),Z(T))\right]\\
			\leqslant&\mathbb{E}^{ \mathbb{Q}^{\varphi_i,\chi_i,\phi_i,\vartheta_{i}}}_{s,y,z}\int_{s}^{T} \frac{\varphi_{i}^{2}(t)}{2 \Psi_1^{i}} +\frac{\chi_{i}^2(t)}{2 \Psi^{i}_2} +\frac{\phi_{i}^{2}(t)}{2 \Psi_3^{i}} +\frac{\vartheta_{i}(t)^T\vartheta_{i}(t)}{2 \Psi^{i}_4}\mathrm{d}t +\mathbb{E}^{ \mathbb{Q}^{\varphi_i,\chi_i,\phi_i,\vartheta_{i}}}_{s,y,z}\left[U_f\left(X_{i}(T) {-\theta_{i}}\bar{X}(T) ; \delta_{i}\right)\right]\\
			=&J_{i}\left(\left\{(\pi_{k}^*,a_k^*)_{k\neq i},(\pi^\circ_{i},a^\circ_i)\right\},(\varphi_i,\chi_{i},\phi_i,\vartheta_{i}),s, y,z\right),
		\end{aligned}
	\end{equation*}
where $ X_i(t) $ and $ \bar{X}(t) $ are driven by $ \left\{(\pi_{k}^*,a_k^*)_{k\neq i},(\pi_{i}^\circ,a_i^\circ)\right\} $.
Therefore,
 $$\begin{aligned}
		& v^{(i)}(s, y,z)\\
		\leqslant& \inf_{(\varphi_i,\chi_i,\phi_i,\vartheta_{i})}J_{i}\left(\left\{(\pi_{k}^*,a_k^*)_{k\neq i},(\pi^\circ_{i},a^\circ_i)\right\},(\varphi_i,\chi_i,\phi_i,\vartheta_{i}),s, y,z\right)\\
		\leqslant &\sup _{(\pi_i,a_i)}\inf_{(\varphi_i,\chi_i,\phi_i,\vartheta_{i})}J_{i}\left(\left\{(\pi_{k}^*,a_k^*)_{k\neq i},(\pi_{i},a_i)\right\},(\varphi_i,\chi_{i},\phi_i,\vartheta_{i}),s, y,z\right)\\
		=&V^{(i)}(s,  {y},z).
	\end{aligned}$$
	On the other hand, based on  Proposition \ref{pi-circ}, we see \begin{equation*}
		-\mathcal{A}^{\left\{(\pi_{k}^*,a_k^*)_{k\neq i},(\pi_{i},a_i)\right\},(\varphi_i^\circ,\chi_{i}^\circ,\phi_i^\circ,\vartheta_{i}^\circ)}v^{(i)}(t,y,z)\geqslant  \frac{(\varphi_{i}^\circ(t))^2}{2 \Psi_1^{i}} +\frac{(\chi_{i}^\circ(t))^2}{2 \Psi^{i}_2}+\frac{(\phi_{i}^\circ(t))^2}{2 \Psi_3^{i}} +\frac{\vartheta_{i}^\circ(t)^T\vartheta_{i}^\circ(t)}{2 \Psi^{i}_4},
	\end{equation*}
	then \begin{equation*}
		\begin{aligned}
			&	v^{(i)}(s, y,z)\\
			=&	-\mathbb{E}^{ \mathbb{Q}^{\varphi_i^\circ,\chi_{i}^\circ,\phi_i^\circ,\vartheta_{i}^\circ}}_{s,y,z}\int_{s}^{T}\mathcal{A}^{\left\{(\pi_{k}^*,a_k^*)_{k\neq i},(\pi_{i},a_i)\right\},(\varphi_i^\circ,\chi_{i}^\circ,\phi_i^\circ,\vartheta_{i}^\circ)}(t, Y(t),Z(t))\rd t\!+\!\mathbb{E}^{ \mathbb{Q}^{\varphi_i^\circ,\chi_{i}^\circ,\phi_i^\circ,\vartheta_{i}^\circ}}_{s,y,z}\left[v^{(i)}(T, Y(T),Z(T))\right]\\
			\geqslant&\mathbb{E}^{ \mathbb{Q}^{\varphi_i^\circ,\chi_{i}^\circ,\phi_i^\circ,\vartheta_{i}^\circ}}_{s,y,z}\!\!\int_{s}^{T}\!\!\frac{(\varphi_{i}^\circ(t))^2}{2 \Psi_1^{i}} \!+\!\frac{(\chi_{i}^\circ(t))^2}{2 \Psi^{i}_2}\!+\!\frac{(\phi_{i}^\circ(t))^2}{2 \Psi_3^{i}} \!+\!\frac{\vartheta_{i}^\circ(t)^T\vartheta_{i}^\circ(t)}{2 \Psi^{i}_4}\mathrm{d}t\! +\!\mathbb{E}^{ \mathbb{Q}^{\varphi_i^\circ,\chi_{i}^\circ,\phi_i^\circ,\vartheta_{i}^\circ}}_{s,y,z}\!\left[U_f\left(X_{i}(T) {-\theta_{i}}\bar{X}(T) ; \delta_{i}\right)\right]\\
			=&J_{i}\left(\left\{(\pi_{k}^*,a_k^*)_{k\neq i},(\pi_{i},a_i)\right\},(\varphi_i^\circ,\chi_{i}^\circ,\phi_i^\circ,\vartheta_{i}^\circ),s, y,z\right),
		\end{aligned}	
	\end{equation*}
where $ X_i(t) $ and $ \bar{X}(t) $ are driven by $ \left\{(\pi_{k}^*,a_k^*)_{k\neq i},(\pi_{i},a_i)\right\} $.

	Then $$\begin{aligned}
		& v^{(i)}(s, y,z)\\
		\geqslant& \sup_{(\pi_{i},a_i)}J_{i}\left(\left\{(\pi_{k}^*,a_k^*)_{k\neq i},(\pi_{i},a_i)\right\},(\varphi_i^\circ,\chi_{i}^\circ,\phi_i^\circ,\vartheta_{i}^\circ),s, y,z\right)\\
		\geqslant &\sup _{(\pi_i,a_i)}\inf_{(\varphi_i,\chi_i,\phi_i,\vartheta_{i})}J_{i}\left(\left\{(\pi_{k}^*,a_k^*)_{k\neq i},(\pi_i,a_i)\right\},\left(\varphi_i,\chi_i,\phi_i,\vartheta_{i}\right)s, y,z\right)\\
		=&V^{(i)}(s,  {y},z).
	\end{aligned}
	$$
	Thus we have\begin{equation*}
		v^{(i)}(s, y,z)=V^{(i)}(s,  {y},z),\quad (\pi^*_{i},a^*_i)=(\pi^\circ_{i},a^\circ_i),\quad \forall s,y,z, 1\leqslant i\leqslant n.
	\end{equation*}that is, $ \forall 1\leqslant i\leqslant n $, we obtain \begin{equation}\label{pi_star}
	{\pi}^*_i(t)=\frac{\theta_i}{n-\theta_i}\sum_{k\neq i}\pi^*_k(t)+\frac{n}{n-\theta_i}(\frac{m}{1-\beta_{i,1}}+\nu\rho h_i(t))\frac{\sqrt{Z(t)}}{\sigma \delta_{i}g_i(t)},
\end{equation} and
	\begin{equation*}
		a_i^*(t)=\left(\frac{Q_i(t)}{R_i(t)}\frac{1}{n}\sum_{k\neq i}\mu_{k 1}a_k^*(t)+\frac{P_i(t)}{R_i(t)}\right)\wedge1,
	\end{equation*}
	\begin{equation*}
		\left\{\begin{aligned}
			&{\varphi^*_{i}(t)}=\frac{\beta_{i,1}}{1-\beta_{i,1}}m\sqrt{Z(t)},\\
			&{\chi^*_{i}(t)}=-\nu\sqrt{1-\rho^2} \beta_{i,2}h_i(t)\sqrt{Z(t)},\\
			&\phi^*_{{i}}(t)=\sqrt{\hat{\lambda}}\left[(1-\frac{\theta_i}{n})\mu_{i1}a^*_{i}(t)-\frac{\theta_{i}}{n}\sum_{k\neq i}\mu_{k1}a_{k}^*(t)\right]\delta_{i}g_i(t){\beta_{i,3}},\\
			&\vartheta^*_{{i,i}}(t)=(1-\frac{\theta_i}{n})a^*_{i}(t)\sqrt{(\hat{\lambda}+\lambda_{i})\mu_{i 2}-\hat{\lambda}\mu_{i 1}^2}\delta_{i}g_i(t){\beta_{i,4}},\\
			&\vartheta^*_{{i,k}}(t)=-\frac{\theta_i}{n}a_{k}^*(t)\sqrt{(\hat{\lambda}+\lambda_{k})\mu_{k 2}-\hat{\lambda}\mu_{k 1}^2}\delta_{i}g_i(t){\beta_{i,4}},~k\neq i.
		\end{aligned}\right.
	\end{equation*}
	Solving \eqref{pi_star}) for $ 1\leqslant i\leqslant n $ simultaneously,  we obtain \begin{equation*}
		\pi^*_i(t)=\left[\theta_i\frac{\sum_{k=1}^n\frac{\frac{m}{1-\beta_{k,1}}+\nu\rho h_k(t)}{\delta_{k}g_k(t)}}{n-\sum_{k=1}^n\theta_k}+\frac{\frac{m}{1-\beta_{i,1}}+\nu\rho h_i(t)}{\delta_{i}g_i(t)}\right]\frac{\sqrt{Z(t)}}{\Sigma(t) },\quad \forall1\leqslant i\leqslant n,t\in[0,T].
	\end{equation*}
 Under the compatible conditions (I) and (II), $\left({\varphi}^*_i, {\chi}^*_i, \phi_i^*, \vartheta_i^* \right)\in\mathscr{A}$ and $ \left(\pi_i^*,a_i^*\right)\in\mathscr{U}_i $.
	\bibliographystyle{apalike}
	\bibliography{RNG}
	
\end{document}